\documentclass[letterpaper, 11pt]{article}

\usepackage[left=1in, top=1in, right=1in, bottom=1in]{geometry}

\usepackage{amsmath,amsfonts,bm}
\usepackage{amssymb}

\def\ceil#1{\lceil #1 \rceil}
\def\floor#1{\lfloor #1 \rfloor}
\def\1{\bm{1}}

\def\eps{{\varepsilon}}

\def\vzero{{\bm{0}}}

\def\vvarphi{{\bm{\varphi}}}

\def\vm{{\bm{m}}}

\def\vw{{\bm{w}}}

\def\vz{{\bm{z}}}

\def\gA{{\mathcal{A}}}

\def\gI{{\mathcal{I}}}

\def\gN{{\mathcal{N}}}

\def\gU{{\mathcal{U}}}

\def\sC{{\mathbb{C}}}

\def\sN{{\mathbb{N}}}

\def\sR{{\mathbb{R}}}

\DeclareMathOperator{\im}{Im}
\DeclareMathOperator{\re}{Re}

\usepackage{amsthm}
\theoremstyle{plain}
\newtheorem{thm}{Theorem}[section]
\newtheorem{definition}[thm]{Definition}

\newtheorem{lemma}[thm]{Lemma}

\newtheorem*{remark}{Remark}
\newtheorem{corollary}[thm]{Corollary}

\newtheorem{claim}[thm]{Claim}

\usepackage{algorithm}
\usepackage{algorithmic}

\makeatletter
\def\Ddots{\mathinner{\mkern1mu\raise\p@
\vbox{\kern7\p@\hbox{.}}\mkern2mu
\raise4\p@\hbox{.}\mkern2mu\raise7\p@\hbox{.}\mkern1mu}}
\makeatother

\makeatletter
\newcommand*{\rom}[1]{\expandafter\@slowromancap\romannumeral #1@}
\makeatother

\usepackage[numbers]{natbib}

\usepackage{bbm}
\usepackage{enumitem}
\usepackage{graphicx}

\def\i{{\mathbbmss{i}}}
\def\open#1{#1^{\circ}}

\title{Complex contraction on trees without proof of correlation decay}
\author{
Liang Li\thanks{School of Cyber Science and Technology, Shandong University; email: \texttt{li.liang@sdu.edu.cn}.} \\
\and
Guangzeng Xie\thanks{Academy for Advanced Interdisciplinary Studies, Peking University; email: \texttt{smsxgz@pku.edu.cn}. } \\
}
\date{}
\begin{document}

\maketitle 
\begin{abstract}
    We prove complex contractions for zero-free regions of several counting problems whose partition functions can thus be approximated via Barvinok's algorithmic paradigm\cite{barvinok2016combinatorics}. Although our approach relies on the well-known computation tree expansion technique, we do not need a proof of the correlation decay property over the real axis before getting zero-freeness. Alternatively, we directly look for a convex region in the complex plane which contracts into its interior as the tree expansion procedure recursively goes from leaf to root. For various counting problems which have tree expansions, the contraction regions obtained by our approach do not depend on the degree of the constraint graphs, so that we can prove zero-freeness for unbounded degree cases. 
    
    As consequences of our proof and using Barvinok's paradigm, we can design fully polynomial-time approximation schemes(FPTAS) for bounded degree 2-spin systems. Since our contraction region is degree independent, we can also design quasi-polynomial time approximation algorithms for 2-spin systems and generalized set cover problems.
    Our result can cover or partially cover several previous results obtained via correlation decay or spatial mixing techniques. We can also improve previous results based on contraction arguments\cite{shao2019contraction} and obtain new algorithmic results for 2-spin systems with negative weights. In contrast to previous zero-free results based on the correlation decay method which needs different potential functions for different problems, our approach is more generic in the sense that our contraction region for different problems share common shape in the complex plane.
\end{abstract}
\newpage
\section{Introduction}

The study of the zero-set of partition functions is of fundamental significance in statistical physics and theoretical computer science. Given a particle system in statistical physics, or equivalently a constraint satisfaction problem in computer science, its partition function can be viewed as a polynomial over specific parameters. As the value of the parameters changes continuously, the global property of the whole system may experience a non-continuous abrupt change. In physics terminology, this phenomenon is called the phase transition while in theoretical computer science, such phase transitions are highly related to counting complexity dichotomies. The zero-set of the partition function polynomial reveals the true mathematical principles underlying the study of phase transitions and also has implications on counting algorithm design.

In physics, the partition function, viewed as a polynomial, has intimate connections to the free energy, which is a quantity that includes all global information of the physics system. In their seminal work, Lee and Yang\cite{lee1952statistical, yang1952statistical} showed that zero-freeness of partition functions is equivalent to the analyticity of global physics quantities and thus implies the absence of phase transitions. 

In computer science, the partition function has the meaning of the sum of weighted solutions of counting problems. Existing work\cite{weitz2006counting, sinclair2014approximation,sly2010computational, sly2012computational, li2013correlation} already show that phase transition in physics is highly correlated to the counting dichotomy theorems, where the correlation decay is the key concept to understanding how physics properties imply algorithmic results.

The cornerstone of counting algorithm design by zero-freeness is developed by Barvinok\cite{barvinok2016combinatorics} and Patel and Regts\cite{patel2017deterministic}. The main idea of such a paradigm is that, when the partition function is zero-free, its logarithm can be expanded within a complex region by the Taylor series such that we can omit high-order terms to get an approximation value up to an arbitrary precision. The key point along this line of research is to find a region in the complex plane within which the partition function is zero-free. For counting problems which can be solved via correlation decay, Barvinok's paradigm can also be applied to design new approximation algorithms. The complex regions found in these problems are usually a complex strip around a segment in the real axis within which correlation decay property holds.Therefore, for this kind of problems, the zero-freeness highly depend on proving correlation decay first, which means to design an algorithm via one paradigm, we need to analyze some property used in another algorithmic paradigm. 

In this paper, we develop new analysis method to study the zero-free regions for counting problems that can be calculated by the celebrated computation tree expansion method\cite{weitz2006counting, bayati2007simple}. Our main contribution is to obtain zero-freeness directly without proving the correlation decay first. The main approach is to view the tree recursion as a dynamic procedure over the complex plane and look for a complex region which contracts into its interior as the dynamic procedure goes on. The basic idea is to look into one step recursion and decouple the recursive function into basic holomorphic mappings over the complex plane. According to the geometric properties of basic holomorphic mappings, we can reversely shape the contracted region as we expect. As most tree recursion has similar decouplings in terms of basic holomorphic mappings, we can uniformly obtain zero-free regions with similar geometric shapes. Without proving correlation decay of the tree recursion first, we do not need to design ingenious ``potential functions'' for different problems. 

We mainly study two general classes of problems, the 2-spin systems and the generalized set cover problem. A 2-spin system(see Section \ref{sec:2-spin}) is defined on an ordinary graph and has for each edge an interaction function parameterized by $\beta>0, \gamma\geq 0$ and for each vertex an external activity parameterized by $\lambda\in \sR$. We study both the ferromagnetic and the anti-ferromagnetic cases, even with negative activities. The generalized set cover problem(see Section \ref{sec:set-cover}) can be defined on a hypergraph where the hyper-edges correspond to \textit{elements} and the vertices correspond to \textit{sets}. A vertex $v$ is involved in an hyper-edge $e$ if and only if the set $v$ contains the element $e$. Both the elements/hyper-edges and the sets/vertices can be assigned non-negative real values as their weights. After normalization, the weight of a \textit{set}  is  $\eta (\eta\ge 0)$ when it is not selected and 1 otherwise, while the weight of an \textit{element} is $1+ \mu (\mu\ge -1)$ when it is uncovered and 1 otherwise. For these two classes of problems, a telescoping expansion technique\cite{weitz2006counting, lin2014simple, liu2014fptas-cnf, liu2014fptas-edge} can be used to break cycles, so that we can have tree recursions for calculating marginal probabilities of one node. 


Several classical counting problems can be viewed as special cases of the two general problems. Examples of the 2-spin system include counting independent sets and counting vertex covers, while examples of the generalized set covers problem include counting monotone CNFs whose dual problem is hypergraph independent sets, counting edge covers, and counting bipartite independent sets. The 2-spin systems and generalized set cover problems overlap at the famous Ising model.


\section{Our Results}
We obtain zero-free regions for both 2-spin systems and generalized set cover problem. 

For 2-spin systems:
\begin{itemize}
    \item[1.] We give zero-free regions for ferromagnetic 2-spin systems up to uniqueness on bounded degree graphs, thus implying an FPTAS  which for $\gamma\le 1$ matches the result in \cite{guo2018uniqueness} and for general $\gamma$ value covers and improves the results in \cite{shao2019contraction};
    \item [2.] We give zero-free regions for anti-ferromagnetic 2-spin systems on bounded degree graphs and the implied FPTAS partially covers the results in \cite{li2013correlation};
    \item [3.] We give zero-free regions for both ferromagnetic and anti-ferromagnetic 2-spin systems on unbounded degree graphs, thus implying quasi-polynomial time approximation algorithms.
    \item [4.] The zero-free region we found also includes negative $\lambda$ values for both bounded and unbounded degree cases. So far as we know, this is first obtained in this paper.
\end{itemize}

For generalized set cover problem, we give a uniform zero-free region:
\begin{itemize}
    \item[1.] For set cover problem, we obtain zero-freeness for any $\mu \in [-1, 0]$ and $\eta \ge 0$. This imply a quasi-polynomial time algorithm for edge cover problem while an FPTAS is proposed in \cite{liu2014fptas-edge} using correlation decay;
    \item[2.] For degree up to 5 set cover problem, we obtain zero-freeness for $\mu = 1$ and $\eta < 1.45399$. This imply a quasi-polynomial time algorithm for bipartite independent set problem while an FPTAS is proposed in \cite{liu2015fptas}.
\end{itemize}

\section{Related Work}

The calculation of partition functions has been widely studied in different areas. There are two main paradigms in designing deterministic algorithms for partition functions inspired by statistical physics. One is the correlation decay algorithm\cite{weitz2006counting,bayati2007simple, sinclair2014approximation,li2012approximate,li2013correlation,lin2014simple, liu2014fptas-cnf,liu2014fptas-edge, liu2015fptas} and the other is Barvinok's paradigm based on zero-free properties\cite{barvinok2016combinatorics, patel2017deterministic}.

In the seminal work of Weitz\cite{weitz2006counting} and Bayati, Gamarnik, Katz, Nair and Tetali\cite{bayati2007simple}, computation tree expansion was established for counting independent sets and counting matchings. Following this paradigm, new FPTAS are designed for various new problems, including general 2-state spin systems\cite{li2012approximate, sinclair2014approximation, li2013correlation}, edge covers\cite{lin2014simple, liu2014fptas-edge}, monotone CNFs\cite{liu2014fptas-cnf}, hypergraph matchings\cite{song2019counting} and hypergraph indepent sets\cite{bezakova2019approximation}. 

Barvinok's algorithmic paradigm is based on the seminal work by Lee and Yang\cite{yang1952statistical, lee1952statistical}. 
There are various strategies in proving zero-freeness, examples include Asano contraction\cite{liu2019ising,guo2019zeros, guo2020zeros} and direct analysis of contraction region for the hardcore model\cite{peters2019conjecture, bencs2018note}. There is also results which shows correlation decay implies zero-freeness\cite{liu2019fisher,liu2019approximate}.

\section{Preliminary}
In this paper, for a subset $A$ of the complex plane $\sC$, we denote $\open{A}$ to be interior of the set $A$, $\overline{A}$ to be the closure of the set $A$, and $\partial A$ to be the boundary of the set $A$ which is defined as $\overline{A} \backslash \open{A}$. For a complex number $z = x + \i y$, we denote its real part by $\re(z) = x$, its imaginary part by $\im(z) = y$. Additionally, the modulus of $z$ is defined as $|z| = \sqrt{x^2 + y^2}$ and the argument of $z$ is defined as $\arg(z) = \arcsin\left( y/\sqrt{x^2 + y^2} \right) \in (-\pi, \pi]$ if $z \neq 0$. 
It is clearly that if $\re(z) > 0$ then $\arg(z) = \arctan\left( y/x \right)$.

We choose a branch of the function $\ln(z)$ such that $\ln(z) = \ln(|z|) + \i \arg(z)$ which is holomorphic on $\sC \backslash (-\infty, 0]$.

We denote the extended complex plane by $\hat{\sC} = \sC \cup \{\infty\}$. Note that the Mobius transform $T(z) = \frac{a z + b}{c z + d}$ is a bijection from $\hat\sC$ to $\hat\sC$ for $a d - b c \neq 0$. 


For a set $U \subseteq \sC$ and $\delta > 0$, define $\gU(U, \delta)$ as
\begin{align*}
    \gU(U, \delta) \triangleq \{z \in \sC: \exists z' \in U \text{ such that } |z - z'| < \delta\}.
\end{align*}

The following lemma is based on invariance of domain. It reveals that an injective and holomorphic function will map the boundary to boundary. 
\begin{lemma}\label{lem:boundary}
    Suppose that a function $\varphi$ is injective and holomorphic on $\overline{U}$ where $U$ is a domain of the complex plane. Then we have
    \begin{align*}
        \partial \varphi(U) = \varphi(\partial U).
    \end{align*}
\end{lemma}

\subsection{2-Spin Systems}\label{sec:2-spin}
Let $A= \left[\begin{array}{ll}{A_{0,0}} & {A_{0,1}} \\ {A_{1,0}} & {A_{1,1}}\end{array}\right] = \left[\begin{array}{ll}{\beta} & {1} \\ {1} & {\gamma}\end{array}\right]$ and $b_0 = 1, b_1 = \lambda$. 
The partition function of 2-spin systems on graph $G = (V, E)$ is given by
\begin{align*}
    Z_{\tau_{S}}^{\beta, \gamma}(G, \lambda) = \sum_{\substack{\sigma \in \{0, 1\}^V \\ \sigma(v) = \tau_{S}(v), \text{ for } v \in S}} \prod_{(u, v) \in E} A_{\sigma(u), \sigma(v)} \prod_{v \in V} b_{\sigma(v)},
\end{align*}
where $S \subseteq V$ and $\tau_{S} \in \{0,1\}^{S}$. 

The following definition of feasible configurations is introduced in \cite{shao2019contraction}. 
\begin{definition}[Feasible configuration of 2-spin systems]
    \label{def:2-spin:feasible}
    Given a graph $G = (V, E)$ of the 2-spin system specified by $\beta, \gamma, \lambda ~(\lambda \neq 0)$, a configuration $\tau_S$ on some vertices $S \subseteq V$ is feasible if
    \begin{enumerate}
        \item there is no edge $e = (u, v) \subseteq S$ such that $\tau_S(u) = \tau_S(v) = 0$ if $\beta = 0$,
        \item there is no edge $e = (u, v) \subseteq S$ such that $\tau_S(u) = \tau_S(v) = 1$ if $\gamma = 0$,
    \end{enumerate}
\end{definition}
It is clearly that if $\tau_S$ is not feasible then $Z_{\tau_{S}}^{\beta, \gamma}(G, \lambda) = 0$. 

Next we define $\hat{x}_d$ which will be useful in analysis of anti-ferromagnetic 2-spin systems with positive activities $\lambda$. 
For $\sqrt{\beta\gamma} \le \frac{d-1}{d+1}, \beta > 0, \gamma \ge 0$, define 
\begin{align}
    \hat{x}_d \triangleq \begin{cases}
        \frac{-1 - \beta\gamma + d(1 - \beta\gamma) - \sqrt{(-1 - \beta\gamma + d (1 - \beta\gamma))^2 - 4\beta\gamma}}{2 \gamma}, ~~~&\text{ for } \gamma > 0, \\
        \frac{\beta}{d - 1}, ~~~&\text{ for } \gamma = 0,
    \end{cases}
\end{align}
and denote 
$
    \lambda_c(d) \triangleq \hat{x}_d \left(\frac{\hat{x}_d + \beta}{\gamma \hat{x}_d + 1}\right)^d.
$
Following from Lemma \ref{lem:lambda-c}, when $\beta \le 1$, $\lambda_c(d)$ is strictly decreasing on $(\bar{d}, +\infty)$, and when $\beta > 1$, there exists $d_c$ such that $\lambda_c(d)$ is strictly deceasing on $(\bar{d}, d_c)$ and is strictly increasing on $(d_c, +\infty)$, where $\bar{d} = \frac{1 + \sqrt{\beta\gamma}}{1 - \sqrt{\beta\gamma}}$.

Further, for analysis of anti-ferromagnetic 2-spin systems with negative activities $\lambda$, we will need the following definition of $\check{x}_d$. For $\beta\gamma < 1, \beta > 0, \gamma \ge 0$, define
\begin{align}
    \check{x}_d \triangleq \frac{1 - \beta \gamma + d(1 + \beta \gamma) +\sqrt{(1 - \beta \gamma + d(1 + \beta \gamma))^2-4 \beta  \gamma  d^2}}{2 \beta  d}.
\end{align}
Detailed discussions about $\hat{x}_d$ and $\check{x}_d$ can be found in Appendix \ref{app:hat-x} and \ref{app:check-x}.

\subsection{Generalized Set Cover}\label{sec:set-cover}
A generalized 2-spin systems defined on a hypergraph $G = (V, E)$ is specified by hyperedge activities $\varphi_e: \{0, 1\}^{|e|} \to \sC$ for $e \in E$, and a vertex activity $\eta$. Its partition function is defined as 
\begin{align*}
   Z_{\tau_S}(G, \vvarphi, \eta) \triangleq \sum_{\substack{\sigma \in \{0, 1\}^{V} \\ \sigma(v) = \tau_S(v) \text{ for } v \in S}} \prod_{e \in E} \varphi_e(\sigma\vert_{e}) \eta^{|\{v: \sigma(v) = 0\}|},
\end{align*}
where $S \subset V$ and $\tau_S \in \{0, 1\}^S$. 

In this paper, we consider set cover problems with $\varphi_e$ defined as 
\begin{align*}
    \varphi^{\mu}_e(\hat{\sigma}) = 1 + \mu \prod_{v \in e} (1 - \hat{\sigma}(v)), ~~~\mu \ge -1,
\end{align*}
and denote the partition function $Z(G, \vvarphi^{\mu}, \eta)$ as $Z(G, \mu, \eta)$ for simplicity.

Several classical counting problems can be viewed as special cases of generalized set cover problem, e.g., $\sharp$monotone-CNF, $\sharp$BIS and edge covers. See Appendix \ref{app:set-cover-problem} for more details.

We also need definition of feasible configurations for set covers. 
\begin{definition}[Feasible configuration of set cover]\label{def:set-cover:feasible}
    We say a configuration $\tau_S \in \{0, 1\}^{S}$ is feasible if $\tau_S$ does not assign any vertices in an edge in $G$ both to $0$ for $\mu = -1$, that is
    \begin{align*}
        \forall e \in E, ~ \text{ if } e \subseteq S, \text{ then } \exists v \in e, \tau_S(v) = 1, \text{ for } \mu = -1.
    \end{align*}
\end{definition}
It is clearly that if $\tau_S$ is infeasible then $Z_{\tau_{S}}(G, \mu, \eta) = 0$.

\section{2-Spin Systems with Bounded Degree}
In this section, we study zero-free regions for 2-spin systems on bounded degree graphs.

\begin{thm}\label{thm:bounded-2-spin}
    For $\beta > 0, \gamma \ge 0$, suppose that $\beta, \gamma, \lambda_0$ satisfies one of the following condition:
    \begin{enumerate}
        \item $\beta\gamma > 1$ and $0 < \lambda_0 < \left(\frac{\beta}{\gamma}\right)^{\frac{\sqrt{\beta\gamma}}{\sqrt{\beta\gamma} - 1}}\max\{1, \gamma\}^{\frac{2 \sqrt{\beta\gamma} }{\sqrt{\beta\gamma} - 1} - \Delta}$, 
        \item $\beta\gamma > 1, 0 > \lambda_0 > -\max\{1, \gamma\}^{\Delta}$, 
        \item $\sqrt{\beta\gamma} \le \frac{\Delta-2}{\Delta}, \lambda_0 > 0$ and
        \begin{align*}
            \lambda_0 < \begin{cases}
                \lambda_c(d_c), ~~~&\text{ for } \beta > 1, \Delta > d_c + 1, \\
                \lambda_c(\Delta-1), ~~~&\text{ otherwise, }
            \end{cases}
        \end{align*}
        \item $\beta\gamma < 1, \lambda_0 < 0$ and
        \begin{align*}
            \lambda_0 > \begin{cases}
                -\min\left\{1, \frac{\beta - 1}{1 - \gamma}\right\}, ~~~&\text{ for } \beta > 1, \Delta > \frac{1 - \beta\gamma}{(\beta - 1)(1 - \gamma)} + 1, \\
                -\frac{\beta \check{x}_d - 1}{(\check{x}_d - \gamma) \check{x}_d^{d}} ~~ (d = \Delta - 1) ~~~&\text{ otherwise. }
            \end{cases}
        \end{align*}
    \end{enumerate}
    Then there exists $\delta > 0$, which depends on $\Delta, \beta, \gamma, \lambda_0$, such that $Z_{\tau_S}^{\beta, \gamma}(G, \lambda) \neq 0$ for all $\lambda \in \gU([0, \lambda_0], \delta)~ (\lambda \neq 0)$ and all graphs $G$ with maximum degree $\Delta$ and all feasible configurations $\tau_S$ (cf. Definition \ref{def:2-spin:feasible}).
\end{thm}

Based on Theorem \ref{thm:bounded-2-spin}, it is easy to check that $\lambda^{-m} Z_{\tau_S}^{\beta, \gamma}(G, \lambda) \neq 0$ for $m = |v \in S: \tau(v) = 1|$ and all $\lambda \in \gU([0, \lambda_0], \delta)$. 
Then following from Barvinok’s paradigm\cite{barvinok2016combinatorics, patel2017deterministic}, there exists an FPTAS for computing $Z_{\tau_S}^{\beta, \gamma}(G, \lambda_0)$ on graphs with maximum degree $\Delta$.

When $\beta = 1$ and $\gamma = 0$, Theorem \ref{thm:bounded-2-spin} provides FPTAS for $-\frac{(\Delta - 1)^{\Delta -1}}{\Delta^{\Delta}} < \lambda_0 < \frac{(\Delta - 1)^{\Delta -1}}{(\Delta - 2)^{\Delta}}$, which covers a well known result for the hard-core model. 
Condition 1 with $\gamma \le 1$ implies an FPTAS which covers the result in \cite{guo2018uniqueness} and improves the results in \cite{shao2019contraction} for general $\gamma$. When $\beta \le 1$ or $\Delta \le d_c + 1$, Condition 3 implies an FPTAS that covers the results in \cite{li2013correlation}, and when $\beta > 1$ and $\Delta > d_c + 1$, \cite{li2013correlation} provided an FPTAS for $\lambda_0 < \min\{\lambda_c(\floor{d_c}), \lambda_c(\ceil{d_c})\}$ which is slightly stronger than our result. The results for negative activities $\lambda_0$, cf. Condition 2 and Condition 4, are first obtained in this paper. 

\begin{remark}
    Note that $Z^{\beta, \gamma}_{\tau_S}(G, \lambda) = \lambda^{|V|} Z^{\gamma, \beta}_{\tilde\tau_S}(G, 1/\lambda)$ where $\tilde\tau_S(v) = 1 - \tau_S(v)$ for $v \in S$. For $\lambda > 0$, if there exists an FPTAS for computing $Z_{\tau_S}^{\beta, \gamma}(G, \lambda)$ for all $\lambda < \lambda_c(\beta, \gamma)$, then there also exists an FPTAS for computing $Z_{\tau_S}^{\beta, \gamma}(G, \lambda)$ for all $\lambda > 1/\lambda_c(\gamma, \beta)$. Similarly, for $\lambda < 0$, if there exists an FPTAS for computing $Z_{\tau_S}^{\beta, \gamma}(G, \lambda)$ for all $\lambda > \tilde\lambda_c(\beta, \gamma)$, then there also exists an FPTAS for computing $Z_{\tau_S}^{\beta, \gamma}(G, \lambda)$ for all $\lambda < 1/\tilde\lambda_c(\gamma, \beta)$.
\end{remark}

\subsection{Computation Tree}
In this section, we recall the tree recursion of 2-spin systems \cite{weitz2006counting}. 
Given a graph $G = (V, E)$ and a feasible configuration $\tau_S$, we define the marginal ratios as 
\begin{align*}
    R^{\beta, \gamma, \lambda}_{\tau_{S}}(G, v) \triangleq
    \begin{cases}
    \frac{Z^{\beta, \gamma}_{\tau_{S}, \sigma(v)=1}(G, \lambda)}{Z^{\beta, \gamma}_{\tau_{S}, \sigma(v)=0}(G, \lambda)},~~~&\text{ for } v \in V \backslash S, \\
    0, ~~~&\text{ for } \tau_{S}(v) = 0, \\
    \infty, ~~~&\text{ for } \tau_{S}(v) = 1.
    \end{cases}
\end{align*}
\begin{remark}
    We will show in the proof of Lemma \ref{lem:bounded-2-spin:contraction-zero-free} that $R^{\beta, \gamma, \lambda}_{\tau_{S}}(G, v)$ is well-defined, i.e., either $Z^{\beta, \gamma}_{\tau_{S}, \sigma(v)=1}(\lambda) \neq 0$ or $Z^{\beta, \gamma}_{\tau_{S}, \sigma(v)=0}(\lambda) \neq 0$, when the complex contraction holds.  
\end{remark}


Before building the tree recursion, it worth noting that for $\beta = 0$ and $\tau_S(u) = 0$ we can pin all neighbors of $u$ to $1$ without changing any values of partition functions and marginal ratios.
In fact, let $\tilde{V} = \{v \in V \backslash S: v \sim u\}$, and $\sigma_{\tilde{V}}$ such that $\sigma_{\tilde{V}}(v) = 1$ for $v \in \tilde{V}$. 
Following from $Z^{\beta,\gamma}_{\tau_S, \sigma(v)=0}(G, \lambda) = 0$ for $v \sim u$, we have
$Z^{\beta,\gamma}_{\tau_S}(G, \lambda) = Z^{\beta,\gamma}_{\tau_S \cup \sigma_{\tilde{V}}}(G, \lambda)$.
Similarly, for $\gamma = 0$ and $\tau_S(u) = 1$, we can pin all neighbors of $u$ to $0$. 
\begin{lemma}\label{lem:bounded-2-spin:computation-tree}
    Given a graph $G = (V, E)$ and a feasible configuration $\tau_S$, consider a free vertex $v \in V \backslash S$. Suppose that the neighbors of $v$ are $v_1, \dots, v_d$. Replace $v$ in each edge $(v, v_i)$ with an independent duplicate $\tilde{v}_i$ for $i = 1, \dots, d$, and denote this new graph by $\tilde{G}$. Then we have
    \begin{align*}
        R^{\beta, \gamma, \lambda}_{\tau_{S}}(G, v) = f^{\beta,\gamma}_{\lambda, d}(R^{\beta, \gamma, \lambda}_{\tau_{S}\cup \sigma_1}(G_1, v_1), \dots, R^{\beta, \gamma, \lambda}_{\tau_{S}\cup \sigma_d}(G_d, v_d)),
    \end{align*}
    where 
    \begin{align*}
        f^{\beta,\gamma}_{\lambda, d}(z_1, z_2, \cdots, z_d) \triangleq \lambda \prod_{i=1}^d \left(\frac{\gamma z_i + 1}{z_i + \beta}\right),
    \end{align*}
    and for $i \in [d]$, $G_i = \tilde{G} - \tilde{v}_i$ and $\sigma_i$ satisfies $\sigma_i(\tilde{v}_j) = 0$ if $j < i$, and $\sigma_i(\tilde{v}_j) = 1$ if $j > i$.

\end{lemma}
We point out that configurations $\{\tau_S \cup \sigma_i\}_{i=1}^d$ are also feasible. For example, when $\gamma = 0$, all neighbors of $u$ with $\tau_S(u) = 1$ are pinned to $0$, so there is no neighbor of a free vertex $v$ pinned to $1$ by $\tau_S$. 

To simplify notations, we will use $f_{\lambda, d}$ for $f^{\beta, \gamma}_{\lambda, d}$ when context is clear.

\subsection{Complex contraction implies absence of zeros}
We here introduce a definition of complex contraction different from the definition in \cite{shao2019contraction}.

\begin{definition}\label{def:2-spin:complex-contraction}
    We say $\lambda \in \sC$ satisfies $\Delta$-complex-contraction property for 2-spin systems with parameters $\beta, \gamma ~(|\beta| + |\gamma| > 0)$ if there is an close region $F_{\lambda} \subseteq \hat{\sC}$ such that 
    \begin{enumerate}
        \item $-1 \notin F_{\lambda}$, $0 \in F_{\lambda}$ if $\beta \neq 0$, and $\infty \in F_{\lambda}$ if $\gamma \neq 0$,
        \item $f_{\lambda, \Delta}(z_1, \dots, z_{\Delta}) \neq -1$ for $z_i \in F_{\lambda}, i = 1, \dots, \Delta$,
        \item $f_{\lambda, d}(z_1, \dots, z_{d}) \in F_{\lambda}$ for any $d = 0, \dots, \Delta-1$ and $z_i \in F_{\lambda}, i = 1, \dots, d$.
    \end{enumerate}
\end{definition}

Similar to \cite{bencs2018note, shao2019contraction}, complex contraction so defined implies zero-freeness of 2-spin systems.
\begin{lemma}\label{lem:bounded-2-spin:contraction-zero-free}
    If $\lambda \neq 0$ satisfies $\Delta$-complex-contraction property for 2-spin systems with parameters $\beta, \gamma$, then $Z^{\beta, \gamma}_{\tau_{S}}(G, \lambda) \neq 0$ for any graphs $G = (V, E)$ with feasible configuration $\tau_S$ and maximum degree of free vertex $v \in V\backslash S$ is at most $\Delta$.
\end{lemma}

\subsection{Convert the Product of Regions to the Sum of Regions}
For hard-core models $(\beta = 1, \gamma = 0)$, complex contraction without proving correlation decay could not find a tight zero-free region around the real axis \cite{peters2019conjecture, bencs2018note}. Other than semicircles in \cite{peters2019conjecture} and circular sectors in \cite{bencs2018note}, we choose isosceles triangles, isosceles trapezoids, or rectangles as complex contraction region after employing a holomorphic mapping $\phi$ which allows us to deal with the sum of same regions instead of the product of same regions. For $\phi(z) = \ln\left( \frac{\gamma z + 1}{z + \beta} \right)$, we have $\phi^{-1}(w) = \frac{\beta e^w - 1}{\gamma - e^w}$ and 
\begin{align}
    f^{\phi}_{\lambda, d}(w_1, \dots, w_d)\triangleq \phi(f_{\lambda, d}(\phi^{-1}(w_1), \dots, \phi^{-1}(w_d))) = \ln\left(\frac{\gamma \lambda \exp(\sum_{i=1}^d w_i) + 1}{\lambda \exp(\sum_{i=1}^d w_i) + \beta}\right).    
\end{align}

\begin{remark}
    For real $\beta > 0, \gamma \ge 0$ and $\beta\gamma < 1$, it is easy to check that $\phi$ is holomorphic and bijective from $\hat{\sC} \backslash [-1/\gamma, -\beta]$ (or $\sC \backslash (-\infty, -\beta]$ for $\gamma = 0$) to $\{w \in \sC: |\im w| < \pi\}$ and $\phi^{-1}(w) = \frac{\beta e^w - 1}{\gamma - e^w}$ is indeed the inverse function of $\phi$. Similar results also hold for $\beta\gamma > 1$ with $\hat{\sC} \backslash [-\beta, -1/\gamma]$.
\end{remark}
The form of $f^{\phi}_{\lambda, d}$ inspires us to find a convex region $U_{\lambda} \subseteq \{w \in \sC: |\im{w}| < \pi \}$ for proving complex contraction. As long as $U_{\lambda}$ is convex, we have $\bar{w} = \frac{1}{d} \sum_{i=1}^d w_i \in U_{\lambda}$ for all $w_i \in U_{\lambda}, \ i = 1, \dots, d$, and
\begin{align*}
    f^{\phi}_{\lambda, d} (w_1, \dots, w_d) = \phi(\lambda e^{d \bar{w}}).
\end{align*}
Our strategy is to first find a convex and compact region $U$ such that complex contraction holds for all real $\lambda \in [0, \lambda_0]$ in the region $\phi^{-1}(U)$. Then uniformly continuity of $\phi$ will result in complex contraction for all complex values $\lambda \in \gU([0, \lambda_0], \delta)$ in the same region $\phi^{-1}(U)$. 
\begin{lemma}\label{lem:bounded-2-spin:U}
    Fix $\lambda_0 \neq 0, \lambda_0 \in \sR, \beta > 0, \gamma \ge 0$. 
    If there exists a convex and compact region $U \subseteq \{w \in \sC: |\im{w}| < \pi \}$ such that 
    \begin{enumerate}
        \item $-1 \notin \phi^{-1}(U)$, $-\ln(\beta) \in U$, and $\ln(\gamma) \in U$ if $\gamma > 0$,
        \item $\lambda e^{\Delta w} \neq -1$ for $w \in U$ and $\lambda \in [0, \lambda_0]$,
        \item $\lambda e^{d w} \notin (-\infty, \max\{-\beta, -1/\gamma\}]$ for $d = 0, 1, \dots, \Delta - 1$, $w \in U$ and $\lambda \in [0, \lambda_0]$, 
        \item $\phi(\lambda e^{d w}) \in \open{U}$ for $d = 0, 1, \dots, \Delta - 1$, $w \in U$ and $\lambda \in [0, \lambda_0]$,
    \end{enumerate}
    then there exists $\delta > 0$, which depends on $\lambda_0, \beta, \gamma, \Delta$, such that all $\lambda \in \gU([0, \lambda_0], \delta)$ satisfies $\Delta$-complex-contraction property for 2-spin systems with parameters $\beta, \gamma$.
\end{lemma}


\subsection{Complex Contraction for Real $\lambda$}
Following from Lemma \ref{lem:bounded-2-spin:U}, we just need to verify complex contraction for real $\lambda \in [0, \lambda_0]$. 
We choose the following class of isosceles triangles for finding $U$. Define
\begin{align}\label{def:U}
    U(x_0, x_1, k) = \{z: \re(z) \in [x_1, x_0], |\im(z)| \le k (x_0 - \re(z)) \},
\end{align}
where 
$x_1 \le 0 \le x_0, x_1 < x_0$, and $0 < k < \frac{\pi}{2 \Delta(x_0 - x_1)}$.

\begin{figure}[h]
\centering
\includegraphics[width=\textwidth]{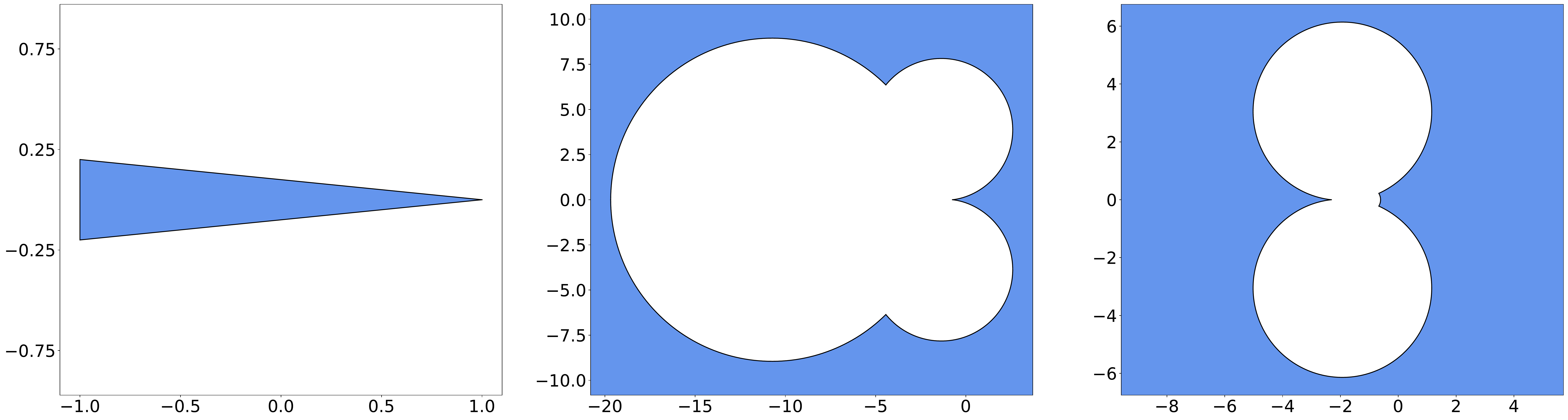}
\caption{Left: the region of $U(1, -1, 0.1)$. Middle: the region of $\phi^{-1}(U(1, -1, 0.1))$ with parameters $\beta = 2, \gamma = 0.3$. Right: the region of $\phi^{-1}(U(1, -1, 0.1))$ with parameters $\beta = 4, \gamma = 0.6$.}
\end{figure}

The key step is to choose $x_0, x_1, k$ so that $\phi(\lambda e^{d w}) \in \open{U(x_0, x_1, k)}$, i.e., 
\begin{align}\label{eq:contain}
    \re\left(\phi(\lambda e^{d w})\right) \in (x_1, x_0) ~~~\text{ and }~~~
    k \left(x_0 - \re\left(\phi(\lambda e^{d w})\right)\right) - \left|\im\left(\phi(\lambda e^{d w})\right)\right| > 0,
\end{align}
where $0 \le d \le \Delta - 1$ and $w \in U(x_0, x_1, k)$. 

The following three observations simplify the difficulty of verifying Equation (\ref{eq:contain}).
\begin{enumerate}
    \item With $0 \in U(x_0, x_1, k)$, we know that $\frac{d}{\Delta - 1} w \in U(x_0, x_1, k)$ for $0 \le d \le \Delta - 1$ and $w \in U(x_0, x_1, k)$. Thus we only need to check the Equation (\ref{eq:contain}) for $d = \Delta - 1$. 
    \item The condition $k < \frac{\pi}{2 \Delta (x_0 - x_1)}$ indicates that $\phi(\lambda e^{d w})$ is an injection with respect to $w$. 
    Then following from boundedness of $\phi(\lambda e^{d w})$ and Lemma \ref{lem:boundary}, 
    which is based on invariance of domain, 
    it is enough to confirm the Equation (\ref{eq:contain}) for $w \in \partial (U(x_0, x_1, k))$, especially for $w \in \{z: \re(z) \in [x_1, x_0], |\im(z)| = k (x_0 - \re(z)) \}$ (cf. Lemma \ref{lem:bounded-2-spin:tilde-G}).
    \item Define
    \begin{align*}
        G(x, k, \lambda) &\triangleq k \left(x_0 - \re\left(\phi(\lambda e^{(\Delta - 1) w})\right)\right) - \left|\im\left(\phi(\lambda e^{(\Delta - 1) w})\right)\right|,
    \end{align*}
    where $w = x + \i k(x_0 - x)$. 
    Then by uniformly continuity of $G$, with $G(x, 0, \lambda) = 0$, the existence of $k > 0$ such that $G(x, k, \lambda) > 0$ can be guaranteed by $H(x, \lambda) > 0$, where 
    \begin{align*}
        H(x, \lambda) \triangleq \frac{\partial G(x, k, \lambda)}{\partial k} \Bigg\vert_{k = 0}.
    \end{align*}
\end{enumerate}

\begin{lemma}\label{lem:bounded-2-spin:H}
    Suppose that $H(x, \lambda) > 0$
    holds for all $x \in [x_1, x_0]$ and $\lambda \in [0, \lambda_0]$, then there exists $k > 0$ such that $G(x, k, \lambda) > 0$ for all $x \in [x_1, x_0]$ and $\lambda \in [0, \lambda_0]$. 
\end{lemma}

Now our strategy for finding $x_0, x_1, k$ is clear. We first choose $x_1$ to be a lower bound of $\re(\phi(\lambda e^{(\Delta - 1) w}))$. Then we determine $x_0$ so that the minimal value of $H(x, \lambda_0)$ on $[x_1, x_0]$ is positive. We finally obtain the following conclusion.
\begin{lemma}\label{lem:bounded-2-spin:x0-x1}
    For $\beta > 0, \gamma \ge 0$, suppose that $\beta, \gamma, \lambda_0$ satisfies one of the conditions in Theorem \ref{thm:bounded-2-spin}. Then there exists $x_0, x_1, k$ such that $\phi(\lambda e^{(\Delta - 1) w}) \in \open{U(x_0, x_1, k)}$ for all $w \in U(x_0, x_1, k)$. Furthermore, we can construct $U$ based on $U(x_0, x_1, k)$ to satisfy all conditions in Lemma \ref{lem:bounded-2-spin:U}. 
\end{lemma}
The values of $x_0, x_1$ in each case can be found in Appendix \ref{app:2-spin:x0-x1}.

\subsection{A Supplementary Result}
For $\frac{\Delta - 2}{\Delta} < \sqrt{\beta\gamma} < \frac{\Delta}{\Delta -2}$, with rectangles instead of isosceles triangles, we can prove following result. 

\begin{thm}\label{thm:bounded-2-spin:supp}
    Suppose that $\frac{\Delta - 2}{\Delta} < \sqrt{\beta\gamma} < \frac{\Delta}{\Delta -2}$, $\beta\gamma \neq 1$ and $\lambda_0 > 0$. 
    Then there exists $\delta > 0$, which depends on $\Delta, \beta, \gamma, \lambda_0$, such that $Z_{\tau_S}^{\beta, \gamma}(G, \lambda) \neq 0$ for all $\lambda \in \gU([0, \lambda_0], \delta)~ (\lambda \neq 0)$ with graphs G of maximum degree at most $\Delta$.
\end{thm}



\section{2-Spin Systems with unbounded degree}
In Theorem \ref{thm:bounded-2-spin}, for $\beta\gamma > 1$ and $\gamma \le 1$, the zero-freeness holds for $\lambda_0 \in \left(-1, \left(\frac{\beta}{\gamma}\right)^{\frac{\sqrt{\beta\gamma}}{\sqrt{\beta\gamma} - 1}}\right)$, which is independent of $d$. A natural question arises: whether there exist $\delta$, which only depends on $\beta, \gamma, \lambda_0$, such that $Z^{\beta, \gamma}_{\tau_S}(G, \lambda) \neq 0$ for arbitrary graph $G$. 

\begin{thm}\label{thm:unbounded-2-spin}
    For $\beta > 0, \gamma \ge 0, \lambda_0 \neq 0$, suppose that $\beta, \gamma, \lambda_0$ satisfies one of the following condition:
    \begin{enumerate}
        \item $\beta\gamma > 1$, $\gamma \le 1$ and $-1 < \lambda_0 < \left(\frac{\beta}{\gamma}\right)^{\frac{\sqrt{\beta\gamma}}{\sqrt{\beta\gamma} - 1}}$, 
        \item $\beta\gamma < 1$, $\beta > 1$ and $\max\left\{-1, -\frac{\beta - 1}{1 - \gamma}\right\} < \lambda_0 < \lambda_c(d_c)$,
    \end{enumerate}
    Then there exists $\delta > 0$, which depends on $\beta, \gamma, \lambda_0$, such that $Z_{\tau_S}^{\beta, \gamma}(G, \lambda) \neq 0$ for all $\lambda \in \gU([0, \lambda_0], \delta)~ (\lambda \neq 0)$ with arbitrary graph $G$ and feasible configuration $\tau_S$.
\end{thm}

Based on Barvinok’s paradigm \cite{barvinok2016combinatorics} and Theorem \ref{thm:unbounded-2-spin}, there exists a quasi-polynomial algorithm for computing $Z_{\tau_S}^{\beta, \gamma}(G, \lambda_0)$ with arbitrary graphs. 

Complex contraction also can be employed to prove Theorem \ref{thm:unbounded-2-spin}. 
\begin{definition}\label{def:unbounded-2-spin:complex-contraction}
    We say $\lambda \in \sC$ satisfies complex-contraction property for 2-spin systems with unbounded degree and parameters $\beta, \gamma ~(|\beta| + |\gamma| > 0)$ if there is an close region $F_{\lambda} \subseteq \hat{\sC}$ such that 
    \begin{enumerate}
        \item $-1 \notin F_{\lambda}$, $0 \in F_{\lambda}$ if $\beta \neq 0$, and $\infty \in F_{\lambda}$ if $\gamma \neq 0$,
        \item $f_{\lambda, d}(z_1, \dots, z_{d}) \in F_{\lambda}$ for any $d \ge 0$ and $z_i \in F_{\lambda}, i = 1, \dots, d$.
    \end{enumerate}
\end{definition}

Similar to Lemma \ref{lem:bounded-2-spin:contraction-zero-free}, complex contraction also implies zero-freeness of 2-spin systems with unbounded degree.
\begin{lemma}\label{lem:unbounded-2-spin:contraction-zero-free}
    If $\lambda \neq 0$ satisfies complex-contraction property for 2-spin systems with unbounded degree and parameters $\beta, \gamma$, then $Z^{\beta, \gamma}_{\tau_{S}}(G, \lambda) \neq 0$ for any graphs $G$ with feasible configuration $\tau_S$. 
\end{lemma}

Furthermore, similar to Lemma \ref{lem:bounded-2-spin:U}, we can prove the following result.
\begin{lemma}\label{lem:unbounded-2-spin:U}
    Fix $\lambda_0 \neq 0, \lambda_0 \in \sR, \beta > 0, \gamma \ge 0$. 
    If there exists a convex and compact region $U \subseteq \{w \in \sC: |\im{w}| < \pi \}$ such that 
    \begin{enumerate}
        \item $-1 \notin \phi^{-1}(U)$, $0 \in U$, $-\ln(\beta) \in U$, and $\ln(\gamma) \in U$ if $\gamma > 0$, 
        \item there exists $\tilde{U}$ such that $d U \subseteq \tilde{U}$ for all $d \ge 0$,
        \item $\lambda e^{w} \notin (-\infty, \max\{-\beta, -1/\gamma, -1\}]$ for $w \in \tilde{U}$ and $\lambda \in [0, \lambda_0]$, 
        \item $\phi(\lambda e^{w}) \in \open{U}$ for $w \in \tilde{U}$ and $\lambda \in [0, \lambda_0]$,
    \end{enumerate}
    then there exists $\delta > 0$, which depends on $\lambda_0, \beta, \gamma$, such that all $\lambda \in \gU([0, \lambda_0], \delta)$ satisfies complex-contraction property for 2-spin systems with unbounded degree and parameters $\beta, \gamma$.
\end{lemma}
We take $U(0, x_1, k)$ as candidates for Lemma \ref{lem:unbounded-2-spin:U}. 
A key observation is that 
\begin{align*}
    \sum_{i=1}^d z_i \in \tilde{U}(k) \triangleq \{z: \re(z) \le 0, |\im(z)| \le -k \re(z) \},
\end{align*}
for all $z_i \in U(0, x_1, k), i =1, \dots, d$ and $d \ge 0$.

To ensure that $\phi(\lambda e^w) \in \open{U(x)}$, we divide $\tilde{U}(k)$ into two parts. For $z \in \tilde{U}(k)$ with $\re(z) \in [x_2, 0]$, the proof is completely same as bounded degree cases. 
On the other hand, for sufficiently small $x_2$ and $z \in \tilde{U}(k)$ with $\re(z) < x_2$, $e^{w}$ is very close to $0$ intuitively which results in that $\re(\phi(\lambda e^{w})) \approx -\ln(\beta)$ and $\im(\phi(\lambda e^{w})) \approx 0$. The formal results are as follows.


\begin{lemma}\label{lem:unbounded-2-spin:x1-k}
    For $\beta > 0, \gamma \ge 0$, suppose that $\beta, \gamma, \lambda_0$ satisfies one of the conditions in Theorem \ref{thm:unbounded-2-spin}. Then there exists $x_1, k$ such that $\phi(\lambda e^{w}) \in \open{U(0, x_1, k)}$ for all $w \in \tilde{U}(k)$. Furthermore, $U(0, x_1, k)$ satisfies all conditions in Lemma \ref{lem:unbounded-2-spin:U}. 
\end{lemma}
\section{Generalized Set Covers}
In this section, we prove zero-free regions for generalized set covers. 
\begin{thm}\label{thm:set-cover}
    Denote $\mu_1 = \frac{e^{1 + \frac{1}{\Delta - 1}}}{\Delta - 1}, \mu_2 = -\frac{e^{1 - \frac{1}{\Delta - 1}}}{\Delta - 1}$ and $\psi(y) = \frac{\ln(1 + y)}{y}$.

    For $\eta_0 > 0$, if one of following condition holds
    \begin{enumerate}
        \item $\mu \in \left[ \mu_2, \mu_1 \right]$, 
        \item $\mu > \mu_1$, and $\eta_0 < \eta_1(\Delta, \mu)$, where
        \[\frac{1}{\eta_1(\Delta, \mu)} = \psi^{-1}\left(\frac{x_1}{1 + \mu e^{-x_1}}\right) \left(1 + \mu e^{-x_1}\right)^{\Delta - 1},\]
        and $x_1$ is the smallest solution to the equation $\frac{\mu x}{\mu + e^x} = \frac{\ln((\Delta - 1) \mu)}{\Delta}$,
        \item $\mu < \mu_2$, and $\eta_0 < \eta_2(\Delta, \mu)$, where 
        \[\frac{1}{\eta_2(\Delta, \mu)} = \psi^{-1}\left(2 - \frac{x_2}{1 + \mu e^{-x_2}}\right) \left(1 + \mu e^{-x_2}\right)^{\Delta - 1},\]
        and $x_2$ is the largest solution to the equation: 
        \[\ln\left(\frac{- (\Delta - 1) x \mu e^{-x}}{2 - x + 2 \mu e^{-x}}\right) \frac{1 + \mu e^{-x}}{x - 2 - 2 \mu e^{-x} - (\Delta - 1) x \mu e^{-x}} = 1,\]
    \end{enumerate}
    then there exists $\delta > 0$, which depends on $\mu, \eta_0, \Delta$, such that $Z_{\tau_S}(G, \mu, \eta) \neq 0$ for for all $\eta \in \gU([0, \eta_0], \delta) \ (\eta \neq 0)$ with hypergraphs of maximum degree at most $\Delta$ and feasible configuration $\tau_S$.
\end{thm}

Based on Theorem \ref{thm:set-cover} and Barvinok's paradigm \cite{barvinok2016combinatorics}, 
there exists a quasi-polynomial algorithm for computing $Z_{\tau_S}(G, \mu, \eta_0)$. 
Moreover, if the edge size is bounded, then there exists an FPTAS for computing $Z_{\tau_S}(G, \mu, \eta_0)$ according to the algorithm in \cite{liu2019fisher} for computing the low-degree coefficients of partition functions of hypergraphs.


\subsection{Complex Contraction Implies Absence of zeros}
Similar to 2-spin systems, there also exists tree recursion for generalized set covers (see Appendix \ref{app:set-cover:tree}). The recursive function can be written as
\begin{align*}
    f(\vz; \vm, d, \mu, \eta) = \eta \prod_{i=1}^d \left( 1 + \mu \prod_{j=1}^{m_i} \left(\frac{z_{i,j}}{1 + z_{i,j}}\right) \right).
\end{align*}

\begin{definition}
    We say $\eta \in \sC$ satisfies $\Delta$-complex-contraction property for weighted set covers with parameter $\mu \ge -1$ if there exists a close region $F_{\eta} \subseteq \hat{\sC}$ such that
\begin{enumerate}
    \item $-1 \notin F_{\eta}$, $\{0, \infty\} \subseteq F_{\eta}$, 
    \item $f(\vz; \vm, \Delta -1, \mu, \eta) \in F_{\eta}$ and $f(\vz; \vm, \Delta, \mu, \eta) \neq -1$ for $\vz \in F_{\eta}^{\sum_{i=1}^\Delta m_i}$, $\vm \in \sN^{\Delta}$.
\end{enumerate}
\end{definition}

\begin{lemma}\label{lem:complex-contraction-zero-free}
    If $\eta \in \sC$ satisfies $\Delta$-complex-contraction property for weighted set covers with parameter $\mu$, then $Z(G, \mu, \eta) \neq 0$ for any hypergraph $G$ with max degree $\Delta$.
\end{lemma}

\subsection{Complex Contraction}

Applying Mobius transformation $\phi(w) = \frac{w}{w + 1}$, we have
\[f^{\phi}(\vz; \vm, \Delta -1, \mu, \eta) \triangleq \phi \circ f \circ \phi^{-1} = \frac{\eta \prod_{i=1}^{\Delta - 1} \left( 1 + \mu \prod_{j=1}^{m_i} w_{i,j} \right)}{\eta \prod_{i=1}^{\Delta - 1} \left( 1 + \mu \prod_{j=1}^{m_i} w_{i,j} \right) + 1}.\] 
Define 
\[W(k) \triangleq \exp(\tilde{U}(k)) \cup \{0\} = \{ w : |w| \le e^{-|\arg(w)|/k} \}.\]
Then for $w_1, \dots, w_m \in W(k)$, it holds $\prod_{j=1}^m w_j \in W(k)$. 
Similar to Lemma \ref{lem:unbounded-2-spin:U}, we intend to prove the following result.
\begin{lemma}\label{lem:set-cover:W-k}
    For $\eta_0 > 0$, if there exists $k > 0$ such that 
    \begin{align}
        \eta \prod_{i=1}^{\Delta-1} \left( 1 + \mu w_i \right) \in \phi^{-1}\left((W(k))^{\circ}\right)
    \end{align}
    for all $\eta \in [0, \eta_0]$ and $\vw \in W(k)^{\Delta - 1}$, 
    then there exists $\delta > 0$, which depends on $\mu, \eta_0, \Delta$, such that all
    $\eta \in \gU([0, \eta_0], \delta)$ satisfies $\Delta$-complex-contraction property for weighted set covers with parameter $\mu$. 
\end{lemma}

We first try to picture the region $\phi^{-1}(\open{(W(k))})$. 
\begin{lemma}\label{lemma:region1}
    Define 
    \begin{align*}
        r_{k, 1}(\theta) &\triangleq \re\left(\ln \left(\frac{e^{-\theta/k + \i\theta}}{1 - e^{-\theta/k + \i\theta}}\right) \right) = -\frac{1}{2} \ln \left( e^{2\theta/k} + 1 - 2 e^{\theta/k} \cos(\theta) \right), \\
        h_{k, 1}(\theta) &\triangleq \im\left(\ln \left(\frac{e^{-\theta/k + \i\theta}}{1 - e^{-\theta/k + \i\theta}}\right) \right) = \arcsin\left(\frac{\sin(\theta)}{\sqrt{e^{-2\theta/k} + 1 - 2 e^{-\theta/k} \cos(\theta)}} \right).
    \end{align*}
    Then $r_{k, 1}(\theta)$ is strictly decreasing on $(0, \pi)$, $h_{k, 1}(\theta)$ is strictly increasing on $(0, \pi)$ and 
    \begin{align*}
        \left\{\ln(\phi^{-1}(w)): w \in \open{W(k)} \backslash (-e^{-\pi/k}, 0]\right\} = \{z: |\im(z)| < p_k(\re(z))) \},
    \end{align*}
    where 
    \begin{align*}
        p_k(x) = \begin{cases}
            h_{k, 1}(r_{k, 1}^{-1}(x)), ~~~&\text{if } x > - \ln (1 + e^{\pi/k}), \\
            \pi, ~~~&\text{if } x \le - \ln (1 + e^{\pi/k}).
        \end{cases}
    \end{align*}
\end{lemma}


Our next goal is to find the convex hull of the region $\ln(1 + \mu W(k))$. Denote this convex hull by $\hat{W}(k)$. 
By convexity of $\hat{W}(k)$, we have $\sum_{i=1}^{\Delta -1} \ln(1 + \mu w_i) \in \{(\Delta -1) z: z \in \hat{W}(k)\}$ for $\vw \in W(k)^{\Delta -1}$.

\begin{lemma}\label{lemma:region2}
    Define 
    \begin{align*}
        r_{k, 2}(\theta) &\triangleq \re\left(\ln \left(1 + \mu e^{-\theta/k + \i\theta}\right) \right) = \frac{1}{2} \ln \left( 1 + \mu^2 e^{-2\theta / k} + 2\mu e^{-\theta / k} \cos(\theta) \right), \\
        h_{k, 2}(\theta) &\triangleq \im\left(\ln \left(1 + \mu e^{-\theta/k + \i\theta}\right) \right) = \arctan \left( \frac{\mu e^{-\theta/k} \sin(\theta)}{1 + \mu e^{-\theta/k} \cos(\theta)} \right).
    \end{align*}
    Then there exists $\theta_0 \in (0, \pi)$ such that $r_{k, 2}(\theta)$ is strictly monotone on $(0, \theta_0)$ and $(\theta_0, \pi)$ respectively. 
    Denote $g_{k, 2}^{-1}$ to be the inverse function of $r_{k, 1}$ such that $r_{k, 2}^{-1}(r_{k, 2}(\theta)) = \theta$ for $\theta \in (0, \theta_0)$, it holds that
    \begin{align*}
        \ln(1 + \open{(\mu W(k))}) \subseteq \hat{W}(k) \triangleq
        \begin{cases}
        \{z: \re(z) \in (r_{k, 2}(\theta_0), r_{k, 2}(0)), |\im(z)| < h_{k, 2}(r_{k, 2}^{-1}(\re(z))) \}, ~~\text{ if } \mu > 0,  \\
        \{z: \re(z) \in (r_{k, 2}(0), r_{k, 2}(\theta_0)), |\im(z)| < -h_{k, 2}(r_{k, 2}^{-1}(\re(z))) \}, ~~\text{ if } \mu \in [-1, 0),
        \end{cases}
    \end{align*}
    where the function $|h_{k, 2}(r_{k, 2}^{-1}(x))|$ is concave, which implies that $\hat{W}(k)$ is convex.
\end{lemma}


For real $\eta > 0$, by convexity of $\hat{W}(k)$, 
to ensure that $\ln(\eta) + \sum_{i=1}^{\Delta-1} \ln(1 + \mu w_i) \in \ln(T(W(k)^{\circ}))$, we need to prove that for $\theta \in [0, \theta_0]$
$(\Delta - 1) |h_{k, 2}\left(\theta\right)| < p_k\big( \ln(\eta) + (\Delta - 1) g_{k, 2}(\theta) \big)$.

    For $\mu = -1$, notice that $\lim_{\theta \to 0} |h_{k, 2}\left(\theta\right)| = \arctan k$, $\lim_{\theta \to 0} g_{k, 2}\left(\theta\right) = -\infty$ and $\lim_{x \to -\infty} p_k\left(x\right) = \pi$, the above inequality is still meaningful.

\begin{lemma}\label{lemma:compare}
    For $\eta_0 > 0, \mu \ge -1$, if one of three condition defined in Theorem \ref{thm:set-cover} holds, then there exists $k > 0$ such that for all $\eta \in [0, \eta_0]$ and $\theta \in [0, \theta_0]$, we have $(\Delta - 1) |h_{k, 2}\left(\theta\right)| < p_k\big( \ln(\eta) + (\Delta - 1) g_{k, 2}(\theta) \big)$.
\end{lemma}

\section{Conclusion and Discussion}

We give complex contraction proof of the zero-free region for counting problems which have tree expansions. We focus on obtaining the zero-freeness other than proving correlation decays. Our region can deal with unbounded degree cases and for 2-spin systems we can also deal with unbounded degree case. As corollary of our proof, we can have FPTAS for counting problems on bounded degree graphs and quasi-polynomial time algorithms for problems on unbounded degree graphs.

We expect our technique can be applied to more problems. For example, those problems which can be solved by the Asano contraction techniques.

For some of the counting problems, there is the uniqueness condition for justifying the tightness of the correlation decay property. An ambitious goal for our approach is: what is the tight condition for a contraction region out of which the zero-freeness does not hold?

\appendix
\section{Discussion on $\hat{x}_d$ and $\lambda_c(d)$}
\label{app:hat-x}
Recall the definition of $\hat{x}_d$ and $\lambda_c(d)$:
\begin{align}
    \hat{x}_d &\triangleq \begin{cases}
        \frac{-1 - \beta\gamma + d(1 - \beta\gamma) - \sqrt{(-1 - \beta\gamma + d (1 - \beta\gamma))^2 - 4\beta\gamma}}{2 \gamma}, ~~~&\text{ for } \gamma > 0, \\
        \frac{\beta}{d - 1}, ~~~&\text{ for } \gamma = 0,
    \end{cases} \\
    \lambda_c(d) &\triangleq \hat{x}_d \left(\frac{\hat{x}_d + \beta}{\gamma \hat{x}_d + 1}\right)^d,
\end{align}
where $\sqrt{\beta\gamma} \le \frac{d-1}{d+1}, \beta > 0, \gamma \ge 0$. 

Denote $\bar{d} \triangleq \frac{1 + \sqrt{\beta\gamma}}{1 - \sqrt{\beta\gamma}}$. Note that $\sqrt{\beta\gamma} \le \frac{d-1}{d+1}$ holds if and only if $d \ge \bar{d}$. 

We provide some properties of $\hat{x}_d$ and $\lambda_c(d)$ in following lemmas.
\begin{lemma}
    For $d \ge \bar{d}$, we have $(-1 - \beta\gamma + d (1 - \beta\gamma))^2 - 4\beta\gamma \ge 0$. Further if $\gamma > 0$, there also holds $\hat{x}_d \le \sqrt{\frac{\beta}{\gamma}}$. 
\end{lemma}
\begin{proof}
    Note that for $d \ge \bar{d}$, we have $\beta\gamma < 1$ and
    \begin{align}
    \notag -1 - \beta\gamma + d (1 - \beta\gamma) &\ge -1 - \beta\gamma + \bar{d} (1 - \beta\gamma) = -1 - \beta\gamma + \frac{(1 + \sqrt{\beta\gamma})(1 - \beta\gamma)}{1 - \sqrt{\beta\gamma}} \\
    \label{eq:hat-x-d} &= -1 - \beta\gamma + (1 + \sqrt{\beta\gamma})^2 = 2\sqrt{\beta\gamma}.
    \end{align}
    
    Moreover, observe that
    \begin{align*}
        &\qquad\qquad \hat{x}_d \le \sqrt{\beta/\gamma}, \\
        &\iff -1 - \beta\gamma + d(1 - \beta\gamma) - \sqrt{(-1 - \beta\gamma + d (1 - \beta\gamma))^2 - 4\beta\gamma} \le 2\sqrt{\beta\gamma}, \\
        &\iff -1 - \beta\gamma + d(1 - \beta\gamma) - 2\sqrt{\beta\gamma} \le \sqrt{(-1 - \beta\gamma + d (1 - \beta\gamma))^2 - 4\beta\gamma}, \\
        &\iff 2\sqrt{\beta\gamma} \le -1 - \beta\gamma + d(1 - \beta\gamma),
    \end{align*}
    which has been proved in Equation (\ref{eq:hat-x-d}).
\end{proof}

\begin{lemma}
    $\hat{x}_d$ is strictly decreasing with respect to $d$.
\end{lemma}
\begin{proof}
    Just note that
    \begin{align*}
        \hat{x}_d = \frac{2 \beta}{-1- \beta\gamma + d(1 - \beta\gamma) + \sqrt{(-1- \beta\gamma + d(1 - \beta\gamma))^2 - 2 \beta\gamma}}.
    \end{align*}
\end{proof}

\begin{lemma}\label{lem:lambda-c}
    For $\beta \le 1$, $\lambda_c(d)$ is strictly decreasing when $d > \bar{d}$. 
    For $\beta > 1$, $\lambda_c(d)$ is strictly decreasing on $(\bar{d}, d_c)$ and is strictly increasing on $(d_c, +\infty)$.
\end{lemma}
\begin{proof}
    We first suppose that $\gamma > 0$. Note that for $d > \bar{d}$, we have $\hat{x}_d \in (0, \sqrt{\beta/\gamma})$ and
    \begin{align*}
        \lambda_c'(d) &= \left(\frac{\hat{x}_d + \beta}{\gamma \hat{x}_d + 1}\right)^d \left(\tilde{x}'(d) + \hat{x}_d \left(\ln\left(\frac{\hat{x}_d + \beta}{\gamma \hat{x}_d + 1}\right) + \frac{d(1 - \beta\gamma)\tilde{x}'(d)}{(\beta + \hat{x}_d)(1 + \gamma \hat{x}_d)}\right)\right) \\
        &= \left(\frac{\hat{x}_d + \beta}{\gamma \hat{x}_d + 1}\right)^d \left( 2 \tilde{x}'(d) + \hat{x}_d \ln\left(\frac{\hat{x}_d + \beta}{\gamma \hat{x}_d + 1}\right) \right) \\
        &= \left(\frac{\hat{x}_d + \beta}{\gamma \hat{x}_d + 1}\right)^d \left( -\frac{2 (1 - \beta\gamma) \hat{x}_d}{\sqrt{(-1-\beta\gamma + d(1 - \beta\gamma))^2 - 4 \beta\gamma}}  + \hat{x}_d \ln\left(\frac{\hat{x}_d + \beta}{\gamma \hat{x}_d + 1}\right) \right) \\
        &= \hat{x}_d \left(\frac{\hat{x}_d + \beta}{\gamma \hat{x}_d + 1}\right)^d \left( -\frac{2 (1 - \beta\gamma) \hat{x}_d}{\beta - \gamma \hat{x}_d^2} + \ln\left(\frac{\hat{x}_d + \beta}{\gamma \hat{x}_d + 1}\right) \right)
    \end{align*}
    
    Now consider the function 
    \begin{align*}
        \psi(x) \triangleq \frac{2 (1 - \beta\gamma) x}{\beta - \gamma x^2}  +  \ln\left(\frac{\gamma x + 1}{x + \beta}\right).
    \end{align*}
    Note that
    \begin{align*}
        \psi'(x) = \frac{(1 - \beta  \gamma) (\beta + 2 \beta\gamma x  + \gamma^2 x) (\beta + 2 x + \gamma  x^2)}{(\beta +x) (\gamma  x+1) \left(\beta -\gamma  x^2\right)^2} > 0.
    \end{align*}
     
    When $\beta \le 1$, we have $\psi(x) > \psi(0) = -\ln(\beta) \ge 0$ for $x \in (0, \sqrt{\beta/\gamma})$, and consequently $\lambda_c'(d) < 0$ for $d > \bar{d}$. 
     
    On the other hand, when $\beta > 1$, observe that $\psi(x) \to + \infty$ for $x \to \sqrt{\beta/\gamma}$ and $\psi(0) = -\ln(\beta) < 0$ which implies that there exists $x_c \in (0, \sqrt{\beta/\gamma})$ such that $\psi(x) > 0$ on $(0, x_c)$ and $\psi(x) < 0$ on $(x_c, \sqrt{\beta/\gamma})$. 
    
    Together with $\hat{x}_d$ is strictly decreasing with respect to $d$ and $\hat{x}_{\bar{d}} = \sqrt{\beta/\gamma}$, $\lim_{d \to +\infty} \hat{x}_d = 0$, there exists $d_c > \bar{d}$ such that $\hat{x}_{d_c} = x_c$. Then it is clearly that $\lambda_c'(d) < 0$ on $(\bar{d}, d_c)$ and $\lambda_c'(d) > 0$ on $(d_c, +\infty)$. 
     
    At last, we remark that with replacing $\sqrt{\beta/\gamma}$ to $+\infty$, all above discussion is still true for $\gamma = 0$. 
\end{proof}

\begin{lemma}\label{lem:hat-x:bar-x}
    For $\sqrt{\beta\gamma} \le \frac{d-1}{d+1}$ and $\lambda \in (0, \lambda_c(d))$, let $\bar{x} \in (0, +\infty)$ be the fix point of $f(x) = \lambda \left(\frac{\gamma x + 1}{x + \beta}\right)^d$, then we have $\bar{x} < \hat{x}_d$ and
    \begin{align*}
        \frac{d(1 - \beta\gamma) \bar{x}}{(\gamma \bar{x} + 1)(\bar{x} + \beta)} < 1.
    \end{align*}
\end{lemma}
\begin{proof}
    By $\beta\gamma < 1$, the function $f(x)$ is decreasing on $[0, +\infty)$. If $\bar{x} \ge \hat{x}_d$, then we have
    \begin{align*}
        f(\bar{x}) \le f(\hat{x}_d) < \lambda_c(d) \left(\frac{\gamma \hat{x}_d + 1}{\hat{x}_d + \beta}\right)^d = \hat{x}_d \le \bar{x} = f(\bar{x}).
    \end{align*}
    Therefore, we have $\bar{x} < \hat{x}_d$. 
    
    Denote \[\psi(x) = \frac{d (1 - \beta\gamma) x}{(\gamma x + 1)(x + \beta)}.\]
    
    For $\gamma = 0$, the function $\psi(x) = \frac{d x}{x + \beta}$ is strictly increasing on $[0, +\infty)$ and therefore
    \[
    \psi(\bar{x}) < \psi(\hat{x}(d)) = \frac{d \hat{x}_d}{\hat{x}_d + \beta} = 1.
    \]
    
    On the other hand, for $\gamma > 0$, the function $\psi(x)$ is strictly increasing when $x < \sqrt{\beta/\gamma}$ and strictly decreasing when $x > \sqrt{\beta/\gamma}$.
    Consequently, together with $\bar{x} < \hat{x}_d \le \sqrt{\beta/\gamma}$, we have
    \begin{align*}
        \psi(\bar{x}) < \psi(\hat{x}_d) = 1.  
    \end{align*}
\end{proof}

\section{Discussion on $\check{x}_d$}
\label{app:check-x}
Recall the definition of $\check{x}_d$:
\begin{align*}
    \check{x}_d \triangleq \frac{1 - \beta \gamma + d(1 + \beta \gamma) +\sqrt{(1 - \beta \gamma + d(1 + \beta \gamma))^2-4 \beta  \gamma  d^2}}{2 \beta  d},
\end{align*}
which is the largest root of the equation 
\begin{align*}
    \psi(x) \triangleq \beta d x^2 - (1 - \beta \gamma + d(1 + \beta \gamma)) x + \gamma d = 0.
\end{align*}

We provide some properties of $\check{x}_d$ in following lemmas.
\begin{lemma}\label{lem:check-x-beta}
    For $\beta \gamma < 1$ and $\beta > 0$, we have $\check{x}_d > \frac{1}{\beta}$.
\end{lemma}
\begin{proof}
    Note that
    \begin{align*}
        \psi(1/\beta) = \frac{d - (1 - \beta\gamma + d(1 + \beta\gamma)) + d\beta\gamma}{\beta} = - \frac{1 - \beta\gamma}{\beta} < 0,
    \end{align*}
    which results in the desired result together with the property of quadratic functions. 
\end{proof}

\begin{lemma}\label{lem:check-x-1}
    For $\beta \gamma < 1$, $\check{x}_d < 1$ holds if and only if $\beta > 1$ and $d > \frac{1 - \beta\gamma}{(\beta - 1)(1 - \gamma)}$.
\end{lemma}
\begin{proof}
    If $\beta \le 1$, then by Lemma \ref{lem:check-x-beta}, there holds $\check{x}_d > 1/\beta \ge 1$. 
    On the other hand, for $\beta > 1$, note that the axis of symmetry of the graph of $\psi(x)$ is
    \begin{align*}
        \frac{1 - \beta \gamma + d(1 + \beta\gamma)}{2 \beta d} \le \frac{2 d}{2 \beta d} < 1.
    \end{align*}
    Consequently, by the property of quadratic functions, $\check{x}_d < 1$ holds if and only if $\psi(1) > 0$, i.e., 
    \begin{align*}
        \psi(1) &= d(\beta + \gamma) - (1 - \beta\gamma + d(1 + \beta\gamma)) \\
        &= d(\beta - 1)(1 - \gamma) - (1 - \beta\gamma) > 0,
    \end{align*}
    which implies the desired result. 
\end{proof}

\begin{lemma}\label{lem:check-x-2}
    For $\beta \gamma < 1$ and $\beta > 1$, $\check{x}_d < \frac{1 - \gamma}{\beta - 1}$ holds if and only if $d > \frac{(\beta - 1)(1 - \gamma)}{1 - \beta\gamma}$.
\end{lemma}
\begin{proof}
    The axis of symmetry of the graph of $\psi(x)$ is
    \begin{align*}
        \frac{1 - \beta \gamma + d(1 + \beta\gamma)}{2 \beta d} \le \frac{2 d}{2 \beta d} < \frac{1 - \gamma}{\beta - 1}.
    \end{align*}
    Consequently, by the property of quadratic functions, $\check{x}_d < \frac{1 - \gamma}{\beta - 1}$ holds if and only if $\psi(1) > 0$, i.e., 
    \begin{align*}
        d \beta (1 - \gamma)^2 - (1 - \beta\gamma + d(1 + \beta\gamma))(1 - \gamma)(\beta - 1) + d \gamma (\beta - 1)^2 
        = (1 - \beta\gamma) \left(d(1 - \beta\gamma) - (\beta - 1)(1 - \gamma)\right) > 0,
    \end{align*}
    which implies the desired result. 
\end{proof}

\section{Proof of Lemma \ref{lem:bounded-2-spin:computation-tree}}
We write $Z_{\tau_{S}}(G, \lambda)$ for $Z_{\tau_{S}}^{\beta,\gamma}(G, \lambda)$ and $R_{\tau_S}(G, v)$ for $R_{\tau_S}^{\beta,\gamma,\lambda}(G, v)$ in this section.
\begin{proof}
    Firstly, it is clearly that 
    \begin{align*}
        R_{\tau_S}(G, v) = \frac{Z_{\tau_{S}, \sigma(v)=1}(G, \lambda)}{Z_{\tau_{S}, \sigma(v)=0}(G, \lambda)}
        = \frac{1}{\lambda^{d-1}} \frac{Z_{\tau_{S} \cup \sigma_0}(\tilde{G}, \lambda)}{Z_{\tau_{S} \cup \sigma_{d+1}}(\tilde{G}, \lambda)} 
        = \frac{1}{\lambda^{d-1}} \prod_{i=1}^d \frac{Z_{\tau_{S} \cup \sigma_i, \sigma(\tilde{v}_i)=1}(\tilde{G}, \lambda)}{Z_{\tau_{S} \cup \sigma_i, \sigma(\tilde{v}_i)=0}(\tilde{G}, \lambda)},
    \end{align*}
    where $\sigma_0(\tilde{v}_j) = 1, \sigma_{d+1}(\tilde{v}_j) = 0$ for $j \in [d]$. 
    
    Note that in graph $\tilde{G}$, $\tilde{v}_i$ only occurs in the edge $(\tilde{v}_i, v_i)$, hence for free vertex $v_i$ we have 
    \begin{align*}
        Z_{\tau_{S} \cup \sigma_i, \sigma(\tilde{v}_i)=1}(\tilde{G}, \lambda) &= \lambda (Z_{\tau_{S} \cup \sigma_i, \sigma(v_i)=0}(\tilde{G} - \tilde{v}_i, \lambda) + \gamma Z_{\tau_{S} \cup \sigma_i, \sigma(v_i)=1}(\tilde{G} - \tilde{v}_i, \lambda)), \\
        Z_{\tau_{S} \cup \sigma_i, \sigma(\tilde{v}_i)=0}(\tilde{G}, \lambda) &= \beta Z_{\tau_{S} \cup \sigma_i, \sigma(v_i)=0}(\tilde{G} - \tilde{v}_i, \lambda) + Z_{\tau_{S} \cup \sigma_i, \sigma(v_i)=1}(\tilde{G} - \tilde{v}_i, \lambda),
    \end{align*}
    and consequently
    \begin{align}
        \frac{Z_{\tau_{S} \cup \sigma_i, \sigma(\tilde{v}_i)=1}(\tilde{G}, \lambda)}{Z_{\tau_{S} \cup \sigma_i, \sigma(\tilde{v}_i)=0}(\tilde{G}, \lambda)} = \lambda\frac{\gamma R_{\tau_S \cup \sigma_i}(G_i, v_i) + 1}{R_{\tau_S \cup \sigma_i}(G_i, v_i) + \beta}.
    \end{align}
    If $v_i$ is pinned to $1$ by $\tau_S$, then there holds
    \begin{align*}
        \frac{Z_{\tau_{S} \cup \sigma_i, \sigma(\tilde{v}_i)=1}(\tilde{G}, \lambda)}{Z_{\tau_{S} \cup \sigma_i, \sigma(\tilde{v}_i)=0}(\tilde{G}, \lambda)} = \frac{\lambda \gamma Z_{\tau_{S} \cup \sigma_i}(G_i, \lambda)}{Z_{\tau_{S} \cup \sigma_i}(G_i, \lambda)} = \lambda \gamma = \lambda\frac{\gamma R_{\tau_S \cup \sigma_i}(G_i, v_i) + 1}{R_{\tau_S \cup \sigma_i}(G_i, v_i) + \beta},
    \end{align*}
    where we have recalled that $R_{\tau_S \cup \sigma_i}(G_i, v_i) = \infty$ for $v_i \in S$ and $\tau_S(v_i) = 1$. 
    
    Similarly, if $v_i \in S$ and $\tau_S(v_i) = 0$, then it holds that
    \begin{align*}
        \frac{Z_{\tau_{S} \cup \sigma_i, \sigma(\tilde{v}_i)=1}(\tilde{G}, \lambda)}{Z_{\tau_{S} \cup \sigma_i, \sigma(\tilde{v}_i)=0}(\tilde{G}, \lambda)} = \frac{\lambda}{\beta} = \lambda\frac{\gamma R_{\tau_S \cup \sigma_i}(G_i, v_i) + 1}{R_{\tau_S \cup \sigma_i}(G_i, v_i) + \beta}.
    \end{align*}
    
    Therefore, we can conclude that 
    \begin{align*}
        R_{\tau_S}(G, v) = \lambda \prod_{i=1}^d \frac{\gamma R_{\tau_S \cup \sigma_i}(G_i, v_i) + 1}{R_{\tau_S \cup \sigma_i}(G_i, v_i) + \beta},
    \end{align*}
    which is the desired result.
\end{proof}

\section{Proof of Lemma \ref{lem:bounded-2-spin:contraction-zero-free}}
We write $Z_{\tau_{S}}(G, \lambda)$ for $Z_{\tau_{S}}^{\beta,\gamma}(G, \lambda)$ and $R_{\tau_S}^{\lambda}(G, v)$ for $R_{\tau_S}^{\beta,\gamma,\lambda}(G, v)$ in this section.
\begin{proof}
    
    Before proving the results, recall that for $\beta = 0$, we can pin all neighbors of $u$ such that $\tau_S(u) = 0$ to $1$ without changing any values of marginal ratios. Similarly, for $\gamma = 0$ and $\tau_S(u) = 1$, we also can pin all neighbors of $u$ to $0$. 
    
    Let $F_{\lambda}$ be the region satisfies the conditions in Definition \ref{def:2-spin:complex-contraction}.
    We will prove a stronger result:
    \begin{enumerate}
        \item $R_{\tau_S}^{\lambda}(G, v) \in F_{\lambda}$ holds for free vertex $v$ with degree $d < \Delta$,
        \item $R_{\tau_S}^{\lambda}(G, v) \neq -1$ holds for free vertex $v$ with degree $d = \Delta$, 
        \item $Z_{\tau_S}(G, \lambda) \neq 0$.
    \end{enumerate}
    We use induction on the number of free vertices, i.e. $t = |V| - |S|$, to prove this result.
    
    If $|V| = |S|$, that is $S = V$. Note that
    \begin{align*}
        Z_{\tau_S}(G, \lambda) = \beta^{n_1} \gamma^{n_2} \lambda^{m},
    \end{align*}
    where $n_1 = |\{e = (u, v) \in E: \tau_S(v) = \tau_S(u) = 0\}|$, $n_2 = |\{e = (u, v) \in E: \tau_S(v) = \tau_S(u) = 1\}|$ and $m = |\{v \in V: \tau_S(v) = 1\}|$.
    Recalling that $\lambda \neq 0$ and $n_1 = 0$ if $\beta = 0$, $n_2 = 0$ if $\gamma = 0$, we know that $Z_{\tau_S}(G, \lambda) \neq 0$.
    
    Now assume that the desired result holds for $|V| - |S| \le t$, let consider the case of $|V| - |S| = t + 1$. 
    Choose a free vertex $v \in V$. By induction hypothesis, we know that $Z_{\tau_S, \sigma(v) = 1}(G, \lambda) \neq 0$ for $\gamma \neq 0$ and $Z_{\tau_S, \sigma(v) = 0}(G, \lambda) \neq 0$ for $\beta \neq 0$. Therefore $R_{\tau_S}^{\lambda}(G, v)$ is well-defined.

    Following from Lemma \ref{lem:bounded-2-spin:computation-tree}, we know that 
    \begin{align*}
        R_{\tau_S}^{\lambda} (G, v) = f_{\lambda, d} (R_{\tau_S \cup \sigma_1}^{\lambda}(G_1, v_1), \dots, R_{\tau_S \cup \sigma_d}^{\lambda}(G_d, v_d)).
    \end{align*}
    Observe that the number of free vertex in $G_i$ with configuration $\tau_S \cup \sigma_i$ is at most $t$. 
    Moreover, if $v_i$ is a free vertex in $G$, then $v_i$ is also a free vertex in $G_i$. 
    Following from $\deg_{G_i}(v_i) = \deg_{\tilde{G}}(v_i) - 1 = \deg_{G}(v_i) - 1 \le \Delta - 1$ and induction hypothesis, we know that $R_{\tau_S \cup \sigma_d}^{\lambda}(G_i, v_i) \in F_{\lambda}$ for free vertex $v_i$. 
    For $v_i \in S$, if $\beta \neq 0, \gamma \neq 0$, we know that $\{0, \infty\} \subseteq F_{\lambda}$ and
    $R_{\tau_S \cup \sigma_d}^{\lambda}(G_i, v_i) \in \{0, \infty\} \subseteq F_{\lambda}$. 
    And if $\beta = 0$, then $\gamma \neq 0$ and $\infty \in F_{\lambda}$. 
    Moreover, if $\tau_S(v_i) = 0$, recall that we have pinned all neighbors of $v_i$, especially $v$, to $1$, and $v$ cannot be a free vertex. 
    Hence $\tau_S(v_i) = 1$ and $R_{\tau_S \cup \sigma_d}^{\lambda}(G_i, v_i) = \infty \in F_{\lambda}$. Similarly if $\gamma = 0$, then $R_{\tau_S \cup \sigma_d}^{\lambda}(G_i, v_i) = 0 \in F_{\lambda}$.
    
    Therefore, we conclude that 
    \[R_{\tau_S \cup \sigma_d}^{\lambda}(G_i, v_i) \in F_{\lambda}, ~~~ i = 1, \dots, d.\]
    
    If $d \le \Delta - 1$, then by $f_{\lambda, d}(z_1, \dots, z_d) \in F_{\lambda}$ for any $d = 0, \dots, \Delta-1$ and $z_i \in F_{\lambda}, i = 1, \dots, d$, we know that $R_{\tau_S}^{\lambda} (G, v) \in F_{\lambda}$. And according to $-1 \notin F_{\lambda}$, we have $R_{\tau_S}^{\lambda} (G, v) \neq -1$.
    
    If $d = \Delta$, then by $f_{\lambda, \Delta}(z_1, \dots, z_{\Delta}) \neq -1$ for $z_i \in F_{\lambda}, i = 1, \dots, \Delta$, it also holds that $R_{\tau_S}^{\lambda} (G, v) \neq -1$.
    
    Finally, $Z_{\tau_S}(G, \lambda) \neq 0$ can be deduced from $Z_{\tau_S}(G, \lambda) = Z_{\tau_S, \sigma(v) = 0}(G, \lambda) + Z_{\tau_S, \sigma(v) = 1}(G, \lambda)$ and $R_{\tau_S}^{\lambda} (G, v) = \frac{Z_{\tau_S, \sigma(v) = 0}(G, \lambda)}{Z_{\tau_S, \sigma(v) = 1}(G, \lambda)} \neq -1$. 
\end{proof}

\section{Proof of Lemma \ref{lem:bounded-2-spin:U}}
\begin{proof}
    Set $F = \phi^{-1}(U)$.  
    By $U \subseteq \{z: |\im{z}| < \pi \}$, $\phi^{-1}$ is a bijection from $U$ to $F$. 
    Then $-1 \notin F$, $\phi^{-1}(-\ln(\beta)) = 0 \in F$, and $\phi^{-1}(\ln(\gamma)) = \infty \in F$ if $\gamma > 0$.
    
    We first prove that there exists $\delta > 0$ such that for all $\lambda \in \gU([0, \lambda_0], \delta)$ and $w \in U$ we have
    \begin{enumerate}
        \item $\lambda e^{\Delta w} \neq -1$,
        \item $\lambda e^{d w} \notin (-\infty, \max\{-\beta, -1/\gamma\}]$ for $d = 0, 1, \dots, \Delta -1$, 
        \item $\phi(\lambda e^{d w}) \in U$ for $d = 0, 1, \dots, \Delta - 1$.
    \end{enumerate}
    
    Let $\gI \triangleq (-\infty, \max\{-\beta, -1/\gamma\}]$.
    According to $U$ is compact, we know that
    \begin{align*}
        U_d \triangleq \{\lambda e^{d w}: \lambda \in [0, \lambda_0], w \in U\}
    \end{align*}
    is also compact. 
    Together with $U_d \cap \gI = \emptyset$ and $-1 \notin U_{\Delta}$, there exists $\eps_1 > 0$ such that $\mathrm{dist}(U_d, \gI) \ge \eps_1$ for $d = 0, 1, \dots, \Delta - 1$, and $\mathrm{dist}(-1, U_{\Delta}) \ge \eps_1$. 
    It is apparent that there exists $\delta_1 > 0$ such that for $|\lambda_1 e^{d w} - \lambda_2 e^{d w}| < \eps_1/2$ for all $|\lambda_1 - \lambda_2| < \delta_1$ and $w \in U$, $d = 0, 1, \dots, \Delta$. Hence, for all $\lambda \in \gU([0, \lambda_0], \delta_1)$ and $w \in U$, we have $\lambda e^{d w} \in \gU(U_d, \eps_1/2)$, and consequently $\lambda e^{\Delta w} \neq -1$ and $\lambda e^{d w} \notin \gI$ for $d = 0, 1, \dots, \Delta -1$.
    
    Based on $\overline{\gU(U_d, \eps_1/2)} \cap \gI = \emptyset$, it is easily to check that $\phi$ is bijiective and holomorphic on $\overline{\gU(U_d, \eps_1/2)}$.
    
    Furthermore, for $d = 0, 1, \dots, \Delta - 1$, we have $\phi(U_d) \subseteq \open{U}$. Thus there exist $\eps_2 > 0$ such that $\mathrm{dist}(\phi(U_d), \partial U) \ge \eps_2$. 
    Then with recalling that $\phi$ is continuous on the compact set $\overline{\gU(U_d, \eps_1/2)}$, there exists $\tilde{\eps} > 0$ such that $|\phi(z_1) - \phi(z_2)| < \eps_2/2$ for all $z_1, z_2 \in \overline{\gU(U_d, \eps_1/2)}$ with $|z_1 - z_2| < \tilde{\eps}$. Moreover, there exists $\delta \in (0, \delta_1)$ such that for all $|\lambda_1 - \lambda_2| < \delta$, there holds $|\lambda_1 e^{d w} - \lambda_2 e^{d w}| < \tilde{\eps}$, which implies that $|\phi(\lambda_1 e^{d w}) - \phi(\lambda_2 e^{d w})| < \eps_2/2$. 
    Therefore, for $\lambda \in \gU([0, \lambda_0], \delta)$, we have $\phi(\lambda e^{d w}) \in \gU(\phi(U_d), \eps_2/2) \subseteq U$ for $w \in U$ and $d = 0,1, \dots, \Delta-1$.
    
    \vskip 8pt
    Next, we prove that all $\lambda \in \gU([0, \lambda_0], \delta)$ satisfies $\Delta$-complex-contraction property with region $F_{\lambda} = F$. 
    
    Note that for $z_i \in F, i \in 1, \dots, \Delta$, there exists $w_i \in U$ such that $z_i = \phi^{-1}(w_i)$.
    Then we have
    \begin{align*}
        f_{\lambda, d}(z_1, \dots, z_{\Delta}) = f_{\lambda, \Delta}(\phi^{-1}(w_1), \dots, \phi^{-1}(w_{\Delta})) = \lambda \exp\left(\sum_{i=1}^{\Delta} w_i\right).
    \end{align*}
    By convexity of $U$ we know that $\frac{1}{\Delta} \sum_{i=1}^{\Delta} w_i \in U$. 
    Hence $f_{\lambda, \Delta}(z_1, \dots, z_{\Delta}) \neq -1$ according $\lambda e^{\Delta w} \neq -1$ for $w \in U$. 
    
    And for $0 \le d \le \Delta-1$, observe that 
    \begin{align*}
        \phi(f_{\lambda, d}(z_1, \dots, z_{d})) = \phi(\lambda e^{d \bar{w}}) \subseteq U = \phi(F),
    \end{align*}
    where $\bar{w} = \frac{1}{d} \sum_{i=1}^d w_i \in U$. Therefore following from $\phi$ is a injection on $\sC \backslash \gI$ and $f_{\lambda, d}(z_1, \dots, z_{d}) = \lambda e^{d \bar{w}} \notin \gI$, we can conclude that $f_{\lambda, d}(z_1, \dots, z_d) \in F$.
\end{proof}


\section{Proof of Lemma \ref{lem:bounded-2-spin:H}}
\begin{proof}
    By (uniformly) continuity of $H$, there exists $\delta > 0$ such that $H(x, \lambda) > \delta$ holds for all $x \in [x_1, x_0]$ and $\lambda \in [0, \lambda_0]$. Similarly, by (uniformly) continuity of $\frac{\partial G(x, k, \lambda)}{\partial k}$, there exists $k_0$ such that 
    $
        \frac{\partial G(x, k, \lambda)}{\partial k} > \frac{\delta}{2}
    $
    for $x \in [x_1, x_0]$, $\lambda \in [0, \lambda_0]$ and $k \in [0, k_0]$. 
    
    Therefore, we have
    \begin{align*}
        G(x, k_0, \lambda) = \int_{0}^{k_0} \frac{\partial G(x, k, \lambda)}{\partial k} d k > \frac{k_0 \delta}{2} > 0
    \end{align*}
    for all $x \in [x_1, x_0]$ and $\lambda \in [0, \lambda_0]$.
\end{proof}

\section{Proof of Lemma \ref{lem:bounded-2-spin:x0-x1}}
\label{app:2-spin:x0-x1}
In this section, we always assume that $\beta > 0$, $\gamma \ge 0$, $\beta\gamma \neq 1$ and $d = \Delta - 1$. 

We first provide the close form of $\phi(\lambda e^{d w})$ for $w \in U(x_0, x_1, k)$. 
\begin{lemma}\label{lem:bounded-2-spin:r-h}
    For $\lambda \in \sR$ and $\lambda \neq 0$, suppose that $1 + \gamma \lambda e^{d x_0} > 0$ and $\beta + \lambda e^{d x_0} > 0$, then for $w = x + \i y$ with $x \le x_0$ and $|y| < \frac{\pi}{2}$, we have
    \begin{align*}
        \re \left(\phi(\lambda e^{w})\right) &= r(x, y, \lambda) \triangleq \frac{1}{2} \ln \left(\frac{1 + \gamma^2 \lambda^2 e^{2 x} + 2 \gamma \lambda e^{x} \cos(y)}{\beta^2 + \lambda^2 e^{2 x} + 2 \beta \lambda e^{x} \cos(y)}\right), \\
        |\im\left(\phi(\lambda e^{w})\right)| &= h(x, y, \lambda) \triangleq \arctan\left( \frac{|(\beta \gamma - 1) \lambda e^{x} \sin(y)|}{\beta + \gamma \lambda^2 e^{2 x} + (\beta \gamma + 1)\lambda e^{x} \cos(y)} \right).
    \end{align*}
\end{lemma}
\begin{proof}
    Recall that 
    \begin{align*}
        \phi(\lambda e^{w}) = \ln\left(\frac{1 + \gamma \lambda e^{w}}{\beta + \lambda e^{w}}\right). 
    \end{align*}
    
    Observe that
    \begin{align*}
        \re(1 + \gamma \lambda e^{w}) = 1 + \gamma \lambda e^{x} \cos(y).
    \end{align*}
    If $\lambda > 0$, then $1 + \gamma \lambda e^{x} \cos(y) > 0$ holds apparently. 
    Otherwise, for $\lambda < 0$, we have
    \begin{align*}
        1 + \gamma \lambda e^{x} \cos(y) \ge 1 + \gamma \lambda e^{x_0} > 0.
    \end{align*}
    Therefore there holds $\re(1 + \gamma \lambda e^{w}) > 0$ and similarly $\re(\beta + \lambda e^{w}) > 0$. 
    Further, we point out that
    \begin{align*}
        \im(1 + \gamma \lambda e^{w}) \im(\beta + \lambda e^{w}) = \gamma \lambda^2 e^{2 x} \sin^2(y) \ge 0.
    \end{align*}
    
    Now let $z_1 = x_1 + \i y_1$ and $z_2 = x_2 + \i y_2$ with $x_1, x_2 > 0$ and $y_1 y_2 \ge 0$. It holds that
    \begin{align*}
        \frac{z_1}{z_2} = \frac{(x_1 + \i y_1)(x_2 - \i y_2)}{x_2^2 + y_2^2} = \frac{x_1 x_2 + y_1 y_2}{x_2^2 + y_2^2} + \i \frac{-x_1 y_2 + y_1 x_2}{x_2^2 + y_2^2}.
    \end{align*}
    Hence we have $\re(z_1/z_2) > 0$ and 
    \begin{align*}
        \arg(z_1/z_2) = \arctan \left( \frac{-x_1 y_2 + y_1 x_2}{x_1 x_2 + y_1 y_2} \right).
    \end{align*}
    Consequently, there holds
    \begin{align*}
        \ln(z_1/z_2) = \frac{1}{2} \ln\left(\frac{x_1^2 + y_1^2}{x_2^2 + y_2^2}\right) + \i \arctan \left( \frac{-x_1 y_2 + y_1 x_2}{x_1 x_2 + y_1 y_2} \right).
    \end{align*}
    Substitute $z_1 = 1 + \gamma \lambda e^{w}$ and $z_2 = \beta + \lambda e^{w}$ into above equation to obtain the desired result.
\end{proof}

Before giving the specific value of $x_0, x_1$, we prove the following result. 
\begin{lemma}\label{lem:bounded-2-spin:tilde-G}
    Suppose that $k < \min\left\{\frac{\pi}{4 \Delta (x_0 - x_1)}, 1\right\}$, and $1 + \gamma \lambda e^{d x_0} > 0, \beta + \lambda e^{d x_0} > 0$. In addition, if $\beta\gamma < 1$, we assume that $k < \frac{\sqrt{2(\beta^2\gamma^2 + 1)} - (\beta\gamma + 1) }{1 - \beta\gamma}$. 
    Then for $d = \Delta - 1$, let
    \begin{align}
        \tilde{G}(x, y, k, \lambda) \triangleq k(x_0 - r(d x, d y, \lambda)) - h(d x, d y, \lambda),
    \end{align}
    where $x \in [x_1, x_0]$ and $y \in [0, k(x_0 - x)]$. 
    Then we have
    \begin{align*}
        \tilde{G}(x, y, k, \lambda) \ge \tilde{G}(d x, d k(x_0 - x), k, \lambda) = G(x, k, \lambda).
    \end{align*}
\end{lemma}
\begin{proof}
    Note that
    \begin{align*}
        \tilde{G}(x, y, k, \lambda) = k \left(x_0 - \frac{1}{2} \ln \left(\frac{1 + \gamma^2 \lambda^2 e^{2 d x} + 2 \gamma \lambda e^{d x} \cos(d y)}{\beta^2 + \lambda^2 e^{2 d x} + 2 \beta \lambda e^{d x} \cos(d y)}\right) \right) \\
        \quad - \arctan\left( \frac{|(\beta \gamma - 1) \lambda| e^{d x} \sin(d y)}{\beta + \gamma \lambda^2 e^{2 d x} + (\beta \gamma + 1)\lambda e^{d x} \cos(d y)} \right).
    \end{align*}
    
    And it is clearly that $\gamma |\lambda| e^{d x} < 1$, $|\lambda| e^{d x} < \beta$ and
    \begin{align*}
        \beta - \gamma \lambda^2 e^{2 d x} > \begin{cases}
            \beta - \frac{1}{\gamma} > 0, ~~~&\text{ for } \beta\gamma > 1, \\
            \beta - \gamma \beta^2 > 0, ~~~&\text{ for } \beta\gamma < 1.
        \end{cases}
    \end{align*}
    
    Further, we also have $d y \le d k (x_0 - x) < \frac{\pi}{4}$ and $\cos(d y) - k \sin(d y) > 0$. 
    
    \vskip 10pt
    \textbf{(1)} If $\lambda = 0$, then we have
    \begin{align*}
        \tilde{G}(x, y, k, 0) = k(x_0 - \ln(\beta)) = G(x, k, 0).
    \end{align*}
    
    \textbf{(2)} For $\lambda (\beta\gamma - 1) < 0$, we have
    {\small\begin{align*}
        \frac{\tilde{G}(x, y, k, \lambda)}{\partial y} 
        = \frac{d \lambda  (\beta  \gamma -1) e^{d x} \left(\gamma  \lambda ^2 e^{2 d x} (\cos (d y)-k \sin (d y))+\beta  (k \sin (d y)+\cos (d y))+\lambda  (\beta  \gamma +1) e^{d x}\right)}{\left(\beta ^2+\lambda ^2 e^{2 d x}+2 \beta  \lambda  e^{d x} \cos (d y)\right) \left(\gamma ^2 \lambda ^2 e^{2 d x}+2 \gamma  \lambda  e^{d x} \cos (d y)+1\right)}.
    \end{align*}}
    
    We intend to prove that
    \begin{align*}
        \Psi(x, y) \triangleq \gamma \lambda ^2 e^{2 d x} (\cos (y)-k \sin (y))+\beta  (k \sin (y)+\cos (y))+\lambda  (\beta  \gamma +1) e^{d x} > 0,
    \end{align*}
    for $x \in [x_1, x_0]$ and $y \in [0, d k(x_0 - x)] \subseteq [0, \pi/4]$. 
    
    For $\beta \gamma < 1$ and $\lambda > 0$, $\Psi(x, y) > 0$ holds apparently. 
    
    For $\beta \gamma > 1$ and $\lambda < 0$, observe that 
    \begin{align*}
        \frac{\partial^2 \Psi(x, y)}{\partial y^2} = \gamma  \lambda ^2 e^{2 d x} (k \sin (y)-\cos (y))- \beta  (k \sin (y)+\cos (y)) < 0,
    \end{align*}
    which implies that $\Psi(x, y) \ge \min\{\Psi(x, 0), \Psi(x, d k(x_0 - x))\}$.
    Together with $\Psi(x, 0) = (\beta + \lambda e^{d x})(1 + \gamma \lambda e^{d x}) > 0$ and 
    \begin{align*}
        \psi(x) &\triangleq \gamma  \lambda  e^{d x} (\cos (d k (x_0-x))-k \sin (d k (x_0-x))) \\
        &\quad + \beta\lambda^{-1}  e^{-d x}  (k \sin (d k (x_0-x))+\cos (d k (x_0-x))) + (\beta  \gamma +1), \\
        \psi'(x) &= -d \left(k^2+1\right) \lambda^{-1} e^{-d x} \cos (d k (x_0-x)) \left(\beta - \gamma  \lambda ^2 e^{2 d x} \right) > 0, \\
        \psi(x) &\le \psi(x_0) = \lambda^{-1} e^{-d x_0} (\beta + \lambda e^{d x_0})(1 + \gamma \lambda e^{d x_0}) < 0, \\
        \Psi(x, d k(x_0 - x)) &= \lambda e^{d x} \psi(x) > 0.
    \end{align*}
    Therefore, we conclude that $\frac{\tilde{G}(x, y, k, \lambda)}{\partial y} > 0$ for $y \in [0, k (x_0- x)]$.
    
    \textbf{(3)} For $\lambda (\beta\gamma - 1) > 0$, we have
    {\small\begin{align*}
        \frac{\tilde{G}(x, y, k, \lambda)}{\partial y} 
        = -\frac{d \lambda  (\beta  \gamma - 1) e^{d x} \left(\gamma  \lambda ^2 e^{2 d x} (\cos (d y)+k \sin (d y))+\beta  (\cos (d y) - k \sin (d y))+\lambda  (\beta  \gamma +1) e^{d x}\right)}{\left(\beta ^2+\lambda ^2 e^{2 d x}+2 \beta  \lambda  e^{d x} \cos (d y)\right) \left(\gamma ^2 \lambda ^2 e^{2 d x}+2 \gamma  \lambda  e^{d x} \cos (d y)+1\right)}.
    \end{align*}}
    
    We need to prove that
    \begin{align*}
        \tilde\Psi(x, y) \triangleq \gamma \lambda ^2 e^{2 d x} (\cos (y)+k \sin (y))+\beta  (\cos (y) - k \sin (y))+\lambda  (\beta  \gamma +1) e^{d x} > 0,
    \end{align*}
    for $x \in [x_1, x_0]$ and $y \in [0, d k(x_0 - x)] \subseteq [0, \pi/4]$. 
    
    For $\beta \gamma > 1$ and $\lambda > 0$, $\Psi(x, y) > 0$ holds apparently. 
    
    For $\beta \gamma < 1$ and $\lambda < 0$, observe that 
    \begin{align*}
        \frac{\partial \tilde\Psi(x, y)}{\partial y} = -k \cos (y) \left(\beta -\gamma  \lambda ^2 e^{2 d x}\right)-\sin (y) \left(\beta +\gamma  \lambda ^2 e^{2 d x}\right) < 0, 
    \end{align*}
    which implies that $\Psi(x, y) \ge \Psi(x, d k (x_0 - x))$.
    Further, we have
    \begin{align*}
        \tilde\psi(x) &\triangleq \gamma  \lambda  e^{d x} (\cos (d k (x_0-x))+k \sin (d k (x_0-x))) \\
        &\quad + \beta\lambda^{-1}  e^{-d x}  (\cos (d k (x_0-x)) - k \sin (d k (x_0-x))) + (\beta  \gamma +1), \\
        \tilde\psi'(x) &= d \lambda^{-1} e^{-d x} \left(\left(k^2-1\right) \cos (d k (x_0-x)) \left(\beta -\gamma  \lambda ^2 e^{2 d x}\right)+2 k \sin (d k (x_0-x)) \left(\beta +\gamma  \lambda ^2 e^{2 d x}\right)\right), \\
        \hat\psi(x) &\triangleq \tan (d k (x_0-x)) - \frac{(1 - k^2) \left(\beta -\gamma  \lambda ^2 e^{2 d x}\right)}{2 k \left(\beta +\gamma  \lambda ^2 e^{2 d x}\right)} 
        \le 1 - \frac{(1 - k^2) (1 -\beta \gamma)}{2k (\beta\gamma + 1)} < 0,
    \end{align*}
    where the last inequality is according to $|\lambda| e^{d x} < \beta$ and $k < \frac{\sqrt{2(\beta^2\gamma^2 + 1)} - (\beta\gamma + 1) }{1 - \beta\gamma}$. 
    And thus we have $\tilde\psi'(x) > 0$ and $\tilde\Psi(x, d k(x_0 - x)) = \lambda e^{d x}\tilde\psi(x) \le \lambda e^{d x} \tilde\psi(x_0) > 0$. 
\end{proof}

Lemma \ref{lem:bounded-2-spin:tilde-G} eliminates the need to verify that $\phi(\lambda e^{d w}) \in U(x_0, x_1, k)$ for $w \in \{z: \re(z) = x_1, |\im(z)| \le k(x_0 - x_1)\}$, so that we can concentrate on proving $G(x, k, \lambda) > 0$. 
However, the proof of Lemma \ref{lem:bounded-2-spin:tilde-G} is too technique to reveal the key point that $g(w) \triangleq \phi(\lambda e^{d w})$ maps the boundary of $U(x_0, x_1, k)$ to the boundary of $g(U(x_0, x_1, k))$, which is intuitive and can be obtained directly based on the invariance of domain (cf. Lemma \ref{lem:boundary}). 

Following lemma is useful for verifying that $-1 \notin \phi^{-1}(U)$.
\begin{lemma}\label{lem:bounded-2-spin:-1}
    For $U \subseteq \{w \in \sC: |\im{w}| < \pi \}$, there holds
\begin{enumerate}
    \item $-1 \notin \phi^{-1}(U)$ for $\max\{\beta, \gamma\} \le 1$ or $\min\{\beta, \gamma\} \ge 1$, 
    \item for $\beta > 1 > \gamma$ or $\beta < 1 < \gamma$, $-1 \notin \phi^{-1}(U)$ is equivalent to $\ln\left(\frac{1 - \gamma}{\beta - 1}\right) \notin U$. 
\end{enumerate}
\end{lemma} 
\begin{proof}
    If $\phi^{-1}(w) = -1$, then $e^{w} = \frac{1 - \gamma}{\beta - 1}$. 
    Thus for $\beta > 1 > \gamma$ or $\beta < 1 < \gamma$, $-1 \notin \phi^{-1}(U)$ is equivalent to $\ln\left(\frac{1 - \gamma}{\beta - 1}\right) \notin U$.
    
    Moreover, with $e^{w} \in \sC \backslash (-\infty, 0]$, for $\max\{\beta, \gamma\} \le 1$ or $\min\{\beta, \gamma\} \ge 1$, $\frac{1 - \gamma}{\beta - 1} \in (-\infty, 0] \cup \{\infty\}$ and hence $\phi^{-1}(w) \neq -1$. 
    
\end{proof}

Condition 2 and Condition 3 in Lemma \ref{lem:bounded-2-spin:U} can be simplified as follows.
\begin{lemma}\label{lem:bounded-2-spin:lambda-e-d-x}
    For $\beta > 0, \gamma \ge 0$ and $k < \frac{\pi}{2 \Delta (x_0 - x_1)}$, suppose that 
    \begin{align*}
        1 + \gamma \lambda_0 e^{(\Delta - 1) x_0} > 0, ~~~ \beta + \lambda_0 e^{(\Delta - 1) x_0} > 0, ~~~ 1 + \lambda_0 e^{\Delta x_0} > 0.
    \end{align*}
    Then for all $w \in U(x_0, x_1, k)$, $0 \le d \le \Delta - 1$ and $\lambda \in [0, \lambda_0]$, we have 
    \begin{align*}
        \re(1 + \gamma \lambda e^{d w}) > 0, ~~~ \re(\beta + \lambda e^{d w}) > 0, ~~~ \re(1 + \lambda e^{\Delta w}) > 0.
    \end{align*}
\end{lemma}
\begin{proof}
    For $w = x + i y \in U(x_0, x_1, k)$, observe that $|d y| \le d k |x_0 - x_1| < \frac{\pi}{2}$ and
    \begin{align*}
        \re(1 + \gamma \lambda e^{d w}) = 1 + \gamma \lambda e^{d x} \cos(d y).
    \end{align*}
    For $\lambda \ge 0$, it is trivial that $1 + \gamma \lambda e^{d x} \cos(d y) > 0$. 
    And for $\lambda < 0$, we have
    \begin{align*}
        1 + \gamma \lambda e^{d x} \cos(d y) \ge 1 + \gamma \lambda_0 e^{(\Delta - 1) x_0} > 0,
    \end{align*}
    where we have recalled that $x_0 \ge 0$. 
    
    Further, $\re(\beta + \lambda e^{d w}) > 0$ and $\re(1 + \lambda e^{\Delta w}) > 0$ hold for the same reason.
    
\end{proof}

Therefore, combining Lemma \ref{lem:bounded-2-spin:-1}, Lemma \ref{lem:bounded-2-spin:lambda-e-d-x}, Lemma \ref{lem:bounded-2-spin:tilde-G} and Lemma \ref{lem:bounded-2-spin:H}  to check $U = \{z \in U(x_0, x_1, k): \re(z) \in [x_2, x_3]\}$ satisfying all conditions in Lemma \ref{lem:bounded-2-spin:U}, we just need to verify that
\begin{enumerate}
    \item $-\ln(\beta) \in U$ and $\ln(\gamma) \in U$ if $\gamma > 0$,
    \item $\ln\left( \frac{1 - \gamma}{\beta - 1} \right) \notin U$ if $(1 - \gamma)(\beta - 1) > 0$,
    \item $1 + \gamma \lambda_0 e^{(\Delta - 1) x_0} > 0, \beta + \lambda_0 e^{(\Delta - 1) x_0} > 0, 1 + \lambda_0 e^{\Delta x_0} > 0$, 
    \item $r(d x, d y, \lambda) \in (x_2, x_3)$ for $d = \Delta - 1$, $x \in [x_1, x_0], y \in [0, k (x_0 - x)]$ and $\lambda \in [0, \lambda_0]$,
    \item $H(x, \lambda) > 0$ for $x \in [x_1, x_0]$ and $\lambda \in [0, \lambda_0]$. 
\end{enumerate}

\subsection{The case $\beta\gamma > 1$ and $\lambda_0 > 0$}
Let $d = \Delta -1$. 
\begin{lemma}\label{lem:ferro-positive:range-r}
    For $\beta\gamma > 1$ and $\lambda > 0$, there holds $r(d x, d y, \lambda) \in (-\ln(\beta), \ln(\gamma))$ for $x + \i y \in U(x_0, x_1, k)$.
\end{lemma}
\begin{proof}
    By $k < \frac{\pi}{2 d (x_0 - x_1)}$, we know that $|d y| \le d k (x_0 - x_1) < \frac{\pi}{2}, \cos(d y) > 0$. Hence, there holds
    \begin{align*}
        \frac{1 + \gamma^2 \lambda^2 e^{2 d x} + 2 \gamma \lambda e^{d x} \cos(d y)}{\beta^2 + \lambda^2 e^{2 d x} + 2 \beta \lambda e^{d x} \cos(d y)} - \gamma^2 
        &= -\frac{(\beta  \gamma -1) \left(\beta  \gamma + 1 +2 \gamma  \lambda  e^{d x} \cos (d y)\right)}{\beta ^2+\lambda ^2 e^{2 d x}+2 \beta  \lambda  e^{d x} \cos (d y)} < 0, \\
        \frac{1 + \gamma^2 \lambda^2 e^{2 d x} + 2 \gamma \lambda e^{d x} \cos(d y)}{\beta^2 + \lambda^2 e^{2 d x} + 2 \beta \lambda e^{d x} \cos(d y)} - \frac{1}{\beta^2}
        &= \frac{ (\beta  \gamma -1) \lambda e^{d x} \left( (\beta  \gamma +1) \lambda e^{d x}+2 \beta  \cos (d y)\right)}{\beta ^2 \left(\beta ^2+\lambda ^2 e^{2 d x}+2 \beta  \lambda  e^{d x} \cos (d y)\right)} > 0.
    \end{align*}
\end{proof}

\begin{lemma}\label{lem:ferro-positive:H}
    For $\beta\gamma > 1$ and $0 < \lambda_0 < \left(\frac{\beta} {\gamma}\right)^{\frac{\sqrt{\beta\gamma}}{\sqrt{\beta\gamma} - 1}} \max\{1, \gamma\}^{\frac{\sqrt{\beta \gamma} + 1}{\sqrt{\beta\gamma} - 1} - d}$, with $x_0 = \max\{0, \ln(\gamma)\}$, there holds $H(x, \lambda) > 0$ for $x \in (-\infty, x_0]$ and $\lambda \in [0, \lambda_0]$. 
\end{lemma}
\begin{proof}
    Recalling the definition of $H(x, \lambda)$, for $\beta\gamma > 1$ and $\lambda \ge 0$, we have
    \begin{align*}
        H(x, \lambda) = x_0 - \ln \left(\frac{\gamma \lambda e^{d x} + 1}{\beta + \lambda  e^{d x}}\right) - \frac{d \lambda  (\beta  \gamma -1) e^{d x} (x_0-x)}{(\gamma \lambda e^{d x} + 1)(\beta + \lambda  e^{d x})}.
    \end{align*}
    
    Note that 
    \begin{align*}
        \frac{\partial H(x, \lambda)}{\partial x} = -\frac{d^2 (\beta \gamma -1) \lambda e^{d x} (x_0 - x) \left(\beta -\gamma  \lambda ^2 e^{2 d x} \right)}{\left(\gamma \lambda e^{d x} + 1\right)^2 \left(\beta +\lambda e^{d x} \right)^2 }.
    \end{align*}
    
    If $\lambda e^{d x_0} \le \sqrt{\frac{\beta}{\gamma}}$, then $\frac{\partial H(x, \lambda)}{\partial x} \le 0$ holds for $x \in [x_1, x_0]$ and consequently we have
    \begin{align*}
        H(x, \lambda) \ge H(x_0, \lambda) = x_0 - \ln \left(\frac{\gamma \lambda e^{d x_0} + 1}{\beta + \lambda  e^{d x_0}}\right) > 0
    \end{align*}
    where we have recalled that $r(d x_0, 0, \lambda) < \ln(\gamma)$ referred to Lemma \ref{lem:ferro-positive:range-r}. 
    
    On the other hand, if $\lambda e^{d x_0} > \sqrt{\frac{\beta}{\gamma}}$, 
    then there exists $\tilde{x}_{\lambda} \in (-\infty, x_0)$ such that $\lambda e^{d \tilde{x}_{\lambda}} = \sqrt{\frac{\beta}{\gamma}}$, and $H(x, \lambda)$ is decreasing on $(-\infty, \tilde{x}_{\lambda}]$ and is increasing on $[\tilde{x}_{\lambda}, x_0]$. Therefore, we have
    \begin{align*}
        H(x, \lambda) &\ge H(\tilde{x}_{\lambda}, \lambda) \\
        &= x_0 + \frac{1}{2}\ln(\beta/\gamma) - \frac{d(\beta\gamma -1) (x_0 - \tilde{x}_{\lambda})}{(\sqrt{\beta\gamma} + 1)^2} \\
        &= \frac{\sqrt{\beta\gamma}}{\sqrt{\beta\gamma} + 1} \ln(\beta/\gamma) -  \frac{\sqrt{\beta\gamma} - 1}{\sqrt{\beta\gamma} + 1} \ln(\lambda) - \left( \frac{d(\sqrt{\beta\gamma} - 1)}{\sqrt{\beta\gamma} + 1} - 1 \right) x_0 \\
        &\ge \frac{\sqrt{\beta\gamma} - 1}{\sqrt{\beta\gamma} + 1}\left( \frac{\sqrt{\beta\gamma}}{\sqrt{\beta\gamma} - 1} \ln(\beta/\gamma) + \left(\frac{\sqrt{\beta\gamma} + 1}{\sqrt{\beta\gamma} - 1} - d \right)\max\{0, \ln(\gamma)\} - \ln(\lambda_0) \right) > 0.
    \end{align*}
\end{proof}

\begin{corollary}\label{coro:ferro-positive}
    For $\sqrt{\beta\gamma} > 1$ and $0 < \lambda_0 < \left(\frac{\beta} {\gamma}\right)^{\frac{\sqrt{\beta\gamma}}{\sqrt{\beta\gamma} - 1}} \max\{1, \gamma\}^{\frac{\sqrt{\beta \gamma} + 1}{\sqrt{\beta\gamma} - 1} - d}$, with setting $x_0 = \max\{0, \ln(\gamma)\}, x_1 = \min\{0, - \ln(\beta)\}$, there exists $k > 0$ such that 
    \begin{align*}
        \phi(\lambda e^{d w}) \in \open{U(x_0, x_1, k)}, ~~~\forall w \in U(x_0, x_1, k).
    \end{align*}
    
    Furthermore, $U = U(x_0, x_1, k)$ satisfies all conditions in Lemma \ref{lem:bounded-2-spin:U}.
\end{corollary}

\begin{proof}
    If $\beta > 1 > \gamma$, then we have
    $
        \ln\left(\frac{1 - \gamma}{\beta - 1}\right) < -\ln(\beta) = x_1, 
    $
    and thus $\ln\left(\frac{1 - \gamma}{\beta - 1}\right) \notin U$. 
    
    If $\beta < 1 < \gamma$, then there holds
    $
        \ln\left(\frac{\gamma - 1}{1 - \beta}\right) > \ln(\gamma) = x_0.
    $
    and hence $\ln\left(\frac{1 - \gamma}{\beta - 1}\right) \notin U$. 
    
    And it is trivial that $1 + \gamma \lambda_0 e^{d x_0} > 0, \beta + \lambda_0 e^{d x_0} > 0, 1 + \lambda_0 e^{(d + 1) x_0} > 0$.
    
    Then with Lemma \ref{lem:ferro-positive:range-r} and Lemma \ref{lem:ferro-positive:H}, we know that $U$ satisfies all conditions in Lemma \ref{lem:bounded-2-spin:U}. 
    
\end{proof}

\begin{figure}[H]
\centering
\includegraphics[width=\textwidth]{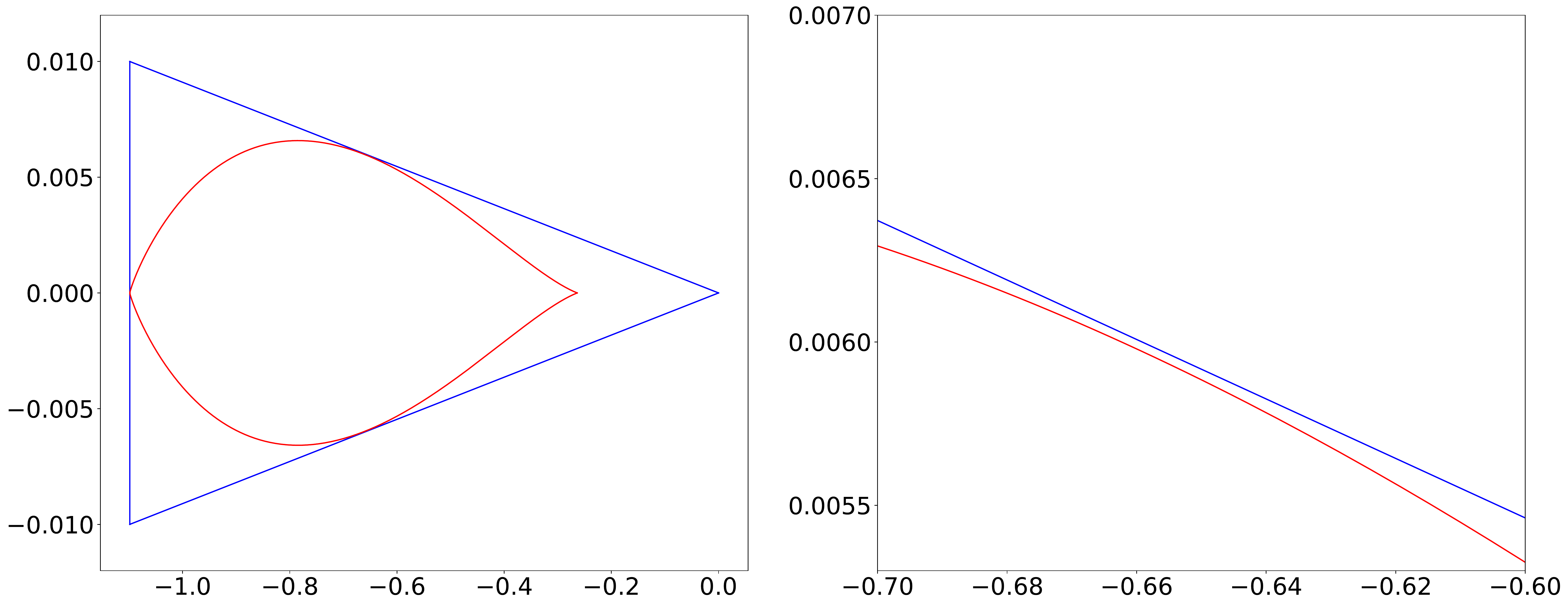}
\caption{Blue lines: the boundary of $U(0, -\ln(\beta), 0.01)$. Red curve: the boundary of $\{\phi(\lambda e^{d w}): w \in U(0, -\ln(\beta), 0.01)\}$ where $\beta = 3, \gamma = 0.8, d = 10$ and $\lambda = 41 < \lambda_c \approx 41.6$. The figure on the right enlarges part of the figure on the left to clarify the containment relationship between the two regions clearly. }
\end{figure}

\begin{figure}[H]
\centering
\includegraphics[width=\textwidth]{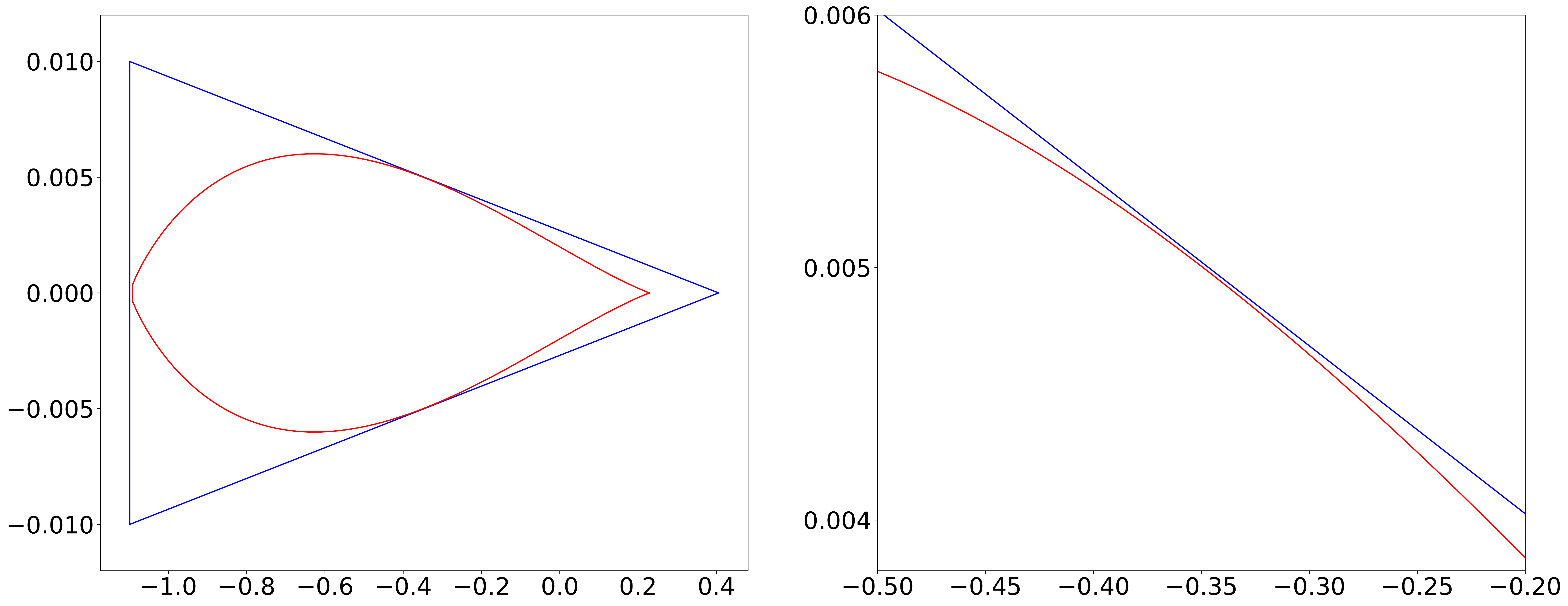}
\caption{Blue lines: the boundary of $U(\ln(\gamma), -\ln(\beta), 0.01)$. Red curve: the boundary of $\{\phi(\lambda e^{d w}): w \in U(\ln(\gamma), -\ln(\beta), 0.01)\}$ where $\beta = 3, \gamma = 1.5, d = 5$ and $\lambda = 1.5 < \lambda_c \approx 1.51$. The figure on the right enlarges part of the figure on the left to clarify the containment relationship between the two regions clearly. }
\end{figure}

\subsection{The case $\beta\gamma > 1$ and $\lambda_0 < 0$}
\begin{lemma}\label{lem:ferro-negative:range-r}
    For $\beta\gamma > 1$ and $-\max\{1, \gamma\}^{-(d + 1)} < \lambda_0 < 0$,
    with setting $x_0 = \max\{0, \ln(\gamma)\}$, there holds $r(d x, d y, \lambda) \in \left[\ln\left(\frac{1 + \gamma \lambda_0 e^{d x_0}}{\beta + \lambda_0 e^{d x_0}}\right), \ln(\gamma)\right)$ for $x + \i y \in U(x_0, x_1, k)$ and $\lambda \in [0, \lambda_0]$.
\end{lemma}
\begin{proof}
    First note that 
    \begin{align*}
        |\lambda| e^{d x} \le |\lambda_0| e^{d x_0} < \begin{cases}
            1 \le \frac{1}{\gamma}, ~~~&\text{ for } \gamma \le 1, \\
            \frac{1}{\gamma^{d+1}} \gamma^d = \frac{1}{\gamma}, ~~~&\text{ for } \gamma > 1.
        \end{cases}
    \end{align*}
    By $k < \frac{\pi}{2 d (x_0 - x_1)}$, we know that $|d y| \le d k (x_0 - x_1) < \frac{\pi}{2}, \cos(d y) > 0$. \\
    Hence, there holds $\beta\gamma + 1 + 2\gamma \lambda e^{d x} \cos(d y) > \beta\gamma - 1 > 0$ and
    \begin{align*}
        \frac{1 + \gamma^2 \lambda^2 e^{2 d x} + 2 \gamma \lambda e^{d x} \cos(d y)}{\beta^2 + \lambda^2 e^{2 d x} + 2 \beta \lambda e^{d x} \cos(d y)} - \gamma^2 
        &= -\frac{(\beta  \gamma -1) \left(\beta  \gamma + 1 +2 \gamma  \lambda  e^{d x} \cos (d y)\right)}{\beta ^2+\lambda ^2 e^{2 d x}+2 \beta  \lambda  e^{d x} \cos (d y)} < 0.
    \end{align*}
    On the other hand, we have $\beta - \gamma \lambda^2 e^{2d x} > \beta - 1/\gamma > 0$ and
    \begin{align*}
        &\quad \frac{1 + \gamma^2 \lambda^2 e^{2 d x} + 2 \gamma \lambda e^{d x} \cos(d y)}{\beta^2 + \lambda^2 e^{2 d x} + 2 \beta \lambda e^{d x} \cos(d y)} \\
        &= \frac{\beta}{\gamma} -\frac{(\beta  \gamma -1) \left(\beta -\gamma  \lambda ^2 e^{2 d x}\right)}{\beta  \left(\beta ^2+\lambda ^2 e^{2 d x}+2 \beta  \lambda  e^{d x} \cos (d y)\right)} \\
        &\ge \frac{\beta}{\gamma} -\frac{(\beta  \gamma -1) \left(\beta -\gamma  \lambda ^2 e^{2 d x}\right)}{\beta  \left(\beta ^2+\lambda ^2 e^{2 d x} + 2 \beta  \lambda  e^{d x} \right)} \\
        &= \left(\frac{1 + \gamma \lambda e^{d x}}{\beta + \lambda e^{d x}}\right)^2 
        \ge \left(\frac{1 + \gamma \lambda_0 e^{d x_0}}{\beta + \lambda_0 e^{d x_0}}\right)^2,
    \end{align*}
    where we have used that the function $\psi(s) = \frac{1 + \gamma s}{\beta + s}$ is strictly increasing on $(-1/\gamma, +\infty)$. 
\end{proof}

\begin{lemma}\label{lem:ferro-negative:H}
    For $\beta\gamma > 1$ and $-\max\{1, \gamma\}^{-(d + 1)} < \lambda_0 < 0$, with $x_0 = \max\{0, \ln(\gamma)\}$, there holds $H(x, \lambda) > 0$ for $x \in (-\infty, x_0]$ and $\lambda \in [\lambda_0, 0]$.
\end{lemma}
\begin{proof}
    Recalling the definition of $H(x, \lambda)$, for $\beta\gamma > 1$ and $\lambda \le 0$, we have
    \begin{align*}
        H(x, \lambda) = x_0 - \ln \left(\frac{\gamma \lambda e^{d x} + 1}{\beta + \lambda  e^{d x}}\right) + \frac{d \lambda  (\beta  \gamma -1) e^{d x} (x_0-x)}{(\gamma \lambda e^{d x} + 1)(\beta + \lambda  e^{d x})}.
    \end{align*}
    It is clearly that $H(x, 0) = x_0 + \ln(\beta) \ge \ln(\beta \gamma) > 0$. Now we assume that $\lambda < 0$. 
    Note that 
    \begin{align*}
        \frac{\partial H(x, \lambda)}{\partial x} 
        = \frac{d \lambda  (\beta  \gamma -1) e^{d x} \left(d(x_0 - x) (\beta - \gamma \lambda^2 e^{2 d x}) -2(\gamma \lambda e^{d x} + 1)(\lambda e^{d x} + \beta) \right)}{\left(\beta +\lambda  e^{d x}\right)^2 \left(\gamma  \lambda  e^{d x}+1\right)^2}.
    \end{align*}
    
    Consider the function 
    \begin{align*}
        \Psi(x) = d(x_0 - x) (\beta - \gamma \lambda^2 e^{2 d x}) - 2(\gamma \lambda e^{d x} + 1)(\lambda e^{d x} + \beta).
    \end{align*}
    It is easily to check that  
    \begin{align*}
        \Psi'(x) &= - d \left(\beta +2 \lambda  (\beta  \gamma +1) e^{d x}+\gamma  \lambda ^2 e^{2 d x} (2 d (x_0 - x)+3)\right), \\
        \Psi''(x) &= -2 d^2 \lambda  e^{d x} \left((\beta  \gamma +1) +2 \gamma  (d (x_0 - x)+1) \lambda e^{d x} \right), \\
        \psi(x) &\triangleq (\beta  \gamma +1) +2 \gamma  (d (x_0 - x)+1) \lambda e^{d x}, \\
        \psi'(x) &= 2 \gamma  d^2 \lambda  e^{d x} (x_0 - x) \le 0 ~~~ \text{ for } x \in (-\infty, x_0].
    \end{align*}
    Therefore, we can conclude that 
    \begin{align*}
        \psi(x) &\ge \psi(x_0), \\
        \Psi''(x) &\ge 2 d^2 |\lambda|  e^{d x} \psi(x_0) = 2 d^2 |\lambda|  e^{d x} \left((\beta  \gamma +1) +2 \gamma \lambda e^{d x_0} \right) > 0, 
    \end{align*}
    where the last inequality is according to $|\lambda_0| e^{d x_0} < \frac{1}{\gamma}$ and $\beta\gamma > 1$. 
    Hence, $\Psi(x)$ is strictly convex on $(-\infty, x_0]$. 
    
    Next based on $|\lambda_0| e^{d x_0} < \frac{1}{\gamma}$, 
    we have
    \begin{align*}
        \Psi(x_0) &= -2 (\beta + \lambda e^{d x_0})(1 + \gamma \lambda e^{d x_0}) < 0, \\
        \Psi(x) &\to +\infty, ~~~\text{ for } x \to -\infty.
    \end{align*}
    Together with strictly convexity of $\Psi$ on $(-\infty, x_0]$, there exists unique $\tilde{x}_{\lambda} \in (-\infty, x_0)$ such that $\Psi(\tilde{x}_{\lambda}) = 0$, i.e.,
    \begin{align*}
        d(x_0 - \tilde{x}_{\lambda}) (\beta - \gamma \lambda^2 e^{2 d \tilde{x}_{\lambda}}) = 2(\gamma \lambda e^{d \tilde{x}_{\lambda}} + 1)(\lambda e^{d \tilde{x}_{\lambda}} + \beta)
    \end{align*}
    
    Furthermore, there holds $\frac{\partial H(x, \lambda)}{\partial x} \le 0$ for $x \in (-\infty, \tilde{x}_{\lambda}]$ and $\frac{\partial H(x, \lambda)}{\partial x} \ge 0$ for $x \in [\tilde{x}_{\lambda}, x_0]$, which implies that
    \begin{align*}
        H(x, \lambda) &\ge H(\tilde{x}_{\lambda}, \lambda) \\
        &= x_0 - \ln \left(\frac{\gamma \lambda e^{d \tilde{x}_{\lambda}} + 1}{\beta + \lambda  e^{d \tilde{x}_{\lambda}}}\right) + \frac{d \lambda  (\beta  \gamma -1) e^{d \tilde{x}_{\lambda}} (x_0-\tilde{x}_{\lambda})}{(\gamma \lambda e^{d \tilde{x}_{\lambda}} + 1)(\beta + \lambda  e^{d \tilde{x}_{\lambda}})} \\
        &= x_0 - \ln \left(\frac{\gamma \lambda e^{d \tilde{x}_{\lambda}} + 1}{\beta + \lambda  e^{d \tilde{x}_{\lambda}}}\right) + \frac{2 \lambda  (\beta  \gamma -1) e^{d \tilde{x}_{\lambda}}}{\beta - \gamma \lambda^2 e^{2 d \tilde{x}_{\lambda}}}.
    \end{align*}
    
    
    Here let 
    \[\varphi(s) = \frac{2(\beta\gamma - 1) s}{\beta - \gamma s^2} - \ln \left(\frac{\gamma s + 1}{\beta + s}\right).\]
    Then there holds
    \begin{align*}
        \varphi'(s) = \frac{(\beta  \gamma -1) (\beta + 2\beta \gamma s + \gamma s^2) (\beta + 2 s + \gamma s^2)}{(\beta +s) (\gamma  s+1) \left(\beta -\gamma  s^2\right)^2}.
    \end{align*}
    Letting $s_0 \triangleq -\beta + \sqrt{(\beta\gamma - 1) \beta/\gamma}$, it is easily to check that $\beta + 2 s + \gamma s^2 > 0$ for $s \in \sR$ and $\beta + 2 \beta\gamma s + \gamma s^2 \le 0$ for $s \in [-1/\gamma, s_0]$, $\beta + 2 \beta\gamma s + \gamma s^2 \ge 0$ for $s \in [ s_0, 0]$. 
    Hence $\varphi$ is decreasing on $[-1/\gamma, s_0]$ and is increasing on $[s_0, 0]$.
    Together with $\lambda e^{d \tilde{x}_{\lambda}} \in (-1/\gamma, 0]$, thus we have
    \begin{align*}
        H(x, \lambda) \ge x_0 + \varphi(\lambda e^{d \tilde{x}_{\lambda}}) \ge x_0 + \varphi(s_0) = x_0 -\sqrt{\frac{\beta\gamma - 1}{\beta\gamma}} - \ln\left( \gamma\left(1 - \sqrt{\frac{\beta\gamma - 1}{\beta\gamma}}\right) \right) > x_0 - \ln(\gamma) \ge 0,
    \end{align*}
    where we have applied that $\ln(1 - t) < -t$ for $t > 0$.
    
\end{proof}

\begin{corollary}\label{coro:ferro-negative}
    For $\beta\gamma > 1$ and $-\max\{1, \gamma\}^{-(d + 1)} < \lambda_0 < 0$, with setting $x_0 = \max\{0, \ln(\gamma)\}$ and $x_1 \le 0$ such that 
    \begin{align*}
        x_1 < \ln\left(\frac{1 + \gamma \lambda_0 e^{d x_0}}{\beta + \lambda_0 e^{d x_0}}\right) ~~~\text{ and } ~~~ x_1 > \ln\left(\frac{1 - \gamma}{\beta - 1}\right) ~\text{ if }~ \gamma < 1,
    \end{align*} there exists $k > 0$ such that 
    \begin{align*}
        \phi(\lambda e^{d w}) \in \open{U(x_0, x_1, k)}, ~~~\forall w \in U(x_0, x_1, k).
    \end{align*}
    
    Furthermore, $U = U(x_0, x_1, k)$ satisfies all conditions in Lemma \ref{lem:bounded-2-spin:U}.
\end{corollary}

\begin{proof}
    If $\gamma < 1$, then we have $\beta > 1$ according to $\beta\gamma > 1$.
    Moreover there holds $\frac{1 - \gamma}{\beta - 1} < 1$,
    $\lambda_0 e^{d x_0} = \lambda_0 > -1$ and
    $
        \frac{1 - \gamma}{\beta - 1} < \frac{1 - \gamma |\lambda_0|}{\beta - |\lambda_0|}.
    $
    Hence there exists $x_1 \le 0$ satisfies 
    \begin{align*}
         \ln\left(\frac{1 - \gamma}{\beta - 1}\right) < x_1 < \ln\left(\frac{1 + \gamma \lambda_0 e^{d x_0}}{\beta + \lambda_0 e^{d x_0}}\right).
    \end{align*}
    And it is clearly that $\ln\left(\frac{1 - \gamma}{\beta - 1}\right)\notin U$. 
    
    If $\beta < 1 < \gamma$, then there holds
    $
        \ln\left(\frac{\gamma - 1}{1 - \beta}\right) > \ln(\gamma),
    $
    and hence $\ln\left(\frac{1 - \gamma}{\beta - 1}\right) \notin U$. 
    
    Further, by $|\lambda_0| e^{d x_0} < \frac{1}{\gamma}$, we have $1 + \gamma \lambda_0 e^{d x_0} > 0, \beta + \lambda_0 e^{d x_0} > 0, 1 + \lambda_0 e^{(d + 1) x_0} > 0$.
    
    Then with Lemma \ref{lem:ferro-negative:range-r} and Lemma \ref{lem:ferro-negative:H}, we know that $U$ satisfies all conditions in Lemma \ref{lem:bounded-2-spin:U}. 
    
\end{proof}

\begin{figure}[H]
\centering
\includegraphics[width=\textwidth]{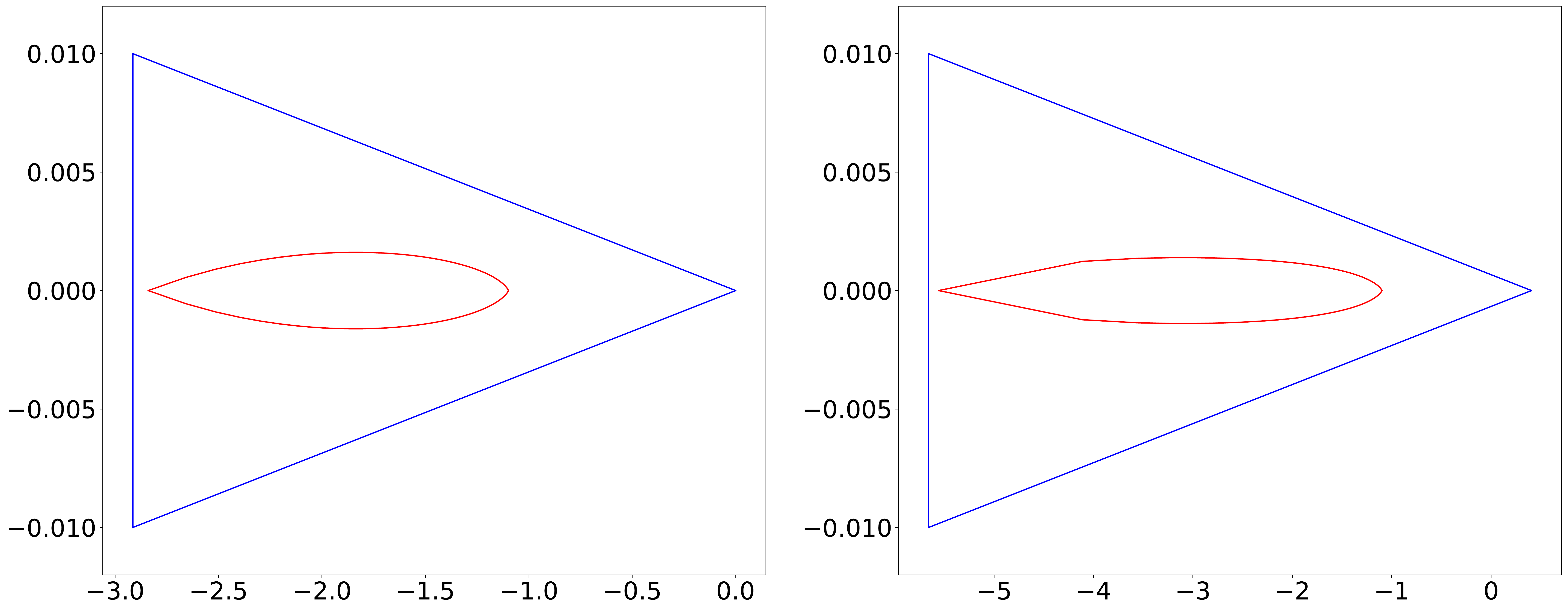}
\caption{Blue lines: the boundary of $U(x_0, x_1, 0.01)$. Red curve: the boundary of $\{\phi(\lambda e^{d w}): w \in U(x_0, x_1, 0.01)\}$. Left: $\beta = 3, \gamma = 0.9, d = 10$ and $\lambda = -0.98 > \lambda_c = -1$. Right: $\beta = 3, \gamma = 1.5, d = 5$ and $\lambda = -0.087 > \lambda_c \approx -0.0878$.}
\end{figure}

\subsection{The case $\beta \gamma < 1$ and $\lambda_0 > 0$}
See Appendix \ref{app:hat-x} for properties of $\hat{x}_d$ and $\lambda_c(d)$. 

\begin{lemma}\label{lem:anti-ferro-positive:range-r}
    For $\beta\gamma < 1$ and $\lambda_0 > 0$, then for $x + \i y \in U(x_0, x_1, k)$ and $\lambda \in [0, \lambda_0]$, there holds 
    \begin{align*}
        r(d x, d y, \lambda) \in \begin{cases} 
            (\ln(\gamma), -\ln(\beta)), ~~~&\text{ for } \gamma > 0, \\
            \left(-\ln(\beta + \lambda_0 e^{d x_0}), -\ln(\beta)\right), ~~~&\text{ for } \gamma = 0.
        \end{cases}
    \end{align*}
\end{lemma}
\begin{proof}
    By $k < \frac{\pi}{2 d (x_0 - x_1)}$, we know that $|d y| \le d k (x_0 - x_1) < \frac{\pi}{2}, \cos(d y) > 0$. Hence, there holds
    \begin{align*}
        \frac{1 + \gamma^2 \lambda^2 e^{2 d x} + 2 \gamma \lambda e^{d x} \cos(d y)}{\beta^2 + \lambda^2 e^{2 d x} + 2 \beta \lambda e^{d x} \cos(d y)} - \gamma^2 
        &= \frac{(1 - \beta  \gamma) \left(\beta  \gamma + 1 +2 \gamma  \lambda  e^{d x} \cos (d y)\right)}{\beta ^2+\lambda ^2 e^{2 d x}+2 \beta  \lambda  e^{d x} \cos (d y)} > 0, \\
        \frac{1 + \gamma^2 \lambda^2 e^{2 d x} + 2 \gamma \lambda e^{d x} \cos(d y)}{\beta^2 + \lambda^2 e^{2 d x} + 2 \beta \lambda e^{d x} \cos(d y)} - \frac{1}{\beta^2}
        &= -\frac{ (1 - \beta  \gamma) \lambda e^{d x} \left( (\beta  \gamma +1) \lambda e^{d x}+2 \beta  \cos (d y)\right)}{\beta ^2 \left(\beta ^2+\lambda ^2 e^{2 d x}+2 \beta  \lambda  e^{d x} \cos (d y)\right)} < 0.
    \end{align*}
    Further for $\gamma = 0$, we have
    \begin{align*}
        \frac{1 + \gamma^2 \lambda^2 e^{2 d x} + 2 \gamma \lambda e^{d x} \cos(d y)}{\beta^2 + \lambda^2 e^{2 d x} + 2 \beta \lambda e^{d x} \cos(d y)} = \frac{1}{\beta^2 + \lambda^2 e^{2 d x} + 2 \beta \lambda e^{d x} \cos(d y)} \ge \frac{1}{(\beta + \lambda_0 e^{d x_0})^2}.
    \end{align*}
\end{proof}

\begin{lemma}\label{lem:anti-ferro-positive:H}
    Suppose that $\sqrt{\beta\gamma} \le \frac{d-1}{d+1}$ and $\lambda_0 > 0$ such that
    \begin{align*}
        \lambda_0 < \begin{cases}
            \lambda_c(d_c), ~~~&\text{ for } \beta > 1, d > d_c, \\
            \lambda_c(d), ~~~&\text{ otherwise. }
        \end{cases}
    \end{align*}
    With setting
    \begin{align*}
        x_0 = \max\left\{\frac{2(1 - \beta\gamma)\bar{x}}{\beta - \gamma \bar{x}^2} + \ln\left(\frac{\gamma \bar{x} + 1}{\bar{x} + \beta}\right), 0\right\},
    \end{align*}
    where $\bar{x}$ satisfies $\lambda_0 \left(\frac{\gamma \bar{x} + 1}{\bar{x} + \beta}\right)^d = \bar{x}$,
    there holds $x_0 \ge -\ln(\beta)$ and $H(x, \lambda) > 0$ for $x \in (-\infty, x_0]$ and $\lambda \in [0, \lambda_0]$. 
\end{lemma}
\begin{proof}
    Recalling the definition of $H(x, \lambda)$, for $\sqrt{\beta\gamma} \le \frac{d+1}{d-1}$ and $\lambda \ge 0$, we have
    \begin{align*}
        H(x, \lambda) &= x_0 - \ln \left(\frac{\gamma \lambda e^{d x} + 1}{\beta + \lambda  e^{d x}}\right) - \frac{d \lambda  (1 - \beta  \gamma) e^{d x} (x_0-x)}{(\gamma \lambda e^{d x} + 1)(\beta + \lambda  e^{d x})}, \\
        \frac{\partial H(x, \lambda)}{\partial x} 
        &= -\frac{d \lambda  (1 - \beta  \gamma) e^{d x} \left(d(x_0 - x) (\beta - \gamma \lambda^2 e^{2 d x}) -2(\gamma \lambda e^{d x} + 1)(\lambda e^{d x} + \beta) \right)}{\left(\beta +\lambda  e^{d x}\right)^2 \left(\gamma  \lambda  e^{d x}+1\right)^2}.
    \end{align*}
    
    With letting
    \begin{align*}
        \Psi(x) = d(x_0 - x) (\beta - \gamma \lambda^2 e^{2 d x}) - 2(\gamma \lambda e^{d x} + 1)(\lambda e^{d x} + \beta),
    \end{align*}
    it is easily to check that 
    \begin{align*}
        \Psi'(x) &= - d \left(\beta +2 \lambda  (\beta  \gamma +1) e^{d x}+\gamma  \lambda ^2 e^{2 d x} (2 d (x_0 - x)+3)\right) < 0, 
    \end{align*}
    which implies that $\Psi(x)$ is a decreasing function on $(-\infty, x_0]$.
    
    
    Note that $\Psi(x_0) < 0$ and $\Psi(x) \to +\infty$ as long as $x \to -\infty$.
    Hence there exists $\tilde{x}_{\lambda} \in (-\infty, x_0)$ such that $\Psi(\tilde{x}_{\lambda}) = 0$. Moreover, $\frac{\partial H(x, \lambda)}{\partial x} \le 0$ holds for $x \in (-\infty, \tilde{x}_{\lambda}]$ and $\frac{\partial H(x, \lambda)}{\partial x} \ge 0$ holds for $x \in [\tilde{x}_{\lambda}, x_0]$.

    
    
    Hence, we have
    \begin{align*}
        H(x, \lambda) &\ge H(\tilde{x}_{\lambda}, \lambda) \\
        &= x_0 - \ln\left(\frac{\gamma \lambda e^{d \tilde{x}_{\lambda}} + 1}{\beta + \lambda  e^{d \tilde{x}_{\lambda}}}\right) - \frac{d \lambda  (1 - \beta  \gamma) e^{d \tilde{x}_{\lambda}} (x_0-\tilde{x}_{\lambda})}{(\gamma \lambda e^{d \tilde{x}_{\lambda}} + 1)(\beta + \lambda  e^{d \tilde{x}_{\lambda}})} \\
        &= x_0 - \ln\left(\frac{\gamma \lambda e^{d \tilde{x}_{\lambda}} + 1}{\beta + \lambda  e^{d \tilde{x}_{\lambda}}}\right) - \frac{2 (1 - \beta\gamma) \lambda e^{d \tilde{x}_{\lambda}}}{\beta - \gamma \lambda^2 e^{2 d \tilde{x}_{\lambda}}} \\
        &= \max\{\varphi(\bar{x}), 0\} - \varphi(\lambda e^{d \tilde{x}_{\lambda}}),
    \end{align*}
    where 
    \begin{align*}
        \varphi(s) \triangleq \frac{2(1 - \beta\gamma)s}{\beta - \gamma s^2} + \ln\left(\frac{\gamma s + 1}{s + \beta}\right).
    \end{align*}
    
    Observe that 
    \begin{align*}
        \varphi'(s) = \frac{(1 - \beta  \gamma) (\beta + 2\beta\gamma s + \gamma s^2) (\beta + 2 s + \gamma s^2)}{(\beta + s) (\gamma s + 1) \left(\beta -\gamma  s^2\right)^2} > 0.
    \end{align*}
    And for $\gamma > 0$, we have $\bar{x} < \hat{x}_d \le \sqrt{\beta/\gamma}$ according to Lemma \ref{lem:hat-x:bar-x}. 
    Hence there holds $\varphi(\bar{x}) \ge \varphi(0) = -\ln(\beta)$. 
    
    For $\beta \le 1$, we always have $\varphi(\bar{x}) \ge 0$. And in order to verify that $H(x, \lambda) > 0$, it is enough to check that
    $
        \lambda e^{d \tilde{x}_{\lambda}} < \bar{x}.
    $
    Denoting $\tilde{s} = \lambda e^{d \tilde{x}_{\lambda}}$, we have
    \begin{align*}
        \frac{2(\beta + \tilde{s})(1 + \gamma \tilde{s})}{\beta - \gamma \tilde{s}^2} &= d(x_0 - \tilde{x}_{\lambda}), \\
        \tilde{s} \exp\left( \frac{2(\beta + \tilde{s})(1 + \gamma \tilde{s})}{\beta - \gamma \tilde{s}^2} \right) &= \lambda e^{d x_0} \le \lambda_0 e^{d \varphi(\bar{x})} \\
        &= \lambda_0 \exp\left(\frac{2d(1-\beta\gamma)\bar{x}}{\beta - \gamma \bar{x}^2}\right) \left(\frac{\gamma \bar{x} + 1}{\bar{x} + \beta}\right)^d \\
        &= \bar{x} \exp\left(\frac{2d(1-\beta\gamma)\bar{x}}{\beta - \gamma \bar{x}^2}\right) \\
        &< \bar{x} \exp\left(\frac{2(\beta + \bar{x})(1 + \gamma \bar{x})}{\beta - \gamma \bar{x}^2}\right),
    \end{align*}
    where we have recalled that $\frac{d(1 - \beta\gamma)\bar{x}}{(\beta + \bar{x})(1 + \gamma \bar{x})} < 1$ due to Lemma \ref{lem:hat-x:bar-x}. 
    Note that
    \begin{align*}
        \Phi(s) &= s \exp\left( \frac{2(\beta + s)(1 + \gamma s)}{\beta - \gamma s^2} \right), \\
        \Phi'(s) &= \exp\left(\frac{2 (\beta +s) (\gamma  s+1)}{\beta -\gamma  s^2}\right) \frac{ (\beta +2 \beta\gamma s + \gamma s^2) (\beta + 2 s + \gamma s^2)}{\left(\beta -\gamma  s^2\right)^2} > 0.
    \end{align*}
    Therefore, we conclude that $\tilde{s} < \bar{x}$. 
    Similarly, if $\beta > 1$ and $\varphi(\bar{x}) \ge 0$, $\lambda e^{d \tilde{x}_{\lambda}} < \bar{x}$ holds for the same reason.
    
    On the other hand, for $\beta > 1$ and $\varphi(\bar{x}) < 0$, we claim that $\lambda_0 < \lambda_c(d_c)$. In fact, for $d \le d_c$, we have
    \begin{align*}
        \lambda_0 = \bar{x} \left(\frac{\beta + \bar{x}}{1 + \gamma \bar{x}}\right)^d \le \bar{x} \left(\frac{\beta + \bar{x}}{1 + \gamma \bar{x}}\right)^{d_c} < \hat{x}_{d_c} \left(\frac{\beta + \hat{x}_{d_c}}{1 + \gamma \hat{x}_{d_c}}\right)^{d_c} = \lambda_c(d_c),
    \end{align*}
    where we have used that $\frac{\beta + \bar{x}}{1 + \gamma \bar{x}} \ge \beta > 1$ and $\bar{x} < \hat{x}_{d_c}$ according to $\varphi(\hat{x}_{d_c}) = 0$ (cf. Lemma \ref{lem:lambda-c}).
    
    Then together with $\lambda e^{d x_0} < \lambda_c(d_c) = \Phi(\hat{x}_{d_c})$, we know that $\tilde{s} < \hat{x}_{d_c}$ which implies that $\varphi(\tilde{s}) < \varphi(\hat{x}_{d_c}) = 0$. 
    
\end{proof}

\begin{corollary}\label{coro:anti-ferro-positive}
    Suppose that $\sqrt{\beta\gamma} \le \frac{d-1}{d+1}$ and $\lambda_0 > 0$ such that
    \begin{align*}
        \lambda_0 < \begin{cases}
            \lambda_c(d_c), ~~~&\text{ for } \beta > 1, d > d_c, \\
            \lambda_c(d), ~~~&\text{ otherwise. }
        \end{cases}
    \end{align*}
    With setting
    \begin{align*}
        x_0 = \max\left\{\frac{2(1 - \beta\gamma)\bar{x}}{\beta - \gamma \bar{x}^2} + \ln\left(\frac{\gamma \bar{x} + 1}{\bar{x} + \beta}\right), 0\right\},
    \end{align*}
    where $\bar{x}$ satisfies $\lambda_0 \left(\frac{\gamma \bar{x} + 1}{\bar{x} + \beta}\right)^d = \bar{x}$,
    and $x_1 = \min\{0, \ln(\gamma)\}$ for $\gamma > 0$, and $x_1 = \min\{0, -\ln(\beta + \lambda_0 e^{d x_0}) - 1\}$ for $\gamma = 0$,
    there exists $k > 0$ such that 
    \begin{align*}
        \phi(\lambda e^{d w}) \in \open{U(x_0, x_1, k)}, ~~~\forall w \in U(x_0, x_1, k).
    \end{align*}
    
    Furthermore, $U = \{z \in U(x_0, x_1, k): \re(z) \in [x_2, -\ln(\beta)]\}$, where $x_2 = \ln(\gamma)$ for $\gamma > 0$ and $x_2 = -\ln(\beta + \lambda_0 e^{d x_0}) - 1$ for $\gamma = 0$, satisfies all conditions in Lemma \ref{lem:bounded-2-spin:U}.
\end{corollary}

\begin{proof}
    If $\beta > 1 > \gamma$, then we have
    $
        \frac{1 - \gamma}{\beta - 1} > \frac{1}{\beta}, 
    $
    and thus $\ln\left(\frac{1 - \gamma}{\beta - 1}\right) \notin U$. 
    
    If $\beta < 1 < \gamma$, then there holds
    $
        \frac{\gamma - 1}{1 - \beta} < \gamma,
    $
    and hence $\ln\left(\frac{1 - \gamma}{\beta - 1}\right) \notin U$. 
    
    And it is trivial that $1 + \gamma \lambda_0 e^{d x_0} > 0, \beta + \lambda_0 e^{d x_0} > 0, 1 + \lambda_0 e^{(d + 1) x_0} > 0$.
    
    Then with Lemma \ref{lem:anti-ferro-positive:range-r} and Lemma \ref{lem:anti-ferro-positive:H}, we know that $U$ satisfies all conditions in Lemma \ref{lem:bounded-2-spin:U}. 
    
\end{proof}

\begin{figure}[H]
\centering
\includegraphics[width=\textwidth]{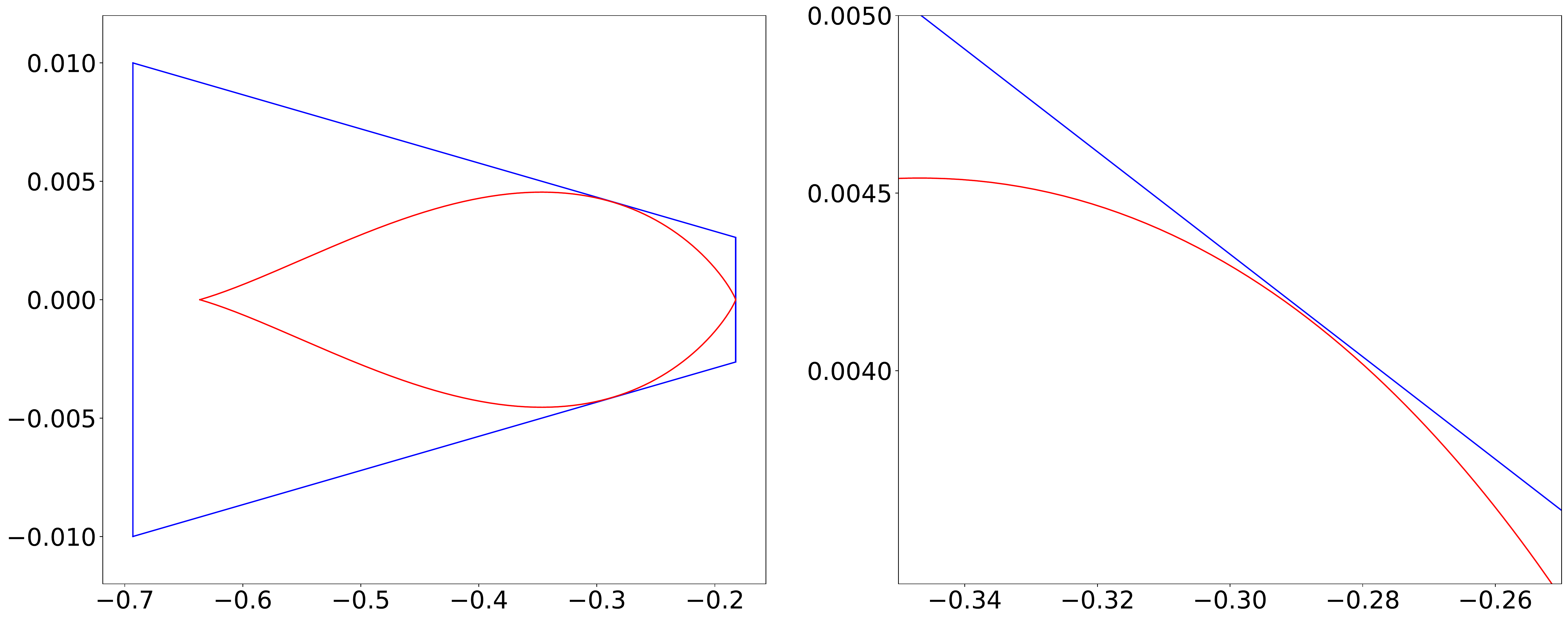}
\caption{Blue lines: the boundary of $U = \{z \in U(x_0, \ln(\gamma), 0.01): \re(z) \le -\ln(\beta)\}$. Red curve: the boundary of $\{\phi(\lambda e^{d w}): w \in U(\ln(\gamma), x_0, 0.01)\}$ where $\beta = 1.2, \gamma = 0.5, d = 15$ and $\lambda = 12.5 < \lambda_c \approx 12.6$. The figure on the right enlarges part of the figure on the left to clarify the containment relationship between the two regions clearly. }
\end{figure}

\subsection{The case $\beta\gamma < 1$ and $\lambda_0 < 0$}

\begin{lemma}\label{lem:anti-ferro-negative:range-r}
    For $\beta\gamma < 1$ and $\lambda_0 < 0$ such that 
    \begin{align*}
        \lambda_0 > \begin{cases}
        -\min\left\{1, \frac{\beta - 1}{1 - \gamma}\right\}, ~~~&\text{ for } \beta > 1, d > \frac{1 - \beta\gamma}{(\beta - 1)(1 - \gamma)}, \\
        -\frac{\beta \check{x}_d - 1}{(\check{x}_d - \gamma) \check{x}_d^{d}}, ~~~&\text{ otherwise, }
        \end{cases}
    \end{align*} 
    with setting $x_0 = \max\{\ln(\check{x}_d), 0\}$, $x_1 = \ln(\gamma)$ if $\gamma > 0$ and $x_1 = -\ln(\beta) - 1$ if $\gamma = 0$,
    there holds 
    \begin{align}
        r(d x, d y, \lambda) \in \begin{cases}
            \left(x_1, \frac{1 - \gamma |\lambda_0|}{\beta - |\lambda_0|}\right], ~~~&\text{ for } \beta > 1, \beta + \gamma > 2, \\
            (x_1, x_0), ~~~&\text{ otherwise, } 
        \end{cases}
    \end{align} 
    for $x + \i y \in U(x_0, x_1, k)$ and $\lambda \in [\lambda_0, 0]$.
\end{lemma}
\begin{proof}
    Following from Lemma \ref{lem:check-x-1}, we know that $\ln(\check{x}_d) < 0$ if and only if $\beta > 1$ and $d > \frac{1 - \beta\gamma}{(\beta - 1)(1 - \gamma)}$. 
    Next, We point out that 
    \begin{align*}
        |\lambda| e^{d x} \le |\lambda_0| e^{d x_0} < \begin{cases}
        \frac{\beta - 1}{1 - \gamma} < \beta, ~~~&\text{ for } \beta > 1, d > \frac{1 - \beta\gamma}{(\beta - 1)(1 - \gamma)}, \\
        \frac{\beta \check{x}_d - 1}{\check{x}_d - \gamma} < \beta, ~~~&\text{ otherwise, }
        \end{cases}
    \end{align*}
    which implies that $\beta + \lambda e^{d x} > 0$ and $1 + \gamma \lambda e^{d x} > 1 - \beta \gamma > 0$. 
    Further, by $k < \frac{\pi}{2 d (x_0 - x_1)}$, we know that $|d y| \le d k (x_0 - x_1) < \frac{\pi}{2}, \cos(d y) > 0$. Hence, there holds $\beta\gamma + 1 + 2 \gamma \lambda e^{d x} \cos(d y) > \beta \gamma + 1- 2 \beta\gamma > 0$ and
    \begin{align*}
        \frac{1 + \gamma^2 \lambda^2 e^{2 d x} + 2 \gamma \lambda e^{d x} \cos(d y)}{\beta^2 + \lambda^2 e^{2 d x} + 2 \beta \lambda e^{d x} \cos(d y)} - \gamma^2 
        &= \frac{(1 - \beta  \gamma) \left(\beta  \gamma + 1 +2 \gamma  \lambda  e^{d x} \cos (d y)\right)}{\beta ^2+\lambda ^2 e^{2 d x}+2 \beta  \lambda  e^{d x} \cos (d y)} > 0,
    \end{align*}
    which implies that $r(x, y, k, \lambda) > \ln(\gamma)$. 
    
    If $\gamma = 0$, then there holds
    \begin{align*}
        \frac{1 + \gamma^2 \lambda^2 e^{2 d x} + 2 \gamma \lambda e^{d x} \cos(d y)}{\beta^2 + \lambda^2 e^{2 d x} + 2 \beta \lambda e^{d x} \cos(d y)} = \frac{1}{\beta^2 + \lambda^2 e^{2 d x} + 2 \beta \lambda e^{d x} \cos(d y)} \ge \frac{1}{\beta^2 + \lambda^2 e^{2d x}} \ge \frac{1}{2\beta^2}.
    \end{align*}
    
    On the other hand, we have $\beta - \gamma \lambda^2 e^{2 d x} > \beta - \gamma \beta^2 > 0$ and
    \begin{align*}
        &\quad \frac{1 + \gamma^2 \lambda^2 e^{2 d x} + 2 \gamma \lambda e^{d x} \cos(d y)}{\beta^2 + \lambda^2 e^{2 d x} + 2 \beta \lambda e^{d x} \cos(d y)} \\
        &= \frac{\beta}{\gamma} +\frac{(1 - \beta \gamma) \left(\beta -\gamma  \lambda ^2 e^{2 d x}\right)}{\beta  \left(\beta ^2+\lambda ^2 e^{2 d x}+2 \beta  \lambda  e^{d x} \cos (d y)\right)} \\
        &\le \frac{\beta}{\gamma} +\frac{(1 - \beta \gamma) \left(\beta -\gamma  \lambda ^2 e^{2 d x}\right)}{\beta  \left(\beta ^2+\lambda ^2 e^{2 d x} + 2 \beta  \lambda  e^{d x} \right)} \\
        &= \left(\frac{1 + \gamma \lambda e^{d x}}{\beta + \lambda e^{d x}}\right)^2
        \le \left(\frac{1 + \gamma \lambda_0 e^{d x_0}}{\beta + \lambda_0 e^{d x_0}}\right)^2,
    \end{align*}
    where the last inequality is according to the function $\psi(s) \triangleq \frac{1 + \gamma s}{\beta + s}$ is strictly decreasing on $(- \beta, 0]$.
    
    Then for $\check{x}_d \ge 1$, we have $\lambda_0 e^{d x_0} > - \frac{\beta \check{x}_d - 1}{\check{x}_d - \gamma}$ and 
    \begin{align*}
        \frac{1 + \gamma \lambda_0 e^{d x_0}}{\beta + \lambda_0 e^{d x_0}} < \frac{1 - \gamma \frac{\beta \check{x}_d - 1}{\check{x}_d - \gamma}}{\beta - \frac{\beta \check{x}_d - 1}{\check{x}_d - \gamma}} = \check{x}_d.
    \end{align*}
    
    And if $\beta > 1$ and $\beta + \gamma > 2$, then we have $\frac{1 - \beta\gamma}{(\beta - 1)(1 - \gamma)} < 1$, which means that $\check{x}_d < 1$ for all $d \ge 1$. With $\frac{\beta - 1}{1 - \gamma} > 1$, we know that $|\lambda_0| < 1$ and
    \begin{align*}
        \frac{1 + \gamma \lambda_0 e^{d x_0}}{\beta + \lambda_0 e^{d x_0}} = \frac{1 - \gamma |\lambda_0|}{\beta - |\lambda_0|}.
    \end{align*}
    
    At last, for $\beta > 1$ and $\beta + \gamma \le 2$, with $d > \frac{1 - \beta\gamma}{(\beta - 1)(1 - \gamma)}$ there holds $\lambda_0 e^{d x_0} = \lambda_0 > - \frac{\beta - 1}{1 - \gamma}$ and
    \begin{align*}
        \frac{1 + \gamma \lambda_0 e^{d x_0}}{\beta + \lambda_0 e^{d x_0}} < \frac{1 - \gamma \frac{\beta  - 1}{1 - \gamma}}{\beta - \frac{\beta  - 1}{1 - \gamma}} = 1.
    \end{align*}
\end{proof}

\begin{lemma}\label{lem:anti-ferro-negative:H}
    For $\beta\gamma < 1$ and $\lambda_0 < 0$ such that 
    \begin{align*}
        \lambda_0 > \begin{cases}
        -\min\left\{1, \frac{\beta - 1}{1 - \gamma}\right\}, ~~~&\text{ for } \beta > 1, d > \frac{1 - \beta\gamma}{(\beta - 1)(1 - \gamma)}, \\
        -\frac{\beta \check{x}_d - 1}{(\check{x}_d - \gamma) \check{x}_d^{d}} , ~~~&\text{ otherwise, }
        \end{cases}
    \end{align*} 
    with setting $x_0 = \max\{\ln(\check{x}_d), 0\}$, we have $H(x, \lambda) > 0$ for $x \in (-\infty, x_0]$ and $\lambda \in [\lambda_0, 0]$.
\end{lemma}

\begin{proof}
    Recalling the definition of $H(x, \lambda)$, for $\beta\gamma < 1$ and $\lambda \le 0$, we have
    \begin{align*}
        H(x, \lambda) &= x_0 - \ln \left(\frac{\gamma \lambda e^{d x} + 1}{\beta + \lambda  e^{d x}}\right) + \frac{d \lambda  (1 - \beta  \gamma) e^{d x} (x_0-x)}{(\gamma \lambda e^{d x} + 1)(\beta + \lambda  e^{d x})}, \\
        \frac{\partial H(x, \lambda)}{\partial x} &= \frac{d^2 (1 - \beta \gamma) \lambda e^{d x} (x_0 - x) \left(\beta -\gamma  \lambda ^2 e^{2 d x} \right)}{\left(\gamma \lambda e^{d x} + 1\right)^2 \left(\beta +\lambda e^{d x} \right)^2 } \le 0.
    \end{align*}
    
    Thus it holds that
    \begin{align*}
        H(x, \lambda) \ge H(x_0, \lambda) = x_0 - \ln\left( \frac{\gamma \lambda e^{d x_0} + 1}{\beta + \lambda e^{d x_0}} \right) > 0,
    \end{align*}
    which have verified in proof of Lemma \ref{lem:anti-ferro-negative:range-r}.
\end{proof}

\begin{corollary}\label{coro:anti-ferro-negative}
    For $\beta\gamma < 1$ and $\lambda_0 < 0$ such that 
    \begin{align*}
        \lambda_0 > \begin{cases}
        -\min\left\{1, \frac{\beta - 1}{1 - \gamma}\right\}, ~~~&\text{ for } \beta > 1, d > \frac{1 - \beta\gamma}{(\beta - 1)(1 - \gamma)}, \\
        -\frac{\beta \check{x}_d - 1}{(\check{x}_d - \gamma) \check{x}_d^{d}} , ~~~&\text{ otherwise, }
        \end{cases}
    \end{align*} 
    with setting $x_0 = \max\{\ln(\check{x}_d), 0\}$ and $x_1 = \min\{0, \ln(\gamma)\}$ for $\gamma > 0$, and $x_1 = \min\{0, -\ln(\beta) - 1\}$ for $\gamma = 0$,
    there exists $k > 0$ such that 
    \begin{align*}
        \phi(\lambda e^{d w}) \in \open{U(x_0, x_1, k)}, ~~~\forall w \in U(x_0, x_1, k).
    \end{align*}
    
    Furthermore, 
    $U = \{z \in U(x_0, x_1, k): \re(z) \in [x_2, x_3]\}$, where 
    \begin{align*}
        x_2 &= \begin{cases}
            \ln(\gamma), ~~&\text{ for } \gamma > 0, \\
            -\ln(\beta) -1, ~~&\text{ for } \gamma = 0,
        \end{cases} \\
        x_3 &= 
        \begin{cases}
            \ln\left(\frac{1 - \gamma \tilde{\lambda}}{\beta - \tilde{\lambda}}\right) \ (|\lambda_0| < \tilde{\lambda} < 1), ~~&\text{ for } \beta > 1, \beta + \gamma > 2, \\
            x_0, ~~&\text{ otherwise, }
        \end{cases}
    \end{align*}
    satisfies all conditions in Lemma \ref{lem:bounded-2-spin:U}.
\end{corollary}

\begin{proof}
    For $\beta > 1 > \gamma$ and $\beta + \gamma \le 2$, we have 
    $\frac{(\beta - 1)(1 - \gamma)}{1 - \beta\gamma} \le 1$ which implies that
    $
        \frac{1 - \gamma}{\beta - 1} \ge \check{x}_d
    $
    according to Lemma \ref{lem:check-x-2}.
    And if $\beta + \gamma > 2$, then $\frac{1 - \gamma}{\beta - 1} > \frac{1 - \gamma \tilde{\lambda}}{\beta - \tilde{\lambda}}$. Therefore we have $\ln\left(\frac{1 - \gamma}{\beta - 1}\right) \notin U$. 
    
    If $\beta < 1 < \gamma$, then there holds
    $
        \frac{\gamma - 1}{1 - \beta} < \gamma,
    $
    and hence $\ln\left(\frac{1 - \gamma}{\beta - 1}\right) \notin U$. 
    
    According to $|\lambda_0 e^{d x_0}| < \beta$, we know that that $1 + \gamma \lambda_0 e^{d x_0} > 0, \beta + \lambda_0 e^{d x_0} > 0$.
    And if $\check{x}_d < 1$, then $1 + \lambda_0 e^{(d + 1) x_0} = 1 + \lambda_0 > 0$. Further, for $\check{x}_d \ge 1$, we have
    \begin{align*}
        |\lambda_0| e^{(d + 1) x_0} < \frac{(\beta \check{x}_d - 1)\check{x}_d}{\check{x}_d - \gamma}.
    \end{align*}
    Note that 
    \[\beta \check{x}_d^2 + \gamma = \frac{1 - \beta\gamma + d(1 + \beta\gamma)}{d} \check{x}_d \le 2 \check{x}_d,\] 
    which implies that $\frac{(\beta \check{x}_d - 1)\check{x}_d}{\check{x}_d - \gamma} < 1$. 
    
    Then with Lemma \ref{lem:anti-ferro-positive:range-r} and Lemma \ref{lem:anti-ferro-positive:H}, we know that $U$ satisfies all conditions in Lemma \ref{lem:bounded-2-spin:U}. 
    
\end{proof}

\begin{figure}[H]
\centering
\includegraphics[width=\textwidth]{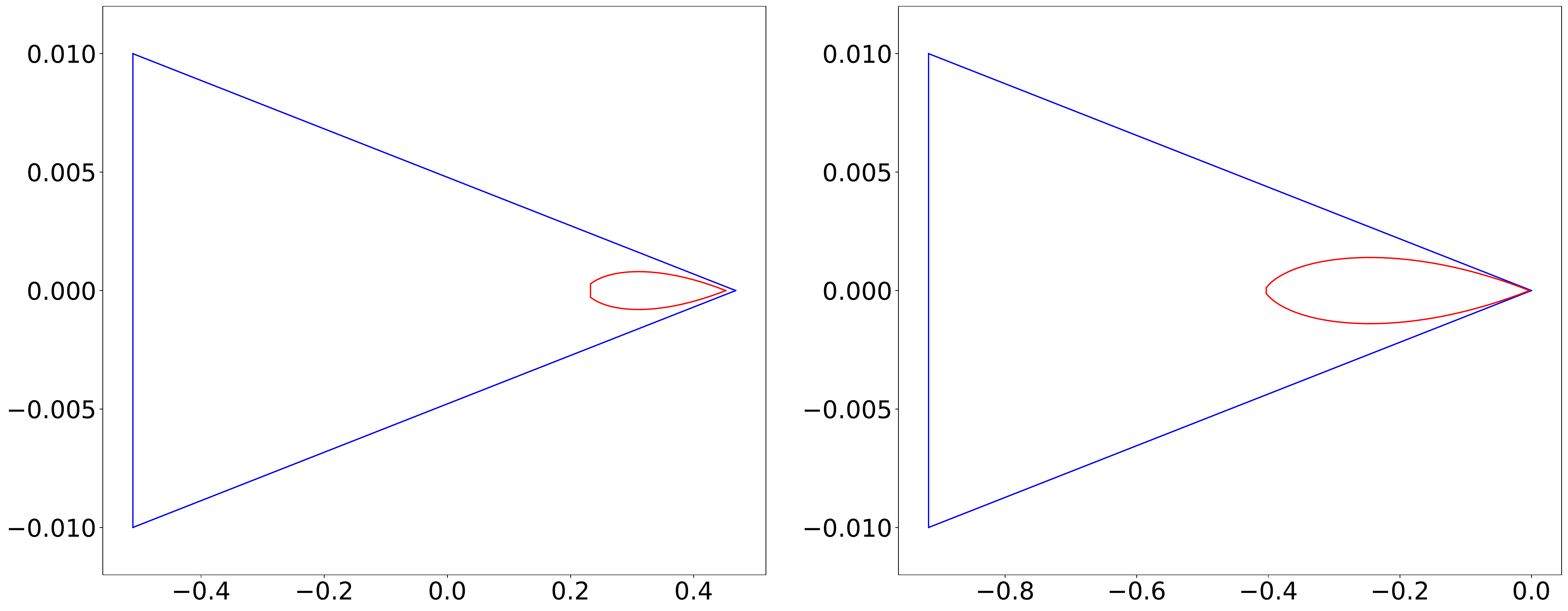}
\caption{Blue lines: the boundary of $U(x_0, \ln(\gamma), 0.01)$. Red curve: the boundary of $\{\phi(\lambda e^{d w}): w \in U(x_0, \ln(\gamma), 0.01)\}$. Left: $\beta = 0.8, \gamma = 0.6, d = 3$ and $\lambda = -0.065 > \lambda_c \approx -0.068$. Right: $\beta = 1.5, \gamma = 0.4, d = 5$ and $\lambda = -0.83 > \lambda_c \approx -0.833$.}
\end{figure}

\section{Proof of Theorem \ref{thm:bounded-2-spin:supp}}
\begin{proof}
    It is easily to check that we also have
    \begin{align*}
        \re\left(\phi(\lambda e^{d w})\right) = r(d x, d y, \lambda), ~~~ |\re\left(\phi(\lambda e^{d w})\right)| = h(d x, d y, \lambda).
    \end{align*}
    
    Following from Lemma \ref{lem:ferro-positive:range-r} and Lemma \ref{lem:anti-ferro-positive:range-r}, we know that 
    \begin{align*}
        r(x, y, \lambda, d) \in \begin{cases}
            (-\ln(\beta), \ln(\gamma)), ~~~&\text{ for } \beta\gamma > 1, \\
            (\ln(\gamma), -\ln(\beta)), ~~~&\text{ for } \beta\gamma < 1.
        \end{cases}
    \end{align*}
    
    For $\beta\gamma > 1$, set $x_0 = \max\{0, \ln(\gamma)\}$ and $x_1 = \min\{0, -\ln(\beta)\}$. 
    Then if $\beta > 1 > \gamma$, then $\frac{1 - \gamma}{\beta - 1} < \frac{1}{\beta}$ and $x_1 = -\ln(\beta)$ which implies that $\ln\left(\frac{1 - \gamma}{\beta - 1}\right) \notin R(x_0, x_1, y_0)$. 
    And if $\beta < 1 < \gamma$, then $\frac{1 - \gamma}{\beta - 1} > \gamma$. With $x_0 = \ln(\gamma)$, we have $\ln\left(\frac{1 - \gamma}{\beta - 1}\right) \notin R(x_0, x_1, y_0)$. 
    
    And for $\beta\gamma < 1$, set $x_0 = \max\{0, -\ln(\beta)\}$ and $x_1 = \min\{0, \ln(\gamma)\}$. Similarly, we also have $\ln\left(\frac{1 - \gamma}{\beta - 1}\right) \notin R(x_0, x_1, y_0)$.
    
    And it is trivial that $1 + \gamma \lambda_0 e^{d x_0} > 0, \beta + \lambda_0 e^{d x_0} > 0, 1 + \lambda_0 e^{(d + 1) x_0} > 0$.
    
    Then in order to prove that $\phi(\lambda e^{d w}) \in R(x_0, x_1, y_0)$ for $w \in R(x_0, x_1, y_0)$ and $\lambda \in [0, \lambda_0]$, we just need to show that
    \begin{align*}
        h(x, y, \lambda, d) \le h(x, y_0, \lambda, d) < y_0,
    \end{align*}
    where $d = \Delta - 1$, $x \in [x_1, x_0], y \in [0, y_0]$ and $\lambda \in [0, \lambda_0]$.
    
    Denote
    \begin{align*}
        \hat{G}(x, y_0, \lambda) \triangleq y_0 - h(x, y_0, \lambda, d), ~~~~
        \hat{H}(x, \lambda) \triangleq \frac{\partial \hat{G}(x, y_0, \lambda)}{\partial y_0} \bigg\vert_{y_0 = 0}.
    \end{align*}
    Based on Lemma \ref{lem:bounded-2-spin:H}, we just need to prove that $\hat{H}(x, \lambda) > 0$ for $x \in [x_1, x_0]$ and $\lambda \in [0, \lambda_0]$.

    Note that
    \begin{align*}
        \hat{H}(x, \lambda) = 1-\frac{d |\beta  \gamma -1|  \lambda e^{d x}}{\left(\beta +\lambda  e^{d x}\right) \left(\gamma  \lambda  e^{d x}+1\right)}.
    \end{align*}
    
    Consider a function
    \begin{align*}
        \psi(s) &\triangleq \frac{s}{(\beta + s)(1 + \gamma s)}, \\
        \psi'(s) &= \frac{\beta -\gamma  s^2}{(\beta +s)^2 (\gamma  s+1)^2}.
    \end{align*}
    It is clearly that $\psi(s)$ is strictly increasing on $[0, \sqrt{\beta/\gamma}]$ and strictly decreasing on $[\sqrt{\beta/\gamma}, +\infty)$.
    Therefore, for $\sqrt{\beta\gamma} \in \left( \frac{d - 1}{d + 1}, \frac{d + 1}{d - 1} \right)$, we have 
    \begin{align*}
        \hat{H}(x, \lambda) \ge 1 - \frac{d |\beta\gamma - 1|}{(1 + \sqrt{\beta\gamma})^2} = 1 - \frac{d |\sqrt{\beta\gamma} - 1|}{(1 + \sqrt{\beta\gamma})} > 0. 
    \end{align*}
    
\end{proof}

\section{Proof of Lemma \ref{lem:unbounded-2-spin:contraction-zero-free}}
\begin{proof}
    Let $F_{\lambda}$ be the region satisfies the conditions in Definition \ref{def:unbounded-2-spin:complex-contraction}. Similar to the proof of Lemma \ref{lem:bounded-2-spin:contraction-zero-free}, we can apply induction on $t = |V| - |S|$ to prove following stronger results:
    \begin{enumerate}
        \item $R^{\lambda}_{\tau_S}(G, v) \in F_{\lambda}$ holds for all free vertex $v$, 
        \item $Z_{\tau_S}(G, \lambda) \neq 0$.
    \end{enumerate}
    
\end{proof}

\section{Proof of Lemma \ref{lem:unbounded-2-spin:U}}
\begin{proof}
    Similar to the proof of Lemma \ref{lem:bounded-2-spin:U}, with setting $F = \phi^{-1}(U)$, we have $-1 \notin F$, $0 \in F$ and $\infty \in F$ if $\gamma > 0$. 
    
    It is clearly that
    \begin{align*}
        \hat{U} \triangleq \{\lambda e^{d w}: \lambda \in [0, \lambda_0], w \in U, d \ge 0\} \subseteq \{\lambda e^{w}: \lambda \in [0, \lambda_0], w \in \tilde{U}\}. 
    \end{align*}
    And note that 
    \begin{align*}
        |\lambda e^{w}| \le |\lambda_0| e^{\re(w)} \le |\lambda_0| e^{a}, ~~\text{ for } w \in \tilde{U},
    \end{align*}
    which implies that $\hat{U}$ is compact.
    
    Similar to the proof of Lemma \ref{lem:bounded-2-spin:U}, together with 
    \begin{align*}
        |\lambda_1 e^{d w} - \lambda_2 e^{d w}| \le |\lambda_1 - \lambda_2| e^{a},
    \end{align*}
    we know that there exists $\delta$ such that $\lambda e^{d w} \notin \gI$ and $\phi(\lambda e^{d w}) \in U$ for all $\lambda \in \gU([0, \lambda_0], \delta)$ and $w \in U$.
    Moreover, based on the same reason in the proof of Lemma \ref{lem:bounded-2-spin:U}, all $\lambda \in \gU([0, \lambda_0], \delta)$ satisfies complex-contraction property with region $F_{\lambda} = F$. 
    
\end{proof}

\section{Proof of Lemma \ref{lem:unbounded-2-spin:x1-k}}
\begin{proof}
    Set $x_0, x_1$ as in Lemma \ref{lem:ferro-positive:H}, Lemma \ref{lem:ferro-negative:H}, Lemma \ref{lem:anti-ferro-positive:H} and Lemma \ref{lem:anti-ferro-negative:H}. We point out that $x_0 = 0$ in each cases. 
    
    We first prove that there exists $\tilde{x} < 0$ and $k_1 > 0$ such that 
    \begin{align*}
        \phi(\lambda e^w) \in \open{U(0, x_1, k)}
    \end{align*}
    for all $w \in \tilde{U}(k)$ with $\re(w) < \tilde{x}$ and $k < k_1$. 
    
    Let $w = x + \i y$. Set $\tilde{x}$ small enough to ensure that $1 - \gamma |\lambda| e^{\tilde{x}} > 0$ and $\beta - |\lambda| e^{\tilde{x}} > 0$. 
    Then it is easily to check that $\beta - \gamma \lambda^2 e^{\tilde{x}} > 0$, $\re(\phi(\lambda e^w)) = r(x, y, \lambda)$ and $|\im(\phi(\lambda e^w))| = h(x, y, \lambda)$ for $x < \tilde{x}$.
    
    Note that for $\beta\gamma < 1$, we have
    \begin{align*}
        &\quad \frac{1 + \gamma^2 \lambda^2 e^{2 d x} + 2 \gamma \lambda e^{d x} \cos(d y)}{\beta^2 + \lambda^2 e^{2 d x} + 2 \beta \lambda e^{d x} \cos(d y)} \\
        &= \frac{\gamma}{\beta} +\frac{(1 - \beta \gamma) \left(\beta -\gamma  \lambda ^2 e^{2 d x}\right)}{\beta  \left(\beta ^2+\lambda ^2 e^{2 d x}+2 \beta  \lambda  e^{d x} \cos (d y)\right)} \\
        &\le \frac{\gamma}{\beta} +\frac{(1 - \beta \gamma) \left(\beta -\gamma  \lambda ^2 e^{2 d x}\right)}{\beta  \left(\beta ^2+\lambda ^2 e^{2 d x} - 2 \beta  |\lambda|  e^{d x} \right)} \\
        &= \left(\frac{1 - \gamma |\lambda| e^{d \tilde{x}}}{\beta - |\lambda| e^{d \tilde{x}}}\right)^2
        < \left(\frac{1 - \gamma/\beta e^{d \tilde{x}}}{\beta - 1/\beta}\right)^2
    \end{align*}
    which means $\re(\phi(\lambda e^w)) < \ln\left(\frac{1 + \gamma/\beta}{\beta + 1/\beta}\right)$. 
    
    Furthermore, for $x < - \frac{\pi}{2 k}$, we have
    \begin{align*}
        h(x, y, \lambda) &\le \frac{(1 - \beta \gamma) |\lambda| e^{x}}{\beta + \gamma |\lambda|^2 e^{2x} - (\beta \gamma + 1)|\lambda| e^{x}} \\
        &= \frac{(1 - \beta \gamma) |\lambda| e^{x}}{(\beta - |\lambda| e^{x})(1 - \gamma |\lambda| e^{x})} \\
        &< \frac{(1 - \beta \gamma) |\lambda| e^{-\frac{\pi}{2 k}}}{(\beta - 1/\beta)(1 - \gamma /\beta)}. 
    \end{align*}
    Then $h(x, y, \lambda) < k |r(x, y, \lambda)|$ holds for sufficiently small $k$ according to $\lim_{k \to 0^{+}} k e^{\frac{\pi}{2 k}} = + \infty$. 
    
    And for $x \in [-\frac{\pi}{2 k}, \tilde{x}]$, we have $|y| \le k |x|$
    \begin{align*}
        h(x, y, \lambda) &\le \frac{(1 - \beta \gamma) |\lambda| e^{x} \sin(k |x|) }{\beta + \gamma |\lambda|^2 e^{2x} - (\beta \gamma + 1)|\lambda| e^{x}} \\
        &< \frac{k (1 - \beta \gamma) |\lambda| e^{\tilde{x}} |\tilde{x}|}{(\beta - 1/\beta)(1 - \gamma /\beta)}. 
    \end{align*}
    Then $h(x, y, \lambda) < k |r(x, y, \lambda)|$ can be guaranteed by
    \begin{align*}
        \frac{(1 - \beta \gamma) |\lambda| e^{\tilde{x}} |\tilde{x}|}{(\beta - 1/\beta)(1 - \gamma /\beta)} < \ln\left(\frac{\beta + 1/\beta}{1 + \gamma/\beta}\right).
    \end{align*}
    
    Furthermore, a same result also holds for $\beta \gamma > 1$.
    
    At last, $\phi(\lambda e^w) \in \open{U(0, x_1, k)}$ for all $w \in \tilde{U}(k)$ with $\re(w) \ge \tilde{x}$ can be proved by the same method for proving bounded degree 2-spin systems, especially Lemma \ref{lem:bounded-2-spin:x0-x1}. 
    
\end{proof}

\section{Examples of weighted Set Covers}
\label{app:set-cover-problem}

In this section, we list some problems which can be transformed into set cover problems. 

\paragraph{Counting the independent sets on a hypergraph.}
Partition function for counting hypergraph independent sets on a hypergraph $G = (V, E)$ is defined as
\begin{align*}
    \mathrm{HIS}(G, \eta) \triangleq \sum_{\sigma \in \{0, 1\}^V} \eta^{|\{\sigma(v) = 1\}|} \prod_{e \in E} \left(1 - \prod_{v \in e} \sigma(v) \right).
\end{align*}
In fact, there holds $\mathrm{HIS}(G, \eta) = Z(G, -1, \eta)$.

\paragraph{Edge Cover}
Partition function of a generalized version of edge cover introduced in \cite{liu2014fptas-edge} for a graph $\tilde{G} = (\tilde{V}, \tilde{E})$ is defined to be
\begin{align*}
    \mathrm{EC}(\tilde{G}, \mu, \eta) \triangleq \sum_{\sigma \in \{0, 1\}^{\tilde{E}}} \prod_{\tilde{e} \in \tilde{E}} \eta^{\sigma(\tilde{e})} \prod_{\tilde{v} \in \tilde{V}} \left(1 - (1 - \mu) \prod_{\tilde{e}: \tilde{v} \in \tilde{e}} (1 - \sigma(\tilde{e})) \right).
\end{align*}

Consider a hypergraph $G = (V, E)$ such that
\begin{align*}
    V = \tilde{E}, ~~ E = \{ e_{\tilde{v}} \triangleq \{\tilde{e} \in \tilde{E}: \tilde{v} \in \tilde{e}\}: \tilde{v} \in \tilde{V} \}.
\end{align*}
The maximum degree of $G$ is 2 following from $\tilde{e} = (\tilde{v}_1, \tilde{v}_2)$ can only appears in two hyperedges $e_{\tilde{v}_1}$ and $e_{\tilde{v}_2}$.
Furthermore, it is easily to check that 
\begin{align*}
    \mathrm{EC}(\tilde{G}; \mu, \eta) = \eta^{|\tilde{E}|} Z(G, \mu - 1, 1/\eta). 
\end{align*}

\paragraph{Counting the independent sets on a bipartite graph.}
Consider a bipartite graph $\tilde{G} = (L \cup R, \tilde{E})$. Define
\begin{align*}
    \mathrm{BIS}(\tilde{G}, \eta, \mu) = \sum_{I \text{ independent}} \eta^{|I \cap L|} \mu^{|I \cap R|}.
\end{align*}

Construct a hypergraph $G = \{ L, E \}$ where $E = \{\{u: u \sim v\}: v \in R\}$.

\begin{lemma}\label{lemma:bis}
    There holds
    \begin{align*}
        \mathrm{BIS}(\tilde{G}, \eta, \mu) = \eta^{|L|} Z(G, \mu, 1/\eta).
    \end{align*}
\end{lemma}
\begin{proof}
    For $I \subseteq L$ and $J \subseteq R$, $I \cup J$ is an independent set if and only if $\gN(I) \cap J = \emptyset$. 
    Hence we have
    \begin{align*}
        \mathrm{BIS}(\tilde{G}, \eta, \mu) &= \sum_{I \subseteq L} \eta^{|I|} \sum_{\substack{J : J \subseteq R \\ I \cup J \text{ independent}} } \mu^{|J|} \\
        &= \sum_{I \subseteq L} \eta^{|I|} (1 + \mu)^{|R| - |\gN(I)|} \\
        &= \sum_{I \subseteq L} \eta^{|I|} \prod_{v \in R} \left(1 + \mu \prod_{u \in L, u \sim v} (1 - I(u)) \right) \\
        &= \eta^{|L|} \sum_{I \subseteq L} \eta^{-|\{v: I(v) = 0\}|} \prod_{e \in E} \varphi_e^{\mu} (I\vert_{e}) \\
        &= \eta^{|L|} Z(G, \mu, 1/\eta).
    \end{align*}
\end{proof}
\section{Computation Tree for Weighted Set Covers}
\label{app:set-cover:tree}

For a hypergraph $G = (V, E)$ and a vertex $v \in V$, we denote $G - v = \{V \backslash \{v\}, \{ e \backslash \{v\}: e \in E \}\}$.
And for $e \in E$, we define $G - e = \{V, E \backslash \{e\}\}$.

Consider a vertex $v \in V$, if either $Z_{\tau_S, \sigma(v) = 0}(G, \mu, \lambda)\neq 0$ or $Z_{\tau_S, \sigma(v) = 1}(G, \mu, \lambda) \neq 0$, define 
\begin{align}
    R^{\mu, \lambda}_{\tau_S}(G, v) = \frac{Z_{\tau_S, \sigma(v) = 0}(G, \mu, \lambda)}{Z_{\tau_S, \sigma(v) = 1}(G, \mu, \lambda)} \in \hat{C}.
\end{align}

If $v$ is pinned by $\tau_S$ such that $\tau_S(v) = 1$, then we have
$Z_{\tau_S, \sigma(v) = 0}(G, \mu, \lambda) = 0, Z_{\tau_S, \sigma(v) = 1}(G, \mu, \lambda) = Z_{\tau_S}(G, \mu, \lambda)$,
that is $R^{\mu, \lambda}_{\tau_S}(G, v) = 0$. On the other hand, if $\tau_S(v) = 0$, then $Z_{\tau_S, \sigma(v) = 0}(G, \mu, \lambda) = Z_{\tau_S}(G, \mu, \lambda), Z_{\tau_S, \sigma(v) = 1}(G, \mu, \lambda) = 0$, which implies that $R^{\mu, \lambda}_{\tau_S}(G, v) = \infty$. 
Moreover, if $v$ is free and does not belong to any edge in $G$. Then it is trivial to check that $R^{\mu, \lambda}_{\tau_S}(G, v) = \lambda$.

Now consider some vertices $S \subseteq V$ pinned by $\tau_S$, and a free (unpinned) vertex $v \notin S$. Assume that $v$ occurs in the edges $e_1, \cdots, e_D$. We replace $v$ in $e_i$ with a independent duplicate $\tilde{v}_i$ for $i = 1, \cdots, D$, and denote this new hypergraph by $\tilde{G}$. We have
\begin{align}\label{eq:tree:v}
    \notag R^{\mu, \lambda}_{\tau_S}(G, v) &= \frac{Z_{\tau_S, \sigma(v) = 0}(G, \mu, \lambda)}{Z_{\tau_S, \sigma(v) = 1}(G, \mu, \lambda)}
    = \frac{ Z_{\tau_S, \sigma(\tilde{v}_i) = 0, i = 1, \cdots, D}(\tilde{G}, \mu, \lambda)}{\lambda^{D-1} Z_{\tau_S, \sigma(\tilde{v}_i) = 1, i = 1, \cdots, D}(\tilde{G}, \mu, \lambda)} \\
    &= \frac{1}{\lambda^{D-1}} \prod_{i=1}^d \frac{Z_{\tau_S, \sigma_i, \sigma(\tilde{v}_i) = 0}(\tilde{G}, \mu, \lambda)}{Z_{\tau_S, \sigma_{i}, \sigma(\tilde{v}_i) = 1}(\tilde{G}, \mu, \lambda)}.
\end{align}
where $\sigma_i$ satisfies $\sigma_i(\tilde{v}_s) = 1$ for $s = 1, \cdots, i - 1$ and $\sigma_i(\tilde{v}_s) = 0$ for $s = i + 1, \cdots, D$.

Suppose that $e_i = \{\tilde{v}_i\} \cup \{v_{i, j}\}_{j=1}^{d_i}$. 
Note that $\tilde{v}_i$ only appears in $e_i$. If $\sigma(\tilde{v}_i) = 1$, we have $\varphi_{e_i}(\sigma\vert_{e_i}) = 1$, that is
\begin{align*}
    Z_{\tau_S, \sigma_{i}, \sigma(\tilde{v}_i) = 1}(\tilde{G}, \mu, \lambda) = Z_{\tau_S, \sigma_{i}}(\tilde{G} - \tilde{v}_i - e_i, \mu, \lambda), 
\end{align*}

If $\tau_S(v_{i, j}) = 1$ for some $j \in [d_i]$, then we also have $\varphi_{e_i}(\sigma\vert_{e_i}) = 1$, that is 
\begin{align*}
     Z_{\tau_S, \sigma_{i}, \sigma(\tilde{v}_i) = 0}(\tilde{G}, \mu, \lambda) = \lambda Z_{\tau_S, \sigma_{i}}(\tilde{G} - \tilde{v}_i - e_i, \mu, \lambda).
\end{align*}
Consequently, we have
\begin{align}\label{eq:tree:e:=1}
     \frac{Z_{\tau_S, \sigma_i, \sigma(\tilde{v}_i) = 0}(\tilde{G}, \mu, \lambda)}{Z_{\tau_S, \sigma_{i}, \sigma(\tilde{v}_i) = 1}(\tilde{G}, \mu, \lambda)} &= \lambda.
\end{align}

Next we assume that $\tau_S(v_{i, j}) = 0$ if $v_{i, j} \in S$. 
Observe that for $\sigma(\tilde{v}_i) = 0$, we have
\begin{align*}
    \varphi_{e_i}(\sigma\vert_{e_i}) = \begin{cases}
    1 + \mu, ~~&\text{if } \sigma(v_{i, j}) = 0, \text{ for } j = 1, \cdots, d_i, \\
    1, ~~&\text{otherwise.}
    \end{cases}
\end{align*}
Therefore, we have
\begin{align*}
     Z_{\tau_S, \sigma_{i}, \sigma(\tilde{v}_i) = 0}(\tilde{G}, \mu, \lambda) &= \lambda\left((1 + \mu) Z_{\tau_S, \sigma_{i}, \hat\sigma_0}(\tilde{G} - \tilde{v}_i - e_i, \mu, \lambda) + \sum_{\hat{\sigma} \in \{0, 1\}^{d_i} \backslash \vzero } Z_{\tau_S, \sigma_{i}, \hat{\sigma}}(\tilde{G} - \tilde{v}_i - e_i, \mu, \lambda)\right) \\
     &= \lambda \left( Z_{\tau_S, \sigma_{i}}(\tilde{G} - \tilde{v}_i - e_i, \mu, \lambda) + \mu Z_{\tau_S, \sigma_{i}, \hat\sigma_0}(\tilde{G} - \tilde{v}_i - e_i, \mu, \lambda)\right),
\end{align*}
where $\hat\sigma_0(v_{i, j}) = 0$ for $j = 1, \cdots, d_i$. Hence with $G_i = \tilde{G} - \tilde{v}_i - e_i$, it holds that
\begin{align}\label{eq:tree:e}
    \notag \frac{Z_{\tau_S, \sigma_i, \sigma(\tilde{v}_i) = 0}(\tilde{G}, \mu, \lambda)}{Z_{\tau_S, \sigma_{i}, \sigma(\tilde{v}_i) = 1}(\tilde{G}, \mu, \lambda)} &= \lambda \frac{Z_{\tau_S, \sigma_{i}}(G_i, \mu, \lambda) + \mu Z_{\tau_S, \sigma_{i}, \hat\sigma_0}(G_i, \mu, \lambda)}{ Z_{\tau_S, \sigma_{i}}(G_i, \mu, \lambda)} \\
    &= \lambda \left(1 + \mu \prod_{j=1}^{d_i} \frac{Z_{\tau_S, \sigma_{i}, \hat\sigma_{j+1}}(G_i, \mu, \lambda)}{Z_{\tau_S, \sigma_{i}, \hat{\sigma}_{j}}(G_i, \mu, \lambda)}\right),
\end{align}
where $\hat{\sigma}_j$ satisfies $\hat{\sigma}_j(v_{i, s}) = 0$ for $s = 1, \cdots, j - 1$.

Next, note that
\begin{align}\label{eq:tree:e2}
    \notag \frac{Z_{\tau_S, \sigma_{i}, \hat\sigma_{j-1}}(G_i, \mu, \lambda)}{Z_{\tau_S, \sigma_{i}, \hat{\sigma}_{j}}(G_i, \mu, \lambda)}
    &= \frac{Z_{\tau_S, \sigma_{i}, \hat\sigma_{j}, \sigma(v_{i, j}) = 0}(G_i, \mu, \lambda)}{Z_{\tau_S, \sigma_{i}, \hat{\sigma}_{j}, \sigma(v_{i, j}) = 0}(G_i, \mu, \lambda) + Z_{\tau_S, \sigma_{i}, \hat{\sigma}_{j}, \sigma(v_{i, j}) = 1}(G_i, \mu, \lambda)} \\
    &= \frac{R^{\mu, \lambda}_{\tau_S, \sigma_i, \hat{\sigma}_j}(G_i, v_{i, j})}{1 + R^{\mu, \lambda}_{\tau_S, \sigma_i, \hat{\sigma}_j}(G_i, v_{i, j})}.
\end{align}

Put Equation (\ref{eq:tree:e}) and (\ref{eq:tree:e2}) together, we obtain that
\begin{align}\label{eq:tree:e3}
     \frac{Z_{\tau_S, \sigma_i, \sigma(\tilde{v}_i) = 0}(\tilde{G}, \mu, \lambda)}{Z_{\tau_S, \sigma_{i}, \sigma(\tilde{v}_i) = 1}(\tilde{G}, \mu, \lambda)} = \lambda \left(1 + \mu \prod_{j=1}^{d_i} \frac{R^{\mu, \lambda}_{\tau_S, \sigma_i, \hat{\sigma}_{j}}(G_i, v_{i, j})}{1 + R^{\mu, \lambda}_{\tau_S, \sigma_i, \hat{\sigma}_j}(G_i, v_{i, j})}\right).
\end{align}

Following from Equation (\ref{eq:tree:v}), (\ref{eq:tree:e:=1}) and (\ref{eq:tree:e3}), we obtain that
\begin{align}\label{eq:recursion}
    R^{\mu, \lambda}_{\tau_S}(G, v) = \lambda \prod_{i=1}^{D'} \left(1 + \mu \prod_{j=1}^{d_i} \frac{R^{\mu, \lambda}_{\tau_S, \sigma_i, \hat{\sigma}_{j}}(G_i, v_{i, j})}{1 + R^{\mu, \lambda}_{\tau_S, \sigma_i, \hat{\sigma}_j}(G_i, v_{i, j})}\right),
\end{align}
where $D' = D - |\{ i: \exists j, \text{ s.t. } \tau_S(v_{i, j}) = 1 \}|$.

\section{Proof of Lemma \ref{lem:complex-contraction-zero-free}}

\begin{definition}
    We say a free vertex $v$ is blocked in a configuration $\tau_S \in \{0, 1\}^{S}$ if $\hat{\tau} \in \{0, 1\}^{S \cup \{v\}}$ is infeasible, where $\hat{\tau}(v) = 0$ and $\hat{\tau}(u) = \tau_S(u)$ for $u \in S$.
\end{definition}

\begin{remark}
    We point out that when $\mu \neq -1$, then any configuration is feasible, and there is no blocked vertex.
\end{remark}

Following lemma implies that any blocked vertex can be pinned by $1$ without changing any value of the partition functions. 
\begin{lemma}\label{lemma:block}
    For a blocked vertex $v$ with a feasible configuration $\tau_S$, 
    we have
    \begin{align*}
        Z_{\tau_S}(G, \mu, \eta) = Z_{\tau_S, \sigma(v)=1}(G, \mu, \eta), \\
        R_{\tau_S}^{\mu, \eta}(G, \tilde{v}) = R_{\tau_S, \sigma(v)=1}^{\mu, \eta}(G, \tilde{v}) ~~\text{for}~~ \forall \tilde{v} \in V.
    \end{align*}
\end{lemma}
\begin{proof}
    It is clearly that
    \begin{align*}
        &~~~~~~~~~Z_{\tau_S, \sigma(v_1)=0}(G, \mu, \eta) = 0, \\
        &Z_{\tau_S}(G, \mu, \eta) = Z_{\tau_S, \sigma(v_1)=1}(G, \mu, \eta). 
    \end{align*}
    Consequently, we have 
    \begin{align*}
        R_{\tau_S}^{\mu, \eta}(G, \tilde{v}) = \frac{Z_{\tau_S, \sigma(\tilde{v})=0}(G, \mu, \eta)}{Z_{\tau_S, \sigma(\tilde{v})=1}(G, \mu, \eta)} = \frac{Z_{\tau_S, \sigma(v)=1, \sigma(\tilde{v})=0}(G, \mu, \eta)}{Z_{\tau_S, \sigma(v)=1, \sigma(\tilde{v})=1}(G, \mu, \eta)} = R_{\tau_S, \sigma(v)=1}^{\mu, \eta}(G, \tilde{v}).
    \end{align*}
\end{proof}

Now we can assume that there is no blocked vertex in $G$. 
We prove Lemma \ref{lem:complex-contraction-zero-free} by following stronger result.
\begin{lemma}\label{lemma:complex-contraction-zero-free-feasible}
    For $\eta \neq 0$, if $\eta$ satisfies $\Delta$-complex-contraction property for weighted set covers with parameter $\mu \ge -1$, then for any hypergraph $G$ with max degree $\Delta$ and a feasible configuration $\tau_S$, we have 
    \begin{enumerate}
        \item $R^{\mu, \eta}_{\tau_S}(G, v) \in F_{\eta}$ for free vertex with degree at most $\Delta - 1$, 
        \item $R^{\mu, \eta}_{\tau_S}(G, v) \neq -1$ for free vertex with degree $\Delta$, 
        \item $Z_{\tau_S}(G, \mu, \eta) \neq 0$.
    \end{enumerate}
\end{lemma}

\begin{proof}
    We use induction on $t = |V| - |S|$ to prove the this lemma. 
    If $|V| = |S|$, that is $S = V$. Note that
    \begin{align*}
        Z_{\tau_S}(G, \mu, \eta) = \eta^{|\{v: \tau_S(v) = 1\}|} \prod_{e \in E} \left( 1 + \mu \prod_{v \in e} (1 - \tau_S(v)) \right) = \eta^{m_1} (1 + \mu)^{m_2},
    \end{align*}
    where $m_1 = |\{v: \tau_V(v) = 1\}|$ and $m_2 = |\{e: \forall v \in e, \tau_V(v) = 0 \}|$.
    If $\mu > -1$, then $Z_{\tau_S}(G, \mu, \eta) \neq 0$ apparently. And if $\mu = -1$, $Z_{\tau_S}(G, \mu, \eta) \neq 0$ follows from $m_2 = 0$ by feasibility of $\tau_S$. \\
    
    Now assume that the desired result holds for $|V| - |S| \le t$, let consider the case of $|V| - |S| = t + 1$. 
    
    Choose a free vertex $v \in V$. By induction hypothesis, we know that $Z_{\tau_S, \sigma(v) = 1}(G, \mu, \eta) \neq 0$ and $R_{\tau_S}^{\mu, \eta}(G, v)$ is well-defined. \\[0.1cm]
    Recall the definition of $\tilde{G}$ and $\sigma_i$.
    By Lemma \ref{lemma:block}, we can assume that $v$ is unblocked in $G$, and consequently the configuration $\tau_S \cup \sigma_i$ is feasible in $\tilde{G}$. Then we fix the blocked vertex in $\tilde{G}$ with configuration $\tau_S \cup \sigma_i$ to $1$.
    We still have 
    \begin{align*}
        R^{\mu, \eta}_{\tau_S}(G, v) = \eta \prod_{i=1}^{D'} \left(1 + \mu \prod_{j=1}^{d_i} \frac{R^{\mu, \eta}_{\tau_S, \sigma_i, \hat{\sigma}_{j}}(G_i, v_{i, j})}{1 + R^{\mu, \eta}_{\tau_S, \sigma_i, \hat{\sigma}_j}(G_i, v_{i, j})}\right),
    \end{align*}
    where $D' = D - |\{ i: \exists j, \text{ s.t. } \tau_S(v_{i, j}) = 1, \text {or } v_{i, j} \text{ is blocked in } \tau_S \cup \sigma_i \}|$, each $\tau_S \cup \sigma_i \cup \hat{\sigma}_j$ is feasible in $G_i$ with at most $t$ free vertex, and the degree of $v_{i,j}$ in $G_i$ is at most $\Delta - 1$.
    
    By induction hypothesis, $R^{\mu, \eta}_{\tau_S, \sigma_i, \hat{\sigma}_j}(G_i, v_{i, j}) \in F_{\eta}$. 
    Then based on $\Delta$-complex-contraction for $\eta$, if $D' \le \Delta - 1$ we have $R^{\mu, \eta}_{\tau_S}(G, v) \in F_{\eta}$, and if $D' = \Delta$ there holds $R^{\mu, \eta}_{\tau_S}(G, v) \neq -1$.
    Finally, note that $R^{\mu, \eta}_{\tau_S}(G, v) \neq -1$ is equivalent to $Z_{\tau_S}(G, \mu, \eta) \neq 0$.
    
\end{proof}

\begin{remark}
    In fact, there holds
    \begin{align*}
        f(\vz; \vm, d, \mu, \eta) = f(\vz; \vm, \Delta - 1, \mu, \eta),
    \end{align*}
    for $d \le \Delta - 1$ with setting $z_{\tilde{d}, 1} = 0$ for $\tilde{d} > d$.
\end{remark}

\section{Proof of Lemma \ref{lem:set-cover:W-k}}
\begin{proof}
    Recall that $W(k) = \{w: |w| \le e^{-|\arg(w)|/k}\}$ and $\phi(w) = \frac{w}{w + 1}$. 
    We remark here that $\phi$ is a bijection from $\hat{\sC}$ to $\hat{\sC}$. 
    
    Let $F = \phi^{-1}(W(k))$. 
    Then $\{0, \infty\} \subseteq F$ according to $\{0, 1\} \subseteq W(k)$. 
    Moreover due to $W(k)$ is bounded, thus $-1 \notin F$. 
    
    Let $\tilde{W} = \{\prod_{i=1}^{\Delta-1}(1 + \mu w_i): \vw \in W(k)^{\Delta-1}\}$. We know that $\frac{\eta \tilde{W}}{1 + \eta \tilde{W}} \subseteq \open{W(k)}$ for $\eta \in [0, \eta_0]$.
    
    Then by the approach we used in the proof of Lemma \ref{lem:bounded-2-spin:U} and compactness of $W(k)$, we know that there exists $\delta > 0$ such that for all $\eta \in \gU([0, \eta_0], \delta)$ we have
    \begin{align*}
        \frac{\eta \tilde{W}}{1 + \eta \tilde{W}} \subseteq W(k), 
    \end{align*}
    which implies that $\eta \tilde{W} \in F$, i.e., 
    \begin{align*}
        f(\vz; \vm, \Delta-1, \mu, \eta) \in F.
    \end{align*}
\end{proof}

\section{Proof of Lemma \ref{lemma:region1}}
Recall that $\tilde{W} = W(k)^{\circ} \backslash (-e^{-\pi/k} , 0]$. Let $g(z) = \ln(\phi^{-1}(z)) = \ln\left(\frac{z}{1 - z}\right)$. 

\begin{lemma}\label{lemma:boundary2}
    We have
    \begin{align*}
        \partial(g(\tilde{W})) = \left\{ g\bigg(\partial W(k) \backslash \{1, -e^{-\pi/k}\}\bigg) \right\} \cup \left\{ z: \re(z) \le - \ln (1 + e^{\pi/k}), |\im(z)| = \pi \right\}.
    \end{align*}
\end{lemma}
\begin{proof}

Define 
\begin{align*}
    \gA_1 =  \left\{ g\bigg(\partial W(k) \backslash \{1, -e^{-\pi/k}\}\bigg) \right\}, ~~~ \gA_2 = \left\{ \ln\left( \frac{\alpha}{1 + \alpha} \right) \pm i \pi: \alpha \in (0, e^{-\pi/k}]  \right\}.
\end{align*}

Consider $w_0 \in \partial(W(k)) \backslash \{1, -e^{-\pi/k}\}$ and a sequence $\{w_k \in \tilde{W}\}_{k \ge 1}$ such that $\lim_{k \to \infty} w_k = w_0$. By continuity of $g$ among the set $W(k) \backslash ((-\infty, 0] \cup \{1\})$, we know that $\lim_{k \to \infty} g(w_k) = g(w_0)$, i.e. $g(w_0) \in \overline{g(\tilde{W})}$.  According to $g$ is holomorphic on $\tilde{W}$ and open mapping theorem, we know that $g(w_0)$ cannot be an interior of $g(\tilde{W})$ and thus $g(w_0) \in \partial(g(\tilde{W}))$, that is $\gA_1 \subseteq \partial(g(\tilde{W}))$. 

Furthermore, for $w_0 \in [-e^{-\pi/k}, 0)$ and a sequence $\{w_k = r_k e^{i \theta_k} \in \tilde{W}\}_{k \ge 1}$ such that $\arg(w_k) \in (0, \pi)$ and $\lim_{k \to \infty} w_k = w_0$, we know that 
\begin{align*}
    \lim_{k \to \infty} r_k = |w_0|, ~~ \lim_{k \to \infty} \theta_k = \pi.
\end{align*}
Note that
\begin{align*}
    g(w_k) = \frac{r_k (\cos(\theta_k) - r_k + i \sin(\theta_k)) }{1 + r_k^2 - 2 r_k \cos(\theta_k)}
\end{align*}
and
\begin{align*}
    \lim_{k \to \infty} \re{g(w_k)} &= \lim_{k \to \infty} \frac{1}{2} \ln \left( \frac{r_k^2}{1 + r_k^2 - 2 r_k \cos(\theta_k)} \right) = \ln \left( \frac{-w_0}{1 - w_0} \right), \\
    \lim_{k \to \infty} \im{g(w_k)} &= \pi - \lim_{k \to \infty} \arctan \left( \frac{\sin(\theta_k)}{r_k - \cos(\theta_k)} \right) = \pi.
\end{align*}

Similarly, if the sequence $\{w_k = r_k e^{i \theta_k} \in \tilde{W}\}_{k \ge 1}$ satisfies $\arg(w_k) \in (-\pi, 0)$ and $\lim_{k \to \infty} w_k = w_0$, then it holds that
\begin{align*}
    \lim_{k \to \infty} r_k = - w_0, ~~ \lim_{k \to \infty} \theta_k = \pi,
\end{align*}
and 
\begin{align*}
    \lim_{k \to \infty} \re{g(w_k)} &= \ln \left( \frac{-w_0}{1 - w_0} \right), \\
    \lim_{k \to \infty} \im{g(w_k)} &= - \pi,
\end{align*}
which implies that $\gA_2 \subseteq \partial(g(\tilde{W}))$. \\[0.1cm]

On the other hand, for $z_0 \in \partial(g(\tilde{W}))$, we know that there exists $\{w_k \in \tilde{W}\}_{k \ge 1}$ such that $\lim_{k \to \infty} g(w_k) = z_0$. Recalling that $\tilde{W}$ is bounded, there exists $\{w_{k_m}\}_{m \ge 1}$ and $w_0$ such that $\lim_{m \to \infty} w_{k_m} = w_0$. 

Following from open mapping theorem, we know that $w_0 \notin \tilde{W}$, otherwise $\lim_{k \to \infty} g(w_k) = g(w_0)$ and $z_0 = g(w_0)$ is also an interior point of $g(\tilde{W})$.
That is $w_0 \in \partial (\tilde{W})$.

Note that 
\begin{align*}
    |g(w)| \ge |\re{g(w)}| = |\log(|\phi^{-1}(w)|)| = \left| \log\left( \frac{|w|}{|1 - w|} \right) \right|.
\end{align*}
Hence, if $w_0 = 0$ or $w_0 = 1$, then the sequence $\left\{\left| g(w_{k_m}) \right|\right\}$ is unbounded, which is impossible because of $\lim_{m \to \infty} |g(w_{k_m})| = |z_0|$. 
Therefore $w_0 \in \partial(W(k))\backslash\{1, -e^{-\pi/k}\}$ or $w_0 \in [-e^{-\pi/k}, 0)$, 
which means that
\begin{align*}
    \partial(g(\tilde{W})) \subseteq \gA_1 \cup \gA_2.
\end{align*}

\end{proof}

Recall that for $\theta \in (0, \pi)$ 
\begin{align*}
    r_{k, 1}(\theta) &\triangleq \re\left(\ln \left(\frac{e^{-\theta/k + i\theta}}{1 - e^{-\theta/k + i\theta}}\right) \right) = -\frac{1}{2} \ln \left( e^{2\theta/k} + 1 - 2 e^{\theta/k} \cos(\theta) \right), \\
    h_{k, 1}(\theta) &\triangleq \im\left(\ln \left(\frac{e^{-\theta/k + i\theta}}{1 - e^{-\theta/k + i\theta}}\right) \right) = \begin{cases}
    \arctan \left( \frac{\sin(\theta)}{\cos(\theta) - e^{-\theta/k}} \right), ~~~&\text{if } \cos(\theta) > e^{-\theta/k}, \\
    \frac{\pi}{2}, ~~~&\text{if } \cos(\theta) = e^{-\theta/k}, \\
    \pi - \arctan \left( \frac{\sin(\theta)}{e^{-\theta/k} - \cos(\theta)} \right), ~~~&\text{if } \cos(\theta) < e^{-\theta/k}.
    \end{cases}
\end{align*}

\begin{lemma}\label{lemma:ab-monotone}
    $r_{k,1}(\theta)$ is strictly decreasing on $(0, \pi)$, $h_{k,1}(\theta)$ is strictly increasing on $(0, \pi)$.
\end{lemma}
\begin{proof}
    Observe that
    \begin{align*}
        r'_{k, 1}(\theta) = -\frac{e^{\theta/k} \left(e^{\theta/k}+k \sin (\theta)-\cos (\theta)\right)}{k \left(e^{2\theta/k}-2 e^{\theta/k} \cos (\theta)+1\right)}.
    \end{align*}
    
    Let $\psi(\theta) = e^{\theta/k}+k \sin (\theta)-\cos (\theta)$. 
    For $\theta \in [\pi/2, \pi)$, $\psi(\theta) > 0$ holds apparently. For $\theta \in (0, \pi/2]$, we have
    \begin{align*}
        \psi'(\theta) = \frac{e^{\theta/k}}{k} + k \cos(\theta) + \sin(\theta) > 0, \\
        \psi(\theta) > \psi(0) = 0.
    \end{align*}
    Thus $r'_{k, 1}(\theta) < 0$ for $\theta \in (0, \pi)$.
    
    Next, observe that
    \begin{align*}
        h'_{k, 1}(\theta) = \frac{e^{\theta /k} \left(k e^{\theta /k} -\sin (\theta )-k \cos (\theta )\right)}{k \left(e^{\frac{2 \theta }{k}}-2 \cos (\theta ) e^{\theta /k}+1\right)}.
    \end{align*}
    Let $\tilde\psi(\theta) = k e^{\theta /k} -\sin (\theta )-k \cos (\theta )$. With $\tilde{\psi}'(\theta) = \psi(\theta) > 0$, we have $\tilde\psi(\theta) > \tilde\psi(0) = 0$, which implies that $h'_{k, 1}(\theta) > 0$ for $\theta \in (0, \pi)$.
    
\end{proof}

\begin{lemma}
    There holds
    \begin{align*}
        g(\tilde{W}) = \{z: \re(z) > - \ln (1 + e^{\pi/k}), |\im(z)| < h_{k,1}(r_{k,1}^{-1}(\re(z))) \} \cup \{z: \re(z) \le - \ln (1 + e^{\pi/k}), |\im(z)| < \pi \}.
    \end{align*}
\end{lemma}

\begin{proof}
    
    Note that $r_{k,1}(\pi) = - \ln (1 + e^{\pi/k})$ and $\lim_{\theta \to 0^{+}} r_{k, 1}(\theta) = +\infty$.
    Hence the function $r_{k, 1}$ has an inverse function $r_{k, 1}^{-1}(x)$ defined on $(- \ln (1 + e^{\pi/k}), +\infty)$.
    
    Moreover, we have
    \begin{align*}
        \{ g(\partial W(k) \backslash \{1, -e^{-\pi/k}\}) \} = \{z: \re(z) > - \ln (1 + e^{\pi/k}), |\im(z)| = h_{k, 1}(r_{k, 1}^{-1}(\re(z))) \}.
    \end{align*}
    
    Then following from Lemma \ref{lemma:boundary2}, we know that 
    \begin{align*}
         \partial(g(\tilde{W})) = &\{z: \re(z) > - \ln (1 + e^{\pi/k}), |\im(z)| = h_{k, 1}(r_{k, 1}^{-1}(\re(z))) \} \\
        & \cup \left\{ z: \re(z) \le - \ln (1 + e^{\pi/k}), |\im(z)| = \pi \right\}.
    \end{align*}
    
    By $\lim_{\theta \to \pi} h_{k, 1}(\theta) = \pi$, we know that the function 
    \begin{align*}
        p_k(x) = \begin{cases}
            h_{k, 1}(r_{k, 1}^{-1}(x)), ~~~&\text{if } x > - \ln (1 + e^{\pi/k}), \\
            \pi, ~~~&\text{if } x \le - \ln (1 + e^{\pi/k}),
        \end{cases}
    \end{align*}
    is continuous on $\sR$. And by Lemma \ref{lemma:ab-monotone}, $p_k(x)$ is a decreasing function. 
    Together with $1/2 \in W(k)$, we have $0 \in g(\tilde{W})$ and
    \begin{align*}
        g(\tilde{W}) = \{z: |\im(z)| < p_k(\re(z)) \}.
    \end{align*}
    
\end{proof}

\section{Proof of Lemma \ref{lemma:region2}}

    
    

\begin{lemma}\label{lemma:boundary}
    Assume that $\mu \neq 0, \mu \ge -1$ and $k < \frac{\pi}{2\ln(\mu + 4)}$. 
    Then the region $\ln(1 + \mu \open{W(k)})$ is open and
    \begin{align*}
        \partial(\ln(1 + \mu \open{W(k)})) = \ln(\partial(1 + \mu W(k)) \backslash \{0\}).
    \end{align*}
\end{lemma}

\begin{proof}
    Denote $W = 1 + \mu W(k)$.

    For $\mu > 0$, by $k < \frac{\pi}{\ln(\mu + 4)}$, we know that $e^{-\pi/k} < \frac{1}{\mu + 4}$, thus $1 - \mu e^{-\pi/k} > 1 - \frac{\mu}{\mu + 4} > 0$. 
    Following from $W \cap \sR = [1 - \mu e^{-\pi/k}, 1+ \mu]$, we have $(-\infty, 0] \cap W = \emptyset$.
    
    And, for $\mu \in [-1, 0)$, we have $W \cap \sR = [1 + \mu, 1 - \mu e^{-\pi/k}]$ and $(-\infty, 0] \cap \open{W} = \emptyset$.
    
    Therefore, the function $\ln(w)$ is holomorphic on $\open{W}$. By open mapping theorem, we know that $\ln(\open{W})$ is an open set. 
    
    Additionally, we point out that $\ln(w)$ is holomorphic on $W \backslash \{0\}$. 
    
    For $w_0 \in \partial W \backslash \{0\}$ and any sequence $\{w_k \in \open{W}\}_{k \ge 1}$ such that $\lim_{k \to \infty} w_k = w_0$, by continuity of $\ln(w)$ at $w_0$, there holds $\lim_{k \to \infty} \ln(w_k) = \ln(w_0)$. By bijection of $\ln(w)$, we know that $\ln(w_0) \in \partial(\ln(\open{W}))$, which implies that
    \begin{align*}
        \ln(\partial W \backslash \{0\}) \subseteq \partial(\ln(\open{W})).
    \end{align*}
    
    On the other hand, for $z_0 \in \partial(\ln(\open{W}))$, there exists a sequence $\{w_k \in \open{W}\}_{k \ge 1}$ such that $\lim_{k \to \infty} \ln(w_k) = z_0$. Recalling that $W$ is compact, there exists $\{w_{k_m}\}_{m \ge 1}$ and $w_0 \in W$ such that $\lim_{m \to \infty} w_{k_m} = w_0$. 
    
    Following from bjiection of $\ln(w)$, we know that $w_0 \in \partial (W)$.
    
    Note that 
    \begin{align*}
        |\ln(w)| \ge |\re{\ln(w)}| = |\ln(|w|)|.
    \end{align*}
    Hence, if $w_0 = 0$, then the sequence $\left\{\left| \ln(w_{k_m}) \right|\right\}$ is unbounded, which is impossible because of $\lim_{m \to \infty} |\ln(w_{k_m})| = |z_0|$. Then by continuity of $\ln(w)$, we know that $\ln(w_0) = z_0$, that is 
    \begin{align*}
        \partial(\ln(\open{W})) \subseteq \ln(\partial W \backslash \{0\}).
    \end{align*}
\end{proof}

Assume that $\mu \neq 0, \mu \ge -1$ and $k < \frac{\pi}{2\ln(\mu + 4)}$. 
By Lemma \ref{lemma:boundary}, we have
\[\partial(\ln(1 + \mu \open{W(k)})) = \ln((1 + \mu \partial(W(k))) \backslash \{0\}).\]

Note that $\partial (W(k)) = \{z = e^{-\theta/k + i \theta}: \theta \in [0, \pi]\} \cup \{z = e^{-\theta/k - i \theta}: \theta \in [0, \pi]\}$. 

Only when $\mu = -1$, there holds $0 \in 1 + \mu \partial(W(k))$, and we need to take $\theta \in (0, \pi]$. 

Recall the definition of $h_{k, 2}$ and $r_{k, 2}$:
\begin{align*}
    r_{k, 2}(\theta) &\triangleq \frac{1}{2} \ln \left( 1 + \mu^2 e^{-2\theta / k} + 2\mu e^{-\theta / k} \cos(\theta) \right), \\
    h_{k, 2}(\theta) &\triangleq \arctan \left( \frac{\mu e^{-\theta/k} \sin(\theta)}{1 + \mu e^{-\theta/k} \cos(\theta)} \right).
\end{align*}    
For simplicity, we write $r, h$ for $r_{k, 2}, h_{k, 2}$ in this section.

\begin{lemma}\label{lemma:r:monotone}
    For $k < \frac{\pi}{2\ln(\mu + 4)}$, the equation $k e^{\theta/k} \sin (\theta)+e^{\theta/k} \cos (\theta)+\mu = 0$ has an unique root $\theta_0$ in $(0, \pi]$.
    
    Moreover, for $\mu > 0$, $r(\theta)$ is strictly decreasing on the interval $(0, \theta_0)$ and strictly increasing on the interval $(\theta_0, \pi)$. And for $\mu \in [-1, 0)$, $r(\theta)$ is strictly increasing on the interval $(0, \theta_0)$ and strictly decreasing on the interval $(\theta_0, \pi)$.
\end{lemma}

\begin{proof}
    Observe that
    \begin{align*}
        r'(\theta) = -\frac{\mu  \left(k e^{\theta/k} \sin (\theta)+e^{\theta/k} \cos (\theta)+\mu\right)}{k \left(2 \mu  e^{\theta/k} \cos (\theta)+e^{2 \theta/k}+\mu^2 \right)}.
    \end{align*}
    
    Let $\psi(\theta) = k e^{\theta/k} \sin (\theta)+e^{\theta/k} \cos (\theta)+\mu$.
    Note that
    \[\psi'(\theta) = \frac{e^{\theta / k} \cos(\theta) (1 + k^2)}{k}.\]
    Thus $\psi$ is increasing over the interval $[0, \pi/2]$ and decreasing over the interval $[\pi/2, \pi]$. 
    Together with $\psi(0) = 1 + \mu \ge 0$, 
    \[\psi(\pi) = - e^{\pi/k} + \mu < -(\mu + 4) + \mu < 0,\] we know that $\psi(\theta)$ has an unique zero $\theta_0$ in the interval $(\pi/2, \pi)$. 
    Then we have $\psi(\theta) > 0$ for $\theta \in (0, \theta_0)$, and $\psi(\theta) < 0$ for $\theta \in (\theta_0, \pi)$. 
    Consequently, $r(\theta)$ is strictly monotone over $(0, \theta_0)$ and $(\theta_0, \pi)$ respectively.
\end{proof}

By monotonicity of $r$, there exists $\tilde{r}_1(x)$ defined on $[r(\theta_0), r(0)]$ such that $\tilde{r}_1(r(\theta)) = \theta$ for $\theta \in [0, \theta_0]$. And there also exists $\tilde{r}_2(x)$ defined on $[r(\theta_0), r(\pi)]$ such that $\tilde{r}_2(r(\theta)) = \theta$ for $\theta \in [\theta_0, \pi]$. 

\begin{lemma}
    We have
    \begin{align}\label{eq:d2}
        \frac{d^2 h(\tilde{r}_i(x))}{d x^2} = -\frac{\left(k^2+1\right) e^{\theta/k} \left(\mu ^2+2 \mu  e^{\theta/k} \cos (\theta)+e^{\frac{2 \theta}{k}}\right) \left(\mu  k \cos (\theta)+k e^{\theta/k}+\mu  \sin (\theta)\right)}{\mu  \left(k e^{\theta/k} \sin (\theta)+e^{\theta/k} \cos (\theta)+\mu \right)^3},
    \end{align}
    where $\theta = \tilde{r}_i(x)$ for $i = 1, 2$.
\end{lemma}

\begin{proof}
    Note that
    \begin{align*}
        \frac{d h(\tilde{r}_i(x))}{d x} = h'(\theta) \frac{d \tilde{h}_i(x)}{ d x} = \frac{h'(\theta)}{r'(\theta)}.
    \end{align*}
    
    Thus
    \begin{align*}
        \frac{d^2 h(\tilde{r}_i(x))}{d x^2} = \frac{d \left(\frac{h'(\tilde{r}_i(x))}{r'(\tilde{r}_i(x))}\right)}{d x} =  \left(\frac{h'(\theta)}{r'(\theta)}\right)' \frac{d \tilde{r}_i(x)}{ d x} = \left(\frac{h'(\theta)}{r'(\theta)}\right)' \frac{1}{r'(\theta)}.
    \end{align*}
\end{proof}

\begin{lemma}
    For $\mu > 0$, $h(\theta)$ is increasing on $(\theta_0, \pi)$. And for $\mu \in [-1, 0)$, $h(\theta)$ is decreasing on $(\theta_0, \pi)$.
\end{lemma}
\begin{proof}
    Note that
    \begin{align*}
        h'(\theta) = \frac{\mu  \left(-\sin (\theta ) e^{\theta /k}+k \cos (\theta ) e^{\theta /k}+k \mu \right)}{k \left(2 \mu  \cos (\theta ) e^{\theta /k}+e^{\frac{2 \theta }{k}}+\mu ^2\right)}.
    \end{align*}
    For $\theta \in (\theta_0, \pi)$, according to Lemma \ref{lemma:r:monotone}, we know that $\mu e^{-\theta/k} + k \sin(\theta) + \cos(\theta) < 0$. Hence there holds
    \begin{align*}
        k\mu e^{-\theta/k} < -k \cos(\theta) -k^2 \sin(\theta) < -k \cos(\theta) + \sin(\theta),
    \end{align*}
    which implies that $-\sin (\theta ) e^{\theta /k}+k \cos (\theta ) e^{\theta /k}+k \mu  < 0$. 
\end{proof}

\begin{lemma}\label{lemma:>0}
    For $\mu > 0$, with $\theta_0$ defined in Lemma \ref{lemma:r:monotone} and $\tilde{r}_1$ such that $\tilde{r}_1(r(\theta)) = \theta$ for $\theta \in [0, \theta_0]$, there holds
    \begin{align*}
        \ln(\open{(1 + \mu W(k))}) \subseteq \{z: \re(z) \in [r(\theta_0), r(0)], |\im(z)| < h(\tilde{r}_1(\re(z))) \}
    \end{align*}
    where the last set is convex.
\end{lemma}
\begin{proof}
    
    For $\mu > 0$, by lemma \ref{lemma:positive}, we know that $h(\theta) > 0$ for $\theta \in (0, \pi)$. 
    Note that $h(0) = h(\pi) = 0$, and $r(0) = \ln(1 + \mu) > 0$, $r(\pi) = \ln \left(1 - \mu e^{-\pi/k}\right) < 0$.
    
    Following from Lemma \ref{lemma:region2}, we know that
    \begin{align*}
        \partial(\ln(1 + \mu \open{W(k)})) &= \{z : \re(z) \in [r(\theta_0), r(0)], |\im(z)| = h(\tilde{r}_1(x))  \} \\
        &~~~ \cup \{z : \re(z) \in [r(\theta_0), r(\pi)], |\im(z)| = h(\tilde{r}_2(x)) \}.
    \end{align*}
    
    Let $\varphi_1(x) \triangleq h(\tilde{r}_1(x))$ and $\varphi_2(x) \triangleq h(\tilde{r}_2(x))$.
    Recalling the definition of $\theta_0$, we have $k e^{\theta/k} \sin (\theta)+e^{\theta/k} \cos (\theta)+\mu > 0$ for $\theta \in (0, \theta_0)$, and $k e^{\theta/k} \sin (\theta)+e^{\theta/k} \cos (\theta)+\mu > 0$ for $\theta \in (\theta_0, \pi)$.
    
    Then by Equation (\ref{eq:d2}), and Lemma \ref{lemma:positive} (ii), we have 
    \begin{align*}
        \varphi_1''(x) < 0, ~~\text{ for } x \in [r(\theta_0), r(0)], \\
        \varphi_2''(x) > 0, ~~\text{ for } x \in [r(\theta_0), r(\pi)],
    \end{align*}
    
    Hence, by convexity of $\varphi_2(x)$ on $[r(\theta_0), r(\pi)]$, we know that
    \begin{align*}
        \varphi_2(x) \le \frac{x - r(\pi)}{r(\theta_0) - r(\pi)} (\varphi_2(r(\theta_0)) - \varphi_2(r(\pi))) = \frac{x - r(\pi)}{r(\theta_0) - r(\pi)} h(\theta_0).
    \end{align*}
    
    Similarly, by concavity of $\varphi_1(x)$ for $x \in [r(\theta_0), r(0)]$, we have
    \begin{align*}
        \varphi_1(x) \ge \frac{x - r(0)}{r(\theta_0) - r(0)} h(\theta_0).
    \end{align*}
    
    According to $r(\theta_0) < r(\pi) < r(0) $, for $x \in [r(\theta_0), r(\pi)]$, there holds
    \begin{align*}
        \frac{x - r(\pi)}{r(\theta_0) - r(\pi)} < \frac{x - r(0)}{r(\theta_0) - r(\pi)},
    \end{align*}
    which implies that $\varphi_2(x) < \varphi_1(x)$ for $x \in [r(\theta_0), r(\pi)]$. 
    
    Consequently, we can conclude that
    \begin{align*}
        \ln(1 + \mu \open{W(k)}) &= \{z: \re(z) \in [r(\pi), r(0)], |\im(z)| < \varphi_1(\re(z)) \} \\
        &~~~ \cup \{z: \re(z) \in [r(\theta_0), r(\pi)], \varphi_2(\re(z)) < |\im(z)| < \varphi_1(\re(z)) \} \\
        &\subseteq \{z: \re(z) \in [r(\theta_0), r(0)], |\im(z)| < \varphi_1(\re(z)) \}. 
    \end{align*}
    
    At last, by concavity of $\varphi_1(x)$, the set
    \[\{ (x, y): x \in [r(\theta_0), r(0)], y < \varphi_1(x) \}\] 
    is convex. Then 
    by symmetry and convexity of intersection of convex sets, we know that the set $\{z: \re(z) \in [r(\theta_0), r(0)], |\im(z)| < \varphi_1(\re(z)) \}$ is convex.
    
\end{proof}

\begin{lemma}\label{lemma:<0}
    For $\mu \in (-1, 0)$, with $\theta_0$ defined in Lemma \ref{lemma:r:monotone} and $\tilde{r}_1$ such that $\tilde{r}_1(r(\theta)) = \theta$ for $\theta \in [0, \theta_0]$, there holds
    \begin{align*}
        \ln(1 + \mu \open{W(k)}) \subseteq \{z: \re(z) \in [r(0), r(\theta_0)], |\im(z)| < -h(\tilde{r}_1(\re(z))) \}
    \end{align*}
    where the last set is convex. 
    Moreover, for $\mu = -1$, we have
    \begin{align*}
        \ln(\open{(1 - W(k))}) \subseteq \{z: \re(z) \in (-\infty, r(\theta_0)], |\im(z)| < -h(\tilde{r}_1(\re(z))) \}
    \end{align*}
    where the last set is also convex.
\end{lemma}

The proof of this Lemma is similar to the proof of the Lemma \ref{lemma:>0}.

\section{Some Technical Lemma}
\begin{lemma}\label{lemma:positive}
    For $\mu \ge -1$, $k < \frac{\pi}{2\ln(\mu + 4)}$ and $\theta \in (0, \pi]$, there holds 
    \begin{enumerate}[label=(\roman*)]
        \item $1 + \mu e^{-\theta / k}  \cos (\theta) > 0$, 
        \item $k e^{\theta/k} + \mu  k \cos (\theta) + \mu  \sin (\theta) > 0$.
    \end{enumerate}
\end{lemma}

\begin{proof}
    Firstly by $k < \frac{\pi}{2\ln(\mu + 4)}$, we have
    \begin{align*}
        e^{\frac{\pi}{k}} > e^{\frac{\pi}{2k}} > \mu + 4.
    \end{align*}
    (i) If $|\mu| \le 1$, we have
    \begin{align*}
        e^{\theta / k} > 1 \ge |\mu \cos(\theta)|.
    \end{align*}
    And if $\mu > 1$, we have
    \begin{align*}
        e^{\theta / k} + \mu  \cos (\theta) \ge \begin{cases}
        1 , ~~&\text{if }  \theta \in [0, \pi/2], \\
        e^{\frac{\pi}{2k}} - \mu > 0, ~~&\text{if } \theta \in [\pi/2, \pi].
        \end{cases}
    \end{align*} \\[0.1cm]
    
    (ii) Let $g(\theta) = k e^{\theta/k} + \mu  k \cos (\theta) + \mu  \sin (\theta)$. 
    
    Note that
    \begin{align*}
        g'(\theta) &= e^{\theta/k} + \mu (\cos(\theta) - k \sin(\theta)), \\
        g''(\theta) &= e^{\theta/k}/k + \mu (-\sin(\theta) - k \cos(\theta)), \\
        g'''(\theta) &= e^{\theta/k}/k^2 + \mu (-\cos(\theta) + k \sin(\theta)).
    \end{align*}
    
    For $\ > 0$, if $\theta \in [0, \pi/2]$, then $g(\theta) > 0$ holds apparently. 
    And if $\theta \in [\pi/2, \pi]$, then we have
    \begin{align*}
        g'''(\theta) &> 0, \\
        g''(\theta) &\ge g''(\pi/2) = \frac{1}{k} e^{\frac{\pi}{2 k}} - \mu > (\mu + 4)/k - \mu > 0, \\
        g'(\theta) &\ge g'(\pi/2) = e^{\frac{\pi}{2 k}} - k \mu > (\mu + 4) - k \mu > 0, \\
        g(\theta) &\ge g(\pi/2) = k e^{\frac{\pi}{2 k}} + \mu > 0.
    \end{align*}
    
    On the other hand, for $\mu \in [-1, 0)$ and $\theta \in (0, \pi/2]$, we have
    \begin{align*}
        g''(\theta) &> 0, \\
        g'(\theta) &> g'(0) = 1 + \mu \ge 0, \\
        g(\theta) &> g(0) = k(1 + \mu) \ge 0.
    \end{align*}
    
    And for $\mu \in [-1, 0)$ and $\theta \in [\pi/2, \pi]$, it holds that
    \begin{align*}
        g'(\theta) &> 0, \\
        g(\theta) &\ge g(\pi/2) = k e^{\frac{\pi}{2 k}} + \mu.
    \end{align*}
    Let $\psi(k) \triangleq k e^{\frac{\pi}{2k}}$. By $\psi'(k) = e^{\frac{\pi}{2k}}(1 - \frac{\pi}{2 k})$, $\psi(k)$ is decreasing over $k \in (0, \pi/2)$ and increasing over $k \in (\pi/2, + \infty)$. Hence
    \begin{align*}
        g(\theta) \ge  k e^{\frac{\pi}{2 k}} + \mu \ge \frac{\pi e}{2} + \mu > 0.
    \end{align*}
\end{proof}

\begin{lemma}\label{lem:h-k-2}
    For $\mu = -1$, we have $|h_{k, 2}(\theta)| < \arctan(k)$.
\end{lemma}
\begin{proof}
    For $\mu = -1$, we have
    \begin{align*}
        h_{k, 2}(\theta) = -\arctan \left( \frac{e^{-\theta/k} \sin(\theta)}{1 - e^{-\theta/k} \cos(\theta)} \right) = -\arctan \left( \frac{ \sin(\theta)}{e^{\theta/k} - \cos(\theta)} \right).
    \end{align*}
    
    Just note that
    \begin{align*}
        \frac{ \sin(\theta)}{e^{\theta/k} - \cos(\theta)} \le \frac{\theta}{1 + \theta/k - \cos(\theta)} \le k.
    \end{align*}
\end{proof}

\begin{lemma}\label{lem:p-k-small}
    For $x = -\ln (e^{2\Delta} + 1)$ and sufficiently small $k > 0$, we have $p_k(x) > \Delta \arctan k$.
\end{lemma}
\begin{proof}
    For sufficiently small $k > 0$, we have
    \begin{align*}
        h_{k,1}(k \Delta) = \arcsin\left( \frac{\sin(k \Delta)}{ \sqrt{e^{-2\Delta} + 1 - 2 e^{-\Delta} \cos( k \Delta)}}\right) \ge k \Delta.
    \end{align*}
    And 
    \begin{align*}
        r_{k, 1}(k \Delta) = -\frac{1}{2} \ln \left( e^{2 \Delta} + 1 - 2 e^{\Delta} \cos(k \Delta) \right) > -\ln (1 + e^{2\Delta}).
    \end{align*}
    
    Note that $r_{k, 1}$ is decreasing on $(0, \pi)$, hence for $x = -\ln (e^{2\Delta} + 1)$, we have $r_{k, 1}^{-1}(x) > k \Delta$ and
    \begin{align*}
        p_k(x) = h_{k, 1}(r_{k, 1}^{-1}(x)) > h_{k,1}(k \Delta) \ge \Delta \arctan k.
    \end{align*}
\end{proof}

\begin{lemma}\label{lem:h-k-1}
    For $\theta \in (0, \pi)$, we have $h_{k, 1}(\theta) > \arctan(k)$.
\end{lemma}

\begin{proof}
    Note that
    \begin{align*}
        \lim_{\theta \to 0} \frac{\sin^2(\theta)}{1 + e^{-2\theta/k} - 2 e^{-\theta/k} \cos(\theta)} &= \lim_{\theta \to 0} \frac{2 k \sin(\theta) \cos(\theta)}{2 e^{-\theta/k} k \sin(\theta) + 2 e^{-\theta/k} \cos(\theta) - 2 e^{-2\theta/k}} \\
        &= \lim_{\theta \to 0} \frac{k \theta e^{\theta/k}}{k \sin(\theta) + \cos(\theta) - e^{-\theta/k}} \\
        &= \lim_{\theta \to 0} \frac{k e^{\theta/k} + \theta e^{\theta/k}}{k \cos(\theta) - \sin(\theta) + e^{-\theta/k}/k} \\
        &= \frac{k^2}{1 + k^2}.
    \end{align*}
    Then by lemma \ref{lemma:ab-monotone}, we know that 
    \begin{align*}
        h_{k, 1}(\theta) > \lim_{\theta \to 0} h_1(\theta) = \arcsin \left( \frac{k}{\sqrt{1 + k^2}} \right) = \arctan k.
    \end{align*}
\end{proof}

\begin{lemma}
    For $k < 1$ and $\theta \in (0, \pi/4]$, we have $\cos(\theta) > e^{-\theta/k}$.
\end{lemma}
\begin{proof}
    Let $\psi(\theta) = \cos(\theta) - e^{-\theta/k}$. 
    There holds
    \begin{align*}
        \psi'(\theta) &= -\sin(\theta) + e^{-\theta/k} / k, \\
        \psi''(\theta) &= -\cos(\theta) - e^{-\theta/k} / k^2 < 0.
    \end{align*}
    Thus $\psi(\theta) > \min\{\psi(0), \psi(\pi/4)\} = 0$.
\end{proof}

\section{Proof of Lemma \ref{lemma:compare}}

We prove Lemma \ref{lemma:compare} by following steps. Recall that we need to ensure 
\begin{align}\label{eq:compare:app}
    \forall \theta \in [0, \theta_0], ~\text{s.t.}~ 
    (\Delta - 1) |h_{k, 2} (\theta)| < p_k\left(\ln(\eta) + (\Delta - 1) r_{k, 2}(\theta)\right).
\end{align}


\vskip 8pt
\textbf{(1) The case where $\mu = -1$ and $\theta/k$ is sufficiently small } 

For $\mu = -1$, if $\theta / k \le \eps$, then
\begin{align*}
    r_{k, 2}(\theta) &= \frac{1}{2} \ln \left( 1 + e^{-2\theta / k} - 2 e^{-\theta / k} \cos(\theta) \right) \le \frac{1}{2} \ln \left( 1 + e^{-2\theta / k} - 2 e^{-\theta / k} (1 - \theta^2/2) \right) \\
    &= \frac{1}{2} \ln \left( (1 - e^{-\theta/k})^2 + e^{-\theta/k} \theta^2 \right) \le \frac{1}{2} \ln \left( \eps^2 + k^2 \eps^2 \right)
    \le \ln(2 \eps).
\end{align*}

Setting $\eta (2\eps)^{\Delta} = \frac{1}{e^{2\Delta} + 1}$, then we have
\begin{align*}
    p_k\left(\ln(\eta) + (\Delta - 1) r_{k, 2}(\theta)\right) > p_k(-\ln(e^{2\Delta} + 1)) > \Delta \arctan k > (\Delta - 1) |h_{k, 2}(\theta)|,
\end{align*}
where we have recalled Lemma \ref{lem:p-k-small} and Lemma \ref{lem:h-k-2}. 
Hence Equation (\ref{eq:compare:app}) holds in this case.

\vskip 8pt
\textbf{(2) The lower bounds of $r_{k, 2}(\theta)$ } 

If $\mu > 0$ and $k < \frac{\pi}{2 \ln(2 \mu)}$, then by Lemma \ref{lemma:r:monotone}, we have $\theta_0 \in (\pi/2, \pi)$, and $r_{k ,2}$ achieves its minimal value at $\theta_0$. Hence
\begin{align*}
    r_{k, 2}(\theta) \ge r_{k, 2}(\theta_0) \ge \frac{1}{2} \ln (1 + \mu^2 e^{-2\theta_0/k} - 2\mu e^{-\theta_0 / k}) = \ln(1 - \mu e^{-\frac{\theta_0}{ k}}) \ge \ln(1 - \mu e^{-\frac{\pi}{2 k}}) \ge -\ln(2).
\end{align*}

For $\mu \in (-1, 0)$, by Lemma \ref{lemma:r:monotone}, $r_{k ,2}$ achieves its minimal value at $0$, which means that
\begin{align*}
    r_{k, 2}(\theta) \ge r_{k, 2}(0) =  \ln(1 + \mu).
\end{align*}

And for $\mu = -1, \theta/k > \eps$, we have
\begin{align*}
    r_{k, 2}(\theta) &= \frac{1}{2} \ln \left( 1 + e^{-2\theta / k} - 2 e^{-\theta / k} \cos(\theta) \right) \ge \ln \left( 1 - e^{-\theta / k}\right) \ge \ln \left( 1 - e^{-\eps}\right).
\end{align*}

Therefore, there exists $A$ such that $r_{k, 2}(\theta) \ge A$.
And hence, for sufficiently small $k > 0$, it holds that
\begin{align*}
    \ln(\eta) + (\Delta - 1) r_{k, 2}(\theta) > -\ln(1 + e^{\pi/k}).
\end{align*}
 
Therefore, we just need to ensure that
\begin{align}\label{eq:compare:app:2}
    \forall \theta \in [0, \theta_0], ~\text{s.t.}~ 
    (\Delta - 1) |h_{k, 2} (\theta)| < h_{k,1}\left(r_{k, 1}^{-1}\left(\ln(\eta) + (\Delta - 1) r_{k, 2}(\theta)\right)\right),
\end{align}
according to $p_k(x) = h_{k,1}(g^{-1}_{k,1}(x))$ for $x > -\ln(1 + e^{\pi/k})$.

\vskip 8pt
\textbf{(3) The case where $\theta/k$ is sufficiently large}

For $\theta > \frac{\pi}{2}$ and sufficiently small $k > 0$, we have
\begin{align*}
    |h_{k, 2}(\theta)| &= \arctan \left( \frac{|\mu| e^{-\theta/k} \sin(\theta)}{1 + \mu e^{-\theta/k} \cos(\theta)} \right) 
    \le \arctan \left( \frac{|\mu| e^{-\theta/k}}{1 - |\mu| e^{-\theta/k}} \right) \\
    &\le \arctan \left( \frac{|\mu| e^{-\pi/(2k)}}{1 - |\mu| e^{-\pi/(2k)}} \right) 
    \le \frac{\arctan k}{\Delta -1}.
\end{align*}

And for $\mu > 0$, choose $M$ such that $\arctan(\mu M e^{-M} k) < \frac{\arctan k}{\Delta -1}$. Then for $ \theta \in (M k, \pi/2)$ and $ k < \frac{\pi}{4 M}$, we have
\begin{align*}
    |h_{k, 2}(\theta)| &= \arctan \left( \frac{\mu e^{-\theta/k} \sin(\theta)}{1 + \mu e^{-\theta/k} \cos(\theta)} \right) \\
    &\le \arctan \left( \mu e^{-\theta/k} \sin(\theta) \right).
\end{align*}
Let $\psi(\theta) = e^{-\theta/k} \sin(\theta)$. We claim that $\psi(\theta)$ is decreasing on $(M k, \pi/2)$. 
In fact, note that
\begin{align*}
    \psi'(\theta) = e^{-\theta/k} (\cos(\theta) - \sin(\theta)/k), 
\end{align*}
and 
\begin{align*}
    \cos(\theta) - \sin(\theta)/k < \cos(M k) - \sin(M k) / k = \frac{\cos(M k)}{k} (k - \tan(M k)) < \cos(M k)(1 - M) < 0,
\end{align*} 
where we have used that $\tan x \ge x$ for $x \in (0, \pi/2)$. Therefore, we have $\psi(\theta) < e^{-M} \sin(M k) < M e^{-M} k$, that is
\begin{align*}
    |h_{k, 2}(\theta)| < \arctan(\mu M e^{-M} k) < \frac{\arctan k}{\Delta -1}.
\end{align*}

For $\mu \in [-1, 0)$, choose $M$ such that $\arctan(\frac{M k}{e^{M} - 1}) < \frac{\arctan k}{\Delta -1}$. Then for $ \theta \in (M k, \pi/2)$ and $ k < \frac{\pi}{4 M}$, we have
    \begin{align*}
        |h_{k, 2}(\theta)| &= \arctan \left( \frac{-\mu e^{-\theta/k} \sin(\theta)}{1 + \mu e^{-\theta/k} \cos(\theta)} \right) \\
        &\le \arctan \left( \frac{ e^{-\theta/k} \sin(\theta)}{1 - e^{-\theta/k}} \right) \\
        &\le \arctan \left( \frac{ M k}{e^{M} - 1} \right) 
        < \frac{\arctan k}{\Delta -1}.
    \end{align*}

Together with $h_{k, 1}(\theta) > \arctan k$ by Lemma \ref{lem:h-k-1}, Equation (\ref{eq:compare:app:2}) holds for $\theta > M k$ and sufficiently small $k > 0$.

Now we just need to ensure that for sufficiently small $k > 0$, there holds
\begin{align}\label{eq:compare:app:3}
    \forall x \in [0, M], ~\text{s.t.}~ 
    (\Delta - 1) |h_{k, 2} (k x)| < h_{k,1}\left(r_{k, 1}^{-1}\left(\ln(\eta) + (\Delta - 1) r_{k, 2}(k x)\right)\right).
\end{align}
We remark that for $\mu = -1$, the interval $[0, M]$ should be $[\eps, M]$. 

\vskip 8pt
\textbf{(4) Eliminate $k$ through a limit operation}

Define a function $S(k, x)$ to be
\begin{align*}
    S(k, x) = h_{k,1}\left(r_{k, 1}^{-1}\left(\ln(\eta) + (\Delta - 1) r_{k, 2}(k x)\right)\right) - (\Delta - 1) |h_{k, 2} (k x)|.
\end{align*}


Letting $x_1(k, x) = r_{k, 1}^{-1}\left(\ln(\eta) + (\Delta - 1) r_{k, 2}(k x)\right)/k$, we have
\begin{align*}
\eta^2 (\mu^2 e^{-2x} + 2 \mu e^{-x} \cos(k x) + 1)^{\Delta - 1} = \frac{1}{e^{2 x_1(k, x)} - 2e^{x_1(k, x)} \cos(k x_1(k, x)) + 1}.
\end{align*}

It is easily to check that
\begin{align*}
    \lim_{k \to 0^{+}} x_1(k, x) &= \ln \left( \frac{1}{\eta(\mu e^{-x} + 1)^{\Delta - 1}} + 1 \right) \triangleq \hat{x}_1(x), \\
    \lim_{k \to 0^{+}} \frac{\partial x_1(k, x)}{\partial k} &= 0.
\end{align*}

Hence for $x \in [0, M]$ (or $x \in [\eps, M]$ if $\mu = -1$), we have 
\begin{align*}
    \lim_{k \to 0^{+}} \cos(k x_1(k, x)) - e^{-x_1(k, x)} &= 1 - e^{-\hat{x}_1(x)} > 0, \\
    \lim_{k \to 0^{+}} 1 + \mu e^{-x} \cos(k x) &= 1 + \mu e^{-x} \ge \begin{cases}
        1 + \mu > 0, ~~~&\text{ for } \mu > -1, \\
        1 - e^{-\eps} > 0, ~~~&\text{ for } \mu = -1.
    \end{cases}
\end{align*}

Consequently there holds
\begin{align*}
    \lim_{k \to 0^{+}} h_{k, 1}(k x_1(k, x)) &= \lim_{k \to 0^{+}} \arctan\left( \frac{\sin(k x_1(k, x))}{\cos(k x_1(k, x)) - e^{-x_1(k, x)}} \right) = 0, \\
    \lim_{k \to 0^{+}} h_{k, 2}(k x) &= \lim_{k \to 0^{+}} \arctan\left(\frac{\mu e^{-x} \sin(k x)}{1 + \mu e^{-x} \cos(k x)}\right) = 0,
\end{align*}
which implies that $\lim_{k \to 0^{+}} S(k, x) = 0$ for $x \in [0, M]$ (or $x \in [\eps, M]$ if $\mu = -1$). 

Similar to Lemma \ref{lem:bounded-2-spin:H}, Condition (\ref{eq:compare:app:3}) can be deduced from
\begin{align}\label{eq:compare:app:4}
    \forall x \in [0, M], H(x) \triangleq \lim_{k \to 0^{+}} \frac{\partial S(k, x)}{\partial k} > 0.
\end{align}

Note that
\begin{align*}
    \lim_{k \to 0^{+}} \frac{\partial h_{k, 1}(k x_1(k, x))}{\partial k} &= \lim_{k \to 0^{+}} \frac{\partial}{\partial k} \arctan\left( \frac{\sin(k x_1(k, x))}{\cos(k x_1(k, x)) - e^{-x_1(k, x)}} \right) \\
    &= \lim_{k \to 0^{+}} \frac{e^{x_1(k, x)} x_1(k, x) (e^{x_1(k, x)} - \cos(k x_1(k, x)))}{1 + e^{2 x_1(k, x)} - 2 e^{x_1(k, x)} \cos(k x_1(k, x))} \\
    &= \frac{\hat{x}_1(x) e^{\hat{x}_1(x)} }{e^{\hat{x}_1(x)} - 1},
\end{align*}
and
\begin{align*}
    \lim_{k \to 0^{+}} \frac{\partial h_{k, 2}(k x)}{\partial k} &= \lim_{k \to 0^{+}} \frac{\partial}{\partial k} \arctan\left( \frac{ \mu e^{-x} \sin (k x)}{1 + \mu e^{-x} \cos (k x)} \right) \\
    &= \frac{\mu x}{ e^{x} + \mu}.
\end{align*}

Therefore, we have
\begin{align*}
    H(x) = \frac{\hat{x}_1(x) e^{\hat{x}_1(x)} }{e^{\hat{x}_1(x)} - 1} - \frac{(\Delta - 1) |\mu| x}{e^x + \mu}. 
\end{align*}

\vskip 8pt
\textbf{(5) The case $\mu > 0$}

Let $D = \Delta - 1$.
For $\mu > $, we need to verify that
\begin{align*}
    H(x) &= \frac{\hat{x}_1(x) e^{\hat{x}_1(x)} }{e^{\hat{x}_1(x)} - 1} - \frac{D \mu x}{e^x + \mu} \\
    &= \frac{\ln \left( \frac{1}{\eta(\mu e^{-x} + 1)^D} + 1 \right) \left( \frac{1}{\eta(\mu e^{-x} + 1)^D} + 1 \right)}{\frac{1}{\eta(\mu e^{-x} + 1)^D}} - \frac{D \mu x}{e^x + \mu} \\
    &= (1 + \eta(\mu e^{-x} + 1)^D) \ln \left( \frac{1}{\eta(\mu e^{-x} + 1)^D} + 1 \right)  - \frac{D \mu x}{e^x + \mu} > 0.
\end{align*}

Let
\begin{align*}
    \tilde{G}_{\mu}(x) = \frac{\mu x}{\mu + e^x}.
\end{align*}
Then $\tilde{G}_{\mu}'(x) = \mu \frac{(1 - x) e^x + \mu}{(e^x + \mu)^2}$. It is clearly that there exists a unique solution $\hat{x} \in [0, +\infty)$ to the equation $\mu e^{-x} = x - 1$.
    
For $\mu \le \frac{e^{1 + \frac{1}{D}}}{D}$, we have $e^{\hat{x}}(\hat{x} - 1) = \mu \le \frac{e^{1 + \frac{1}{D}}}{D}$ which means that $\hat{x} \le 1 + \frac{1}{D}$. And hence
\begin{align*}
    \frac{D \mu x}{e^x + \mu} \le \frac{D \mu \hat{x}}{e^{\hat{x}} + \mu} = D(\hat{x} - 1) \le 1.
\end{align*}
Together with $\ln(1 + y) > \frac{y}{1 + y}$ for $y > 0$, we know that $H(x) > 0$.\\[0.1cm]
    
    Now we suppose that $\mu > \frac{e^{1 + \frac{1}{D}}}{D}$.
    Then we have $\hat{x} > 1 + \frac{1}{D}$. 
    Consider
    \begin{align*}
        H'(x) = \frac{D \mu \left(e^x x-\left(\mu+e^x\right) \eta \left(\mu e^{-x}+1\right)^D \ln \left(\eta^{-1} \left(\mu e^{-x}+1\right)^{-D}+1\right)\right)}{\left(\mu+e^x\right){}^2}.
    \end{align*}
    
    Let $H'(x) = 0$, that is
    \begin{align*}
        \frac{x}{\eta (\mu e^{-x} + 1)^{D+1}} = \ln \left(\frac{1} {\eta \left(\mu e^{-x}+1\right){}^{D}} + 1\right), \\
        \frac{x}{\mu e^{-x} + 1} = \eta (\mu e^{-x} + 1)^{D} \ln \left(\frac{1} {\eta\left(\mu e^{-x}+1\right){}^{D}} + 1\right).
    \end{align*}
    
    According to $\psi(y) = \frac{\ln(1 + y)}{y}$ is deceasing in $y$, RHS is decreasing in $x$. Together with LHS is increasing in $x$, we know that 
    \begin{align*}
        F(x) \triangleq \psi\left( \frac{1} {\eta \left(\mu e^{-x}+1\right){}^{D}} \right) - \frac{x}{\mu e^{-x} + 1}
    \end{align*}
    is decreasing in $x$.
    
    Note that
    \begin{align*}
        F(0) &= \psi\left( \frac{1} {\eta \left(\mu +1\right){}^{D}} \right) > 0, \\
        F(\hat{x}) &= \psi\left( \frac{1} {\eta\left(\mu e^{-\hat{x}} +1\right){}^{D}} \right) - \frac{\hat{x}}{\mu e^{-\hat{x}} + 1} \\
        &= \psi\left( \frac{1} {\eta \hat{x}^{D}} \right) - 1 < 0,
    \end{align*}
    where we have recalled that $\mu e^{-\hat{x}} = \hat{x} - 1$. 
    
    Therefore, there exists $\check{x} \in (0, \hat{x})$ such that $F(\check{x}) = 0$. 
    And we need to check that $H(x) \ge H(\check{x}) > 0$.
    
    \begin{claim}
        For  $\mu > \frac{e^{1 + \frac{1}{D}}}{D}$, we have $D \mu e^{-\check{x}} - 1 > 0$.
    \end{claim}
    \begin{proof}
    If $D \mu e^{-\check{x}} - 1 \le 0$, then we have 
    \begin{align*}
        \mu e^{-\check{x}} \le \frac{1}{D}, ~~~\text{ and }~ \\
        \check{x} \ge \ln(D \mu) = \ln(D (\hat{x} - 1) e^{\hat{x}} ) = \hat{x} + \ln(D (\hat{x} - 1)) > \hat{x},
    \end{align*}
    which is impossible.
    \end{proof}
    
    \begin{claim}
        For  $\mu > \frac{e^{1 + \frac{1}{D}}}{D}$, we have $\frac{\ln(D \mu)}{D+1} < \hat{x} - 1$.
    \end{claim}
    \begin{proof}
    Note that
    \begin{align*}
        \frac{\ln(D \mu)}{D+1} &< \hat{x} - 1, \\
        \iff \ln(D e^{\hat{x}}(\hat{x} - 1)) &< (D + 1)(\hat{x} - 1), \\
        \iff \ln(D(\hat{x} - 1)) &< D(\hat{x} - 1) - 1.
    \end{align*}
    \end{proof}

    Note that
    \begin{align*}
        & H(\check{x}) = \frac{\check{x} (1 + \eta (\mu e^{-\check{x}} + 1)^D)}{\eta (\mu e^{-\check{x}} + 1)^{D + 1}} - \frac{D \check{x} \mu e^{-\check{x}}}{\mu e^{-\check{x}} + 1} > 0, \\
        \iff & 1 > \eta (D \mu e^{-\check{x}} - 1) (\mu e^{-\check{x}} + 1)^D, \\
        \iff & \psi\left( \frac{1}{\eta (\mu e^{-\check{x}} + 1)^D} \right) <  \psi(D \mu e^{-\check{x}} - 1) , \\
        \iff & \frac{\check{x}}{\mu e^{-\check{x}} + 1} < \psi(D\mu e^{-\check{x}} - 1) 
        = \frac{\ln(D\mu e^{-\check{x}})}{D\mu e^{-\check{x}} - 1}, \\
        \iff & \frac{\ln(D \mu)}{D + 1} > \frac{ \mu \check{x} }{\mu + e^{\check{x}} }.
    \end{align*}
    
    Note that the function $\tilde{G}_{\mu}(x) = \frac{\mu x}{\mu + e^x}$ is increasing over $x \in [0, \hat{x}]$, and
    \begin{align*}
        \tilde{G}(\hat{x}) = \hat{x} - 1.
    \end{align*}
    
    Thus $H(\check{x}) > 0$ is equivalent to 
    \begin{align*}
        \check{x} < \tilde{G}^{-1}_{\mu} \left( \frac{\ln(D \mu)}{D + 1} \right).
    \end{align*}

    At last, we obtain our condition as
    \begin{align*}
        F\left(\tilde{G}^{-1}_{\mu} \left( \frac{\ln(D \mu)}{D + 1} \right)\right) < 0,
    \end{align*}
    or equivalently, 
    \begin{align*}
        1/\eta > G^{-1}\left( \frac{x_0}{\mu e^{-x_0} + 1}  \right) (\mu e^{-x_0} + 1)^D,
    \end{align*}
    where $x_0$ is the smallest solution to the equation 
    \begin{align*}
        \frac{\mu x}{\mu + e^x} = \frac{\ln(D \mu)}{D + 1},
    \end{align*}
    and $G(y) = \frac{\ln(1 + y)}{y}$.
    
\vskip 8pt
\textbf{(6) The case $\mu < 0$}

    Let $D = \Delta - 1$.
    We need to ensure that
    \begin{align*}
        H(x) &\triangleq \frac{\hat{x}_1(x) e^{\hat{x}_1(x)} }{e^{\hat{x}_1(x)} - 1} + \frac{D \mu x}{e^x + \mu} \\
        &= \frac{\ln \left( \frac{1}{\eta(\mu e^{-x} + 1)^D} + 1 \right) \left( \frac{1}{\eta(\mu e^{-x} + 1)^D} + 1 \right)}{\frac{1}{\eta(\mu e^{-x} + 1)^D}} + \frac{D \mu x}{e^x + \mu} \\
        &= (1 + \eta(\mu e^{-x} + 1)^D) \ln \left( \frac{1}{\eta(\mu e^{-x} + 1)^D} + 1 \right)  + \frac{D \mu x}{e^x + \mu} > 0.
    \end{align*}
    
    Let
    $
        \tilde{G}_{\mu}(x) = \frac{\mu x}{\mu + e^x}.
    $
    Then we have $\tilde{G}_{\mu}'(x) = \mu \frac{(1 - x) e^x + \mu}{(e^x + \mu)^2}$. 
    
    Note that 
    \begin{align*}
        \tilde\psi(x) &= (1 - x) e^x + \mu, \\
        \tilde\psi'(x) &= -x e^{x} < 0, \\ 
        \tilde\psi(0) &= 1 + \mu \ge 0, ~\tilde\psi(1) = \mu < 0. 
    \end{align*}
    Denote the unique solution to equation $\mu e^{-x} = x - 1$ to be $\hat{x}$. 
    For $-\mu \le \frac{e^{1 - \frac{1}{D}}}{D}$, we have 
    \begin{align*}
        e^{\hat{x}} (1 - \hat{x}) \le \frac{e^{1 - \frac{1}{D}}}{D}, \\
        \hat{x} \ge 1 - \frac{1}{D}.
    \end{align*}
    
    Consequently, there holds
    \begin{align*}
        \frac{D \mu x}{e^x + \mu} \ge \frac{D \mu \hat{x}}{e^{\hat{x}} + \mu} = -D(1 - \hat{x}) \ge -1.
    \end{align*}
    
    Together with $\ln(1 + y) > \frac{y}{y + 1}$ for $y > 0$, we know that $H(x) > 0$.\\[0.1cm]
    
    Now we suppose that $-\mu > \frac{e^{1 - \frac{1}{D}}}{D}$.
    Consider
    \begin{footnotesize}\begin{align*}
        H'(x) = -\frac{D \mu \left(\mu \left(\eta \left(e^{-x} \mu+1\right){}^D \log \left(\frac{\left(e^{-x} \mu+1\right){}^{-D}}{\eta}+1\right)-2\right)+e^x \left(\eta \left(e^{-x} \mu+1\right){}^D \log \left(\frac{\left(e^{-x} \mu+1\right){}^{-D}}{\eta}+1\right)+x-2\right)\right)}{\left(\mu+e^x\right){}^2}.
    \end{align*}\end{footnotesize}
    
    Let $H'(x) = 0$, that is 
    \begin{align*}
        \eta \left( \mu e^{-x} +1\right){}^D \log \left(\frac{\left(\mu e^{-x} +1\right){}^{-D}}{\eta}+1\right) = 2 - \frac{x}{1 + \mu e^{-x}}.
    \end{align*}
    
    Let $\psi(y) = \frac{\ln(1 + y)}{y}$, and
    \begin{align*}
        F(x) \triangleq \psi\left( \frac{1} {\eta \left(\mu e^{-x}+1\right){}^{D}} \right) + \frac{x}{\mu e^{-x} + 1} - 2.
    \end{align*}
    Then $\psi(y) \le 1$ for $y \ge 0$ and $F$ is an increasing function.
    
    Note that
    \begin{align*}
        F(2) &\ge \frac{2}{1 + \mu e^{-2}} - 2 > 0, \\
        F(\hat{x}) &= \psi\left( \frac{1} {\eta\left(\mu e^{-\hat{x}} +1\right){}^{D}} \right) + \frac{\hat{x}}{\mu e^{-\hat{x}} + 1} - 2 \\
        &= \psi\left( \frac{1} {\eta \hat{x}^{D}} \right) - 1 < 0,
    \end{align*}
    where we have recalled that $\mu e^{-\hat{x}} = \hat{x} - 1$. 
    
    Therefore, there exists $\check{x} \in (\hat{x}, 2)$ such that $F(\check{x}) = 0$. 
    

    Note that
    \begin{align*}
        & H(\check{x}) = \frac{\left(2 - \frac{\check{x}}{1 + \mu e^{-\check{x}}}\right) (1 + \eta (\mu e^{-\check{x}} + 1)^D)}{\eta (\mu e^{-\check{x}} + 1)^{D}} + \frac{D \check{x} \mu e^{-\check{x}}}{\mu e^{-\check{x}} + 1} > 0, \\
        \iff & \frac{2 - \frac{\check{x}}{1 + \mu e^{-\check{x}}}}{\eta (\mu e^{-\check{x}} + 1)^{D}} > -\frac{D \check{x} \mu e^{-\check{x}}}{\mu e^{-\check{x}} + 1} -2 + \frac{\check{x}}{1 + \mu e^{-\check{x}}}, \\
        \iff & \frac{1}{\eta (\mu e^{-\check{x}} + 1)^{D}} > \frac{\check{x} - 2 - 2 \mu e^{-\check{x}} - D \check{x} \mu e^{-\check{x}}}{2 - \check{x} + 2 \mu e^{-\check{x}}}, \\
        \iff & 2 - \frac{\check{x}}{1 + \mu e^{-\check{x}}} = \psi\left(\frac{1}{\eta (\mu e^{-\check{x}} + 1)^{D}}\right) < \psi\left(\frac{\check{x} - 2 - 2 \mu e^{-\check{x}} - D \check{x} \mu e^{-\check{x}}}{2 - \check{x} + 2 \mu e^{-\check{x}}}\right), \\
        \iff & 2 - \frac{\check{x}}{1 + \mu e^{-\check{x}}} < \ln\left(\frac{- D \check{x} \mu e^{-\check{x}}}{2 - \check{x} + 2 \mu e^{-\check{x}}}\right) \frac{2 - \check{x} + 2 \mu e^{-\check{x}}}{\check{x} - 2 - 2 \mu e^{-\check{x}} - D \check{x} \mu e^{-\check{x}}}, \\
        \iff & 1 < \ln\left(\frac{- D \check{x} \mu e^{-\check{x}}}{2 - \check{x} + 2 \mu e^{-\check{x}}}\right) \frac{1 + \mu e^{-\check{x}}}{\check{x} - 2 - 2 \mu e^{-\check{x}} - D \check{x} \mu e^{-\check{x}}}, \\ 
        \iff & \check{x} > x_0, \\
        \iff & F(x_0) < 0, \\
        \iff & \psi\left(\frac{1}{\eta (\mu e^{-x_0} + 1)^D} \right) < 2 - \frac{x_0}{\mu e^{-x_0} + 1}, \\
        \iff & \frac{1}{\eta} > \psi^{-1}\left( 2 - \frac{x_0}{\mu e^{-x_0} + 1} \right) (\mu e^{-x_0} + 1)^D,
    \end{align*}
    where $x_0$ is the biggest solution to the solution 
    \begin{align*}
        \ln\left(\frac{- (\Delta - 1) x \mu e^{-x}}{2 - x + 2 \mu e^{-x}}\right) \frac{1 + \mu e^{-x}}{x - 2 - 2 \mu e^{-x} - (\Delta - 1) x \mu e^{-x}} = 1.
    \end{align*}

\newpage
\bibliographystyle{plainnat}
\bibliography{reference}

\end{document}